\pdfoutput=1
\documentclass[11pt,letterpaper]{article}
\usepackage[T1]{fontenc}
\usepackage[utf8]{inputenc} 
\usepackage[letterpaper, left=1in, right=1in, top=0.9in, bottom=0.9in]{geometry}
\usepackage[american]{babel}
\usepackage[normalem]{ulem}
\usepackage{amsmath, amssymb, cases, amsthm}
\usepackage{thmtools}
\usepackage[shortlabels]{enumitem}
\usepackage{mdframed}
\usepackage{bbm}
\usepackage{bm}
\usepackage{microtype}
\usepackage{xcolor}
\usepackage{makecell}
\usepackage{mathtools}
\usepackage{algorithmic}
\usepackage[procnumbered,ruled,vlined,linesnumbered]{algorithm2e}
\usepackage{float}
\usepackage{varwidth}
\usepackage{modletter}
\usepackage{tcolorbox}
\usepackage{empheq}

\newtcolorbox{construction}[2][]
{
	colframe = gray!50,
	colback  = gray!10,
	coltitle = gray!10!black,
	left*=0mm, 
	before skip = 10pt,
	after skip = 10pt,
	title    = \textbf{\space\space #2},
	#1,
}

\SetKwInput{KwData}{Input}
\SetKwInput{KwResult}{Output}
\SetKwInput{KwGlobalVar}{Global variables}

\usepackage[bookmarks,colorlinks,breaklinks]{hyperref}
\hypersetup{urlcolor=blue, colorlinks=true, citecolor=green!50!black, linkcolor=blue}
\usepackage[capitalize,nosort,nameinlink]{cleveref}

\declaretheorem[numberwithin=section,refname={Theorem,Theorems},Refname={Theorem,Theorems}]{theorem}

\declaretheorem[numberlike=theorem]{lemma}

\declaretheorem[numberlike=theorem]{corollary}
\declaretheorem[numberlike=theorem]{definition}
\declaretheorem[numberlike=theorem]{claim}

\theoremstyle{definition}

\def\final{1}  
\ifnum\final=0  
\newcommand{\todo}[1]{{\color{red}[{\tiny TODO: \bf #1}]\marginpar{\color{red}*}}}
\newcommand{\yonggang}[1]{{\color{blue}[{\tiny Yonggang: \bf #1}]\marginpar{\color{blue}*}}}
\newcommand{\thatchaphol}[1]{{\color{purple}[{\tiny Thatchaphol: \bf #1}]\marginpar{\color{purple}*}}}
\else 
\newcommand{\thatchaphol}[1]{}
\newcommand{\yonggang}[1]{}
\newcommand{\todo}[1]{}
\fi  

\newcommand{\eps}{\varepsilon}
\newcommand{\reals}{\mathbb{R}}

\newcommand{\poly}{\mathrm{poly}}
\newcommand{\polylog}{\mathrm{polylog}}
\newcommand{\dist}{\mathrm{dist}}

\newcommand{\res}{\mathrm{res}}
\renewcommand{\neg}{\mathrm{neg}}

\newcommand{\vf}{v^{\mathrm{first}}}
\newcommand{\vl}{v^{\mathrm{last}}}
\newcommand{\first}[1]{#1^{\mathrm{first}}}
\newcommand{\last}[1]{#1^{\mathrm{last}}}

\newcommand{\diam}{\operatorname{diam}}
\newcommand{\spa}{\operatorname{spars}}
\newcommand{\sep}{\operatorname{sep}}

\newcommand{\leng}{\operatorname{leng}}
\newcommand{\conge}{\operatorname{cong}}

\newcommand{\ball}{\operatorname{ball}}
\newcommand{\vol}{\operatorname{vol}}

\newcommand{\degi}{\mathrm{deg}^{\mathrm{in}}}
\newcommand{\dego}{\mathrm{deg}^{\mathrm{out}}}

\newcommand{\PDMCF}{{approximate PD-MCF}}

\newcommand{\OPDMCF}[1]{\cO^\mathrm{PDMCF}_{#1}}

\global\long\def\poly{\mathrm{poly}}%
\global\long\def\vol{\mathrm{vol}}%
\global\long\def\eps{\epsilon}%
\global\long\def\reals{\mathbb{R}}%
\global\long\def\cN{{\cal N}}%
\global\long\def\cS{{\cal S}}%
\global\long\def\cT{{\cal T}}%
\global\long\def\cost{\mathrm{cost}}%
\global\long\def\sep{\mathrm{sep}}%
\global\long\def\val{\mathrm{val}}%
\global\long\def\dist{\mathrm{dist}}%

\global\long\def\deg{\mathrm{deg}}%
\global\long\def\diam{\mathrm{diam}}%
\global\long\def\cov{\mathrm{cov}}%
\global\long\def\ball{\mathrm{ball}}%
\global\long\def\supp{\mathrm{supp}}%

\newcommand{\tl}{\tilde{\ell}}
\newcommand{\Proj}{flow-cost-preserving DAG projection}
\global\long\def\Dem{\mathrm{Dem}}
\global\long\def\cong{\mathrm{cong}}
\global\long\def\path{\mathrm{path}}
\global\long\def\vvalue{\mathrm{value}}

\begin{document}
\sloppy

\title{
Combinatorial Minimum Cost Flow in Almost-Linear Time \\
on Dense Graphs
}

\author{Bernhard Haeupler\thanks{
        INSAIT, Sofia University ``St.~Kliment Ohridski'' and ETH Zürich,
        \texttt{bernhard.haeupler@insait.ai}.
        Partially funded by the Ministry of Education and Science of Bulgaria's support for INSAIT as part of the Bulgarian National Roadmap for Research Infrastructure and through the European Research Council (ERC) under the European Union's Horizon 2020 research and innovation program (ERC grant agreement 949272).} \and 
Yonggang Jiang\thanks{MPI-INF and Saarland University, Germany, \texttt{yjiang@mpi-inf.mpg.de}. Part of this work was done while visiting INSAIT. Supported by Google PhD fellowship.} \and   
Thatchaphol Saranurak\thanks{
        University of Michigan,
        \texttt{thsa@umich.edu}.
        Supported by NSF Grant CCF-2238138 and a Sloan Fellowship. Part of this work was done at INSAIT. Partially funded by the Ministry of Education and Science of Bulgaria's support for INSAIT as part of the Bulgarian National Roadmap for Research Infrastructure. }
}
\date{}
\maketitle
\pagenumbering{gobble}
\begin{abstract}
We give a combinatorial algorithm that computes an exact minimum-cost flow in $n^{2+o(1)}$ time. This is the first combinatorial improvement since Goldberg--Tarjan~\cite{goldberg1987solving}, and is almost optimal in dense graphs. All previous faster algorithms, including the almost-linear-time algorithm of Chen et al.~\cite{ChenKLPGS22}, rely on interior-point methods.

Our main conceptual contribution is a black-box reduction from exact minimum-cost flow on general graphs to its \emph{approximation} on DAGs, with only an $n^{o(1)}$ loss on the running time. This reduction allows future work on combinatorial minimum-cost flow to focus on DAGs, typically much easier to handle than general graphs. Our $n^{2+o(1)}$-time algorithm follows by combining this reduction with a simple augmenting-path-based approximation algorithm on DAGs.

The reduction is based on a new structural result, called a \emph{flow-cost-preserving DAG projection}. Given an arbitrary directed graph, this construction produces a DAG that approximately preserves both the congestion of flows and their cost up to a small loss.
\end{abstract}

\newpage

\tableofcontents
\newpage
\pagenumbering{arabic}
\section{Introduction}

Minimum-cost flow (MCF) is one of the most fundamental problems in combinatorial optimization. Given a directed graph $G=(V,E)$ with $n$ vertices and $m$ edges, capacities $U\in\mathbb{Z}_{>0}^{E}$, costs $\ell\in\mathbb{Z}^{E}$, and a demand vector $D\in\mathbb{R}^{V}$ satisfying $\sum_{v\in V} D(v)=0$, the goal is to find a flow $f\in\mathbb{R}^{E}$ that respects capacities, routes $D$, and minimizes the total cost $\sum_{e\in E}\ell(e)f(e)$. MCF captures shortest paths, transshipment, bipartite matching, maximum flow, and many other basic primitives.

Fast algorithms for MCF have been studied since the 1950s; see, e.g., \cite{dantzig1951application,fulkerson1961out,klein1967primal,edmonds1972theoretical,tardos1985strongly,orlin1988faster}. These classical algorithms share a common combinatorial principle: they repeatedly improve the current solution by augmenting flow along a negative cycle or a shortest augmenting path in the \emph{residual graph}. This line culminated in push-relabel-based algorithms \cite{goldberg1987solving,goldberg1990finding,ahuja1992finding} with $\tilde{O}(mn)$ running time.\footnote{We use $\tilde{O}$ to hide $\polylog(n,U_{\max},\ell_{\max})$ factors, where $U_{\max}=\max_{e}U(e)$ and $\ell_{\max}=\max_{e}|\ell(e)|$.}

Later work improved this long-standing $\tilde{O}(mn)$ bound via interior-point methods (IPMs) \cite{daitch2008faster,lee2019solving,van2021minimum,axiotis2022faster} and culminated in almost-linear $m^{1+o(1)}$-time algorithms \cite{ChenKLPGS22,van2024almost}. These algorithms achieve spectacular running times, but they do so by bypassing the classical directed residual graph: their progress is described through barrier functions and the current position in the polytope, and the key subroutines are on \emph{undirected} graphs.

This leaves a striking gap in our understanding of minimum-cost flow. The fastest exact algorithms are almost-linear, but the residual-graph-based approach for exact MCF appeared stuck at the classical $\tilde{O}(mn)$ bound. This raises a natural question: can one beat the classical $\tilde{O}(mn)$ bound for exact minimum-cost flow while staying within the residual-graph paradigm?

Very recent work showed that this barrier can be broken for closely related special cases while staying entirely within the directed residual graph. For bipartite matching \cite{chuzhoy2024maximum,chuzhoy2024faster} and maximum flow \cite{BernsteinBST24,BernsteinBLST25}, new residual-graph-based algorithms achieved $\tilde{O}(n^2)$ running time, which is near-optimal in dense graphs. Besides the running-time improvement, working with residual graphs leads to new primitives native to directed graphs. For example, \cite{BernsteinBST24} gives a simple and teachable approximate maximum flow algorithm on DAGs. Directed expander hierarchies introduced in \cite{BernsteinBST24,BernsteinBLST25} found applications to directed graph sparsifiers \cite{hoppenworth2025near}, global minimum cut and arborescence packing \cite{jiang2025approximating}, and black-box reductions for maximum flow to approximation on DAGs \cite{haeupler2026dag}. Unfortunately, these recent algorithms do not handle edge costs.

Thus, for exact minimum-cost flow, the fastest residual-graph-based algorithm has remained the $\tilde{O}(mn)$-time algorithm of Goldberg and Tarjan \cite{goldberg1987solving} for nearly forty years.

\paragraph{Our result: the first combinatorial improvement since Goldberg--Tarjan.}

We answer the question above affirmatively: within the residual-graph paradigm, we give the first improvement for exact minimum-cost flow since \cite{goldberg1987solving}.

\begin{restatable}{theorem}{main}\label{thm:main}
    There is a combinatorial algorithm computing a minimum-cost flow in $n^{2+o(1)}$ time.
\end{restatable}

In particular, \Cref{thm:main} is almost optimal in dense graphs. 
Our algorithm is randomized and correct with high probability.

\paragraph{Our approach: exact MCF via weak approximation on DAGs.}

The conceptual core of the paper is a reduction from exact MCF on general directed graphs to a very weak approximate primal-dual version of MCF on DAGs. Let $f$ be a flow and $(\phi,\tilde\ell)$ be a dual solution. We say that $(f,\phi,\tilde\ell)$ is a \emph{$(\kappa,\lambda)$-approximate PD-MCF solution} for demand $D$ if $f$ routes $D$ with congestion at most $\kappa$, and the cost of $f$ is at most $\lambda$ times the value of the dual solution $(\phi,\tilde\ell)$; see \Cref{subsec:MCF} for the formal definition.

We prove \Cref{thm:main} in three steps.

\begin{enumerate}
    \item \textbf{Exact-to-Approximate Reduction.}
    We reduce the problem of computing an exact MCF solution to $n^{o(1)}$ calls to an oracle for computing $(n^{o(1)},1+1/n^{o(1)})$-approximate PD-MCF. Thus, exactness can be bootstrapped from a surprisingly weak approximate solution. This reduction is much flexible than that of the successive approximation framework in \cite{goldberg1990finding}.

    \item \textbf{Flow-Cost-Preserving DAG Projection.}
    We further reduce $(n^{o(1)},1+1/n^{o(1)})$-approximate PD-MCF on general directed graphs to the same problem on DAGs. Informally, a flow-cost-preserving DAG projection replaces a general directed graph by a DAG while preserving both the congestion and the cost of flows up to a small loss. This DAG generalizes the recent DAG projections in \cite{haeupler2026dag} that preserve only congestion or only costs, and is our key structural contribution.

    \item \textbf{Weighted Push-Relabel for MCF on DAGs.}
    We give a simple augmenting-path-based algorithm for computing approximate PD-MCF on DAGs in $n^{2+o(1)}$ time, generalizing a similar weighted push-relabel algorithm in \cite{BernsteinBST24} for the maximum flow problem, a special case of the minimum-cost flow problem.
\end{enumerate}

Combining these three ingredients yields \Cref{thm:main}. More importantly, the first two ingredients already give a black-box reduction from exact MCF on general graphs to approximate MCF on DAGs.

\begin{theorem}
\label{thm:main-reduction}
There is an algorithm computing a minimum-cost flow in time complexity proportional to making oracle calls for computing $(n^{o(1)},1+1/n^{o(1)})$-approximate PD-MCF \emph{on DAGs} with cumulative total size $m^{1+o(1)}$ and cumulative total number of vertices $n^{1+o(1)}$.
\end{theorem}
\Cref{thm:main-reduction} is the conceptual core of our framework, and is useful beyond merely implying our main result, \Cref{thm:main}.

\begin{enumerate}
    \item First, it isolates the main difficulty of exact MCF on general directed graphs into weak approximation on DAGs. Note that a very weak $n^{o(1)}$ approximation on congestion is allowed.
    Thus, future work on \emph{combinatorial} minimum-cost flow can focus on the substantially simpler setting of approximation on DAGs.

    \item Second, it is quite plausible to strengthen \Cref{thm:main-reduction} to work in the deterministic and/or parallel setting. This is because our two main ingredients, DAG projections and directed expander hierarchies, are naturally compatible with these settings, as shown in \cite{abs-2511-08551,haeupler2026dag}. Although we do not pursue these directions here, this suggests a promising route toward deterministic and parallel algorithms for minimum-cost flow, again by reducing attention to DAGs.

\end{enumerate}

\section{Overview}

In this section, we sketch the three ingredients behind our $n^{2+o(1)}$ combinatorial algorithm for minimum-cost flow:
\begin{enumerate}
    \item in \Cref{subsec:overviewreduction}, we reduce exact MCF to the approximate primal-dual version of MCF;
    \item in \Cref{subsec:overviewprojection}, we reduce this approximate problem on general directed graphs to the same problem on DAGs via a \emph{flow-cost-preserving DAG projection};
    \item in \Cref{subsec:overviewdagalgo}, we give a simple augmenting-path-based algorithm for approximate PD-MCF on DAGs.
\end{enumerate}

For a directed graph $G=(V,E)$ and an edge $e\in E$, we write $\ell_G(e)$ for its length and $U_G(e)$ for its capacity. For a vertex $v\in V$, we write  $\deg_G(v)\stackrel{\mathrm{def}}{=}\sum_{v\in e} U_G(e)$
for its capacity-weighted degree. When the graph is clear from context, we omit the subscript $G$. For a flow $f$, we write its value $\val(f)$ as the amount of flow it sends, and $\cost(f)$ as the cost of the flow defined as $\sum_{e\in E}f(e)\cdot \ell(e)$, i.e., we call $\ell(e)$ both the \emph{length} of an edge and also the \emph{cost} of the edge.\footnote{Although mixing cost and length might be confusing, if one considers minimum cost flow as a combination of shortest path and maximum flow, then the notation of `length' makes sense.}

The demands $D$ in this section are all (single-commodity) demands, where each vertex gets a value: positive means flow should leave this node, and negative means flow should enter this node. A flow routing a demand is defined in the natural way. We write $|D|=\sum_{v\in V,D(v)>0}D(v)$.

We say that an algorithm makes $g$-efficient oracle calls to an oracle $\cO$ if:
\begin{itemize}
    \item the total number of vertices over all instances given to $\cO$ is at most $g\cdot n$,
    \item the total number of edges over all instances given to $\cO$ is at most $g\cdot m$, and
    \item the additional running time outside the oracle calls is at most $g\cdot m$.
\end{itemize}

\subsection{Exact-to-Approximation Reduction}\label{subsec:overviewreduction}

Throughout this section, all MCF instances are assumed to have non-negative edge lengths. For \emph{exact} MCF, this is without loss of generality: if an edge $(u,v)$ has negative length, we may reverse it to $(v,u)$, negate its length, and modify the demands by adding $U(u,v)$ to $v$ and $-U(u,v)$ to $u$. This folklore reduction transforms the instance into one with non-negative lengths. For multiplicative approximation problems, non-negative lengths are the natural setting.

We first explain how to reduce exact MCF to the approximate primal-dual version of MCF. The full proof is given in \Cref{sec:boosting}.

\begin{restatable}{theorem}{boosting}\label{thm:boosting}
    For $\eps=o(1/\log^9 n)$, there is an algorithm solving exact minimum-cost flow using
    \[
    \tO{n^{O(1/\log(1/\eps))}\cdot \kappa^2/\eps^3}
    \]
    -efficient calls to a $(\kappa,1+\eps)$-approximate PD-MCF oracle.
\end{restatable}

By setting $\kappa=n^{o(1)}$ and $\eps=1/n^{o(1)}$, \Cref{thm:boosting} implies that exact MCF can be solved using only $n^{o(1)}$ calls to a $(\kappa,1+\eps)$-approximate PD-MCF oracle.

\paragraph{Comparison with successive approximation.}
The closest prior exact-to-approximate framework for minimum-cost flow is the classical
successive-approximation / cost-scaling framework of Goldberg and Tarjan
\cite{goldberg1990finding}. However, that framework is not a
black-box reduction to an arbitrary approximate MCF oracle. Instead, given that the most negative cost is $-\eps$, the oracle in this framework must make the most negative reduced-cost edge less negative, from $-\eps$ to $-\eps/2$, over $O(\log n)$ phases.

This is not suitable for our purposes, since our oracle does not provide any direct control
over the most negative cost edge. In contrast, \Cref{thm:boosting} gives a genuinely
black-box reduction from exact minimum-cost flow to a much weaker approximate
primal-dual oracle. To the best of our knowledge, this is the first black-box reduction of this flavor for
minimum-cost flow. In particular, it can be viewed as the minimum-cost-flow analogue of
well-known exact-to-approximate boosting algorithms such as
shortest paths \cite{RozhonHMGZ23} and maximum flow (repeat in residual graphs). Moreover, our reduction unifies several natural
approximate formulations of minimum-cost flow by showing that they are equivalent for
the purpose of computing an exact solution.

\paragraph{Two intermediate problems.}
To prove \Cref{thm:boosting}, we introduce two intermediate problems.
\begin{description}
    \item[$(\kappa,1+\eps)$-strongly-approximate PD-MCF.]  
    Given a demand $D$, output a $(1\pm\eps)$-approximate primal-dual pair of the MCF LP $(f,\phi,\tilde\ell)$, where the primal flow routes a sub-demand $D'\preceq D$ with $|D'|\ge |D|/\kappa$ and has no congestion. The MCF LP is standard and defined formally in \Cref{subsec:MCF}. \footnote{This notion is stronger than $(\kappa,1+\eps)$-approximate PD-MCF, since by enlarging the primal flow by a factor of $\kappa$, we get a flow with congestion at most $\kappa$ and cost at most $(1+\eps)$ times the optimal, as required by the $(\kappa,1+\eps)$-approximate PD-MCF. The opposite direction generally does not hold.}

    \item[$(1-\eps)$-approximate PD-MBCF.]  
    Given a cost budget $B$, output a $(1\pm\eps)$-approximate primal-dual pair of the maximum bounded-cost flow (MBCF) LP. The MBCF LP can be found in \Cref{subsec:flowboosting}. It outputs the approximate maximum flow given a bounded cost.

\end{description}

The proof of \Cref{thm:boosting} is obtained by combining the following three reductions.

\paragraph{Step 1: exact PD-MCF to $(\kappa,1+\eps)$-strongly-approximate PD-MCF.}
We begin with the warm-up case $\kappa=1$, namely an oracle that routes \emph{all} of the demand and returns an approximate primal-dual pair for the MCF LP. Suppose the oracle returns a flow $f$ together with a dual pair $(\phi,\tilde\ell)$. We define a residual instance
\[
(G_{\res},D_{\res})=(G,D)_{f,\phi}
\]
in the following way.

Let $G_f$ be the residual graph of $G$ with respect to $f$. Recall that reversed edge of $(u,v)$ has length $-\ell(u,v)$. Let $(G_f)_\phi$ be the graph obtained by reweighting residual edges using the potential $\phi$, i.e., for every residual edge $(u,v)$,
\[
\ell_{(G_f)_\phi}(u,v)=\ell_G(u,v)-(\phi(v)-\phi(u)).
\]
Let $E_{\neg}$ be the set of negative-length edges in $(G_f)_\phi$. We define $G_{\res}$ by reversing every edge in $E_{\neg}$ and negating its length, so that all edge lengths in $G_{\res}$ are non-negative. The demand $D_{\res}$ is chosen so that for each edge $(u,v)$ getting reversed to $(v,u)$, we put $U(u,v)$ amount of source demand on $u$, and sink demand on $v$.

We first show a simple lemma saying that that solving this residual instance exactly lets us repair the current approximate solution.

\begin{lemma}\label{lem:overviewresidual}
    Suppose $f,\phi,\tilde\ell$ is a feasible primal-dual pair for the minimum-cost flow problem $(G,D)$, and $f^*_{\res}$ is a minimum-cost flow of $(G,D)_{f,\phi}$. Then $f+f^*_{\res}$ is a minimum-cost flow of $(G,D)$.
\end{lemma}
\begin{proof}
    Since $f^*_{\res}$ is a minimum-cost flow of $(G,D)_{f,\phi}$, the residual graph of $f+f^*_{\res}$ contains no negative cycle. Moreover, $f+f^*_{\res}$ is a feasible flow routing $D$. Hence it is optimal.
\end{proof}

The second key lemma is that if $(f,\phi,\tilde\ell)$ is a good approximate primal-dual pair, then the residual instance has small optimum.

\begin{lemma}\label{lem:overviewsmallresidual}
    Suppose $f,\phi,\tilde\ell$ is a feasible primal-dual pair of $(G,D)$ with primal-dual value difference $\Delta$. Then the optimal cost of the minimum-cost flow instance $(G,D)_{f,\phi}$ is at most $\Delta$.
\end{lemma}

The proof is by a careful calculation of the
sum over all negative residual edges, and comparing it with the primal-dual difference. The full calculation appears in \Cref{lem:smallresidual}.

Thus, if the oracle returns a $(1\pm \eps)$-approximate primal-dual pair with, say, $\eps=0.1$, then the residual instance has only a constant fraction of the original optimum cost. Repeating this for $O(\log n)$ rounds and applying \Cref{lem:overviewresidual} yields an exact algorithm.

We next extend this to general $\kappa$. Now the oracle may route only a sub-demand $D'\preceq D$ with $|D'|\ge |D|/\kappa$. We call the oracle repeatedly. In the $i$-th round it returns a flow $f_i$ and dual $(\phi_i,\tilde\ell_i)$ on the remaining demand. After each round we form the corresponding residual graph, reweight by $\phi_i$, and simply round every negative edge length up to $0$ before making the next call (because we assume the oracle can only work on graphs with non-negative lengths). Since each call routes at least a $1/\kappa$ fraction of the remaining demand, after $O(\kappa\log n)$ calls the total routed flow already covers all of $D$.

The subtle point is to show that these repeated roundings do not destroy the reduction. Intuitively, each round contributes only an $O(\eps)$ fraction of its own routed cost to the final residual error, so the total accumulated rounding loss remains small. This is formalized in \Cref{lem:ApxflowApxcostBoosting}.

\paragraph{Step 2: $(\kappa,1+\eps)$-strongly-approximate PD-MCF from $(1-\eps)$-approximate PD-MBCF.}
We now explain how to use an approximate bounded-cost-flow oracle to obtain a strongly-approximate PD-MCF solution.

Fix a demand $D$ and, for simplicity, assume it is an $s$--$t$ demand of total value $1$. Suppose we are given an oracle for maximum bounded-cost flow: for any budget $B$, it returns an approximately maximum $s$--$t$ flow of cost at most $B$, together with an approximate dual certificate.

A first idea is to perform binary searches over the budget $B_i=(1+\eps)^i$ and find the first budget at which the oracle routes a constant fraction of the demand, say at least $1/2$. However, this fails if the cost-value tradeoff has a sharp jump: it may happen that routing value $\mu$ costs $B$, while routing value $(1+\eps)\mu$ already costs $100B$. In that case, the first budget that certifies a value of at least $\mu$ is not necessarily close to the true cost of routing $\mu$. The bad case is when the budget grows a lot but the flow value grows only slightly.

Our way around this is to look for a \emph{good scale} where the routed value grows noticeably while the budget only increases by a $(1+\alpha)$ factor. We query the MBCF oracle on a geometric sequence of budgets
\[
B_i=(1+\alpha)^{i-2},
\qquad \text{where } \alpha=\eps^{1/3},
\]
and let $q_i$ be the amount of flow returned at budget $B_i$. The “good point” we need is: the budget only changes from $B_a$ to $B_{a+1}=(1+\alpha)B_a$, but the routed flow value jumps by at least a multiplicative $(1+\beta)$ factor. We will set $\alpha\gg\beta$, so that if all the points are not good, then the cost will increase by too much (controlled by $(1+\alpha)$) until reaching the trivial upper bound $\sum_{e\in E}\ell(e)\cdot U(e)$, which we assume to be $\poly(n)$, while the flow value only increase by tiny (not exceeding $1$). This is a contradiction since we can definitely route $1$ unit of flow with some cost below $\sum_{e\in E}\ell(e)\cdot U(e)$. So there must be a good point, and we can also use binary search to find that.
The detailed proof is given in \Cref{lem:weak-to-strong-via-cost-scaling}.
\paragraph{Step 3: $(1-\eps)$-approximate PD-MBCF from $(\kappa,1+\eps)$-approximate PD-MCF.}
The final step is based on the flow-boosting template of \cite[Theorem 9.2]{HaeuplerH0RS24}, which uses multiplicative weight updates (MWU) to reduce bounded-cost flow to repeated calls to an approximate PD-MCF oracle.

The main extra work on our side is to recover a suitable dual solution for the bounded-cost-flow LP. The dual involves the concept of \emph{moving cut}, which serves as a continuous extension of the normal cut in a graph. The goal of the dual is to block every path from the sources to sinks using a small moving cut. 

The MWU template already constructs a moving cut that blocks a large fraction of the relevant paths, but not necessarily all of them. We show that the remaining uncovered part is small: it accounts for only an $O(\eps)$ fraction of the optimum. Therefore, one additional $(\kappa,1+\eps)$-approximate PD-MCF call (which actually gives a $O(\kappa)$-approximate optimal moving cut) suffices to cover the remaining part, at a cost that can be absorbed into the approximation factor by choosing $\eps$ sufficiently small (much smaller than $1/\kappa$). This yields the desired primal-dual pair for bounded-cost flow. The detailed proof can be found in \Cref{subsec:flowboosting}.

\subsection{Flow-Cost-Preserving DAG Projection}\label{subsec:overviewprojection}

We next explain how to reduce approximate PD-MCF on a general directed graph to approximate PD-MCF on a DAG. This is the main structural contribution of the paper.

For convenience, write $\OPDMCF{\kappa,\lambda}$ for an oracle that computes a $(\kappa,\lambda)$-approximate PD-MCF solution. A \emph{flow-cost-preserving DAG projection} of a directed graph $G$ is defined as follows.

\begin{definition}\label{def:DAGprojection}
    A \Proj{} of $G = (V,E)$ is a DAG $G'$ along with a projection map $\pi: V(G')\to V$. We call it $(\kappa,\lambda)$-\Proj{}, or say it has congestion slack $\kappa$ and length slack $\lambda$ if
\begin{enumerate}
\item For every demand $D'$ routable in $G'$ by a feasible flow with cost $B$, $\pi(D')$ is routable in $G$ with congestion $\kappa$ and cost $B$ and,
\item For every demand $D$ routable in $G$ by a feasible flow with cost $B$, $\hat\pi^{-1}(D)$ is routable in $G'$ by a feasible flow with cost at most $\lambda\cdot B$. 
\end{enumerate}
\end{definition}

$\pi(D')$ is defined in the natural way: $\pi(D')(v)=\sum_{\pi(v')=v}D'(v')$. $\hat\pi^{-1}(D)$ is a bit trickier: we fix a topological order of $G'$, if $D(v)>0$, then we choose the first vertex $v'$ in the topological order with $\pi(v')=v$ and let $\hat\pi^{-1}(D)(v')=D(v)$; if $D(v)<0$, then we choose the last vertex $v'$ in the topological order with $\pi(v')=v$ and let $\hat\pi^{-1}(D)(v')=D(v)$; otherwise we let $\hat\pi^{-1}(D)(v')=0$. i.e., we lift the demand to be a demand on the DAG from the first copy of sources to the last copy of sinks.

The following lemma shows the power of DAG projection.

\begin{lemma}[Informal version of \Cref{lem:MCFtoDAG}]\label{lem:overviewMCFtoDAG}
    Let $G$ be a directed graph, and suppose $G'$ is a $(\kappa',\lambda')$-flow-cost-preserving DAG projection of $G$. Then one call to $\OPDMCF{\kappa,\lambda}$ on $G'$ can be transformed into a solution to $\OPDMCF{\kappa\kappa',\lambda\lambda'}$ on $G$.
\end{lemma}

The formal statement also additionally requires a projection algorithm that efficiently maps primal and dual objects between $G$ and $G'$. For the overview, the key intuition is simple: if $G'$ approximately preserves all relevant flows and their costs, then solving PD-MCF on $G'$ is essentially as good as solving it on $G$.

We now sketch how such a DAG projection is constructed.

\subsubsection{Existence}

The construction relies on two primitives: a length-constrained expander hierarchy, and a bounded-width path cover.

\paragraph{Primitive 1: LC expander hierarchy.}
A \emph{moving cut} $C$ assigns a non-negative value to every edge. Each moving cut is associated with a length parameter denoted by $\leng(C)$. When we view $C$ as an $h$-length moving cut, we denote graph $G-C$ by the graph increasing the length of each edge $e$ in $G$ by $C(e)\cdot h$. The following definitions rigorously define LC expander hierarchy, which is similar to normal expander hierarchy using the language of moving cut: instead of cuting edges in an integer way, each layer of this hierarchy using moving cut (length increasing) to represent a cut.

For such a cut $C$, define  $\vol_C(v)\stackrel{\mathrm{def}}{=}\sum_{v\in e} C(e)\cdot U(e)$. 

We say that a demand $D$ is \emph{$\vol_C$-respecting} if $|D(v)|\le \vol_C(v)$ for every vertex $v$.

We say that a demand is \emph{$h$-length symmetrically constrained} if it can be decomposed into pieces, each supported on a strongly connected subgraph of diameter at most $h$.\footnote{The formal definition can be found in \Cref{sec:preliminray}. Here we only give a simplified version, which is not fully rigorous.}

An \emph{$(h,s)$-length $\phi$-expander with respect to $\vol_C$} is a graph in which every $\vol_C$-respecting $h$-length symmetrically constrained demand can be routed with congestion at most $1/\phi$ and average cost at most $h s$. Here $s$ is the \emph{length slack}.

A \emph{$\phi$-expander hierarchy with length slack $s$} is a sequence of moving cuts $C_1,\dots,C_t$ such that for every scale $h$ and every level $i$, if we interpret 
$C_{>i}\stackrel{\mathrm{def}}{=}\sum_{j>i} C_j$
as an $h$-length moving cut, then $G-C_{>i}$ is an $(h,s)$-length $\phi$-expander with respect to $\vol_{C_i}$.

We require $C_1(e)=1$ for all $e\in E$, and in our setting we have $s=n^{o(1)}$ and $\phi=1/n^{o(1)}$. Such hierarchies are known to exist, and in \Cref{sec:LCED} we also give an algorithm to construct them.

\paragraph{Primitive 2: bounded-width path cover.}
Path covers were introduced in \cite{abs-2511-08551} as a directed analogue of neighborhood covers on undirected graphs.

A \emph{$d_{\cov}$-path cover} of $G$ with diameter $d_{\diam}$ is a graph $G'$, together with a projection map $\pi:V(G')\to V(G)$, such that:
\begin{enumerate}
    \item every edge of $G'$ projects to an edge of $G$ (so $\pi$ is essentially a graph homomorphism from $G'$ to $G$);
    \item every path $p$ of length at most $d_{\cov}$ in $G$ lifts to a path in $G'$, i.e., there is a path $p'$ such that $\pi(p')=p$ (under the natural definition of $\pi(p')$); and
    \item every SCC of $G'$ projects to an induced subgraph of $G$ and has strong diameter at most $d_{\diam}$.
\end{enumerate}

The \emph{width} of the path cover is
\[
\omega \stackrel{\mathrm{def}}{=} \max_{v\in V(G)} |\pi^{-1}(v)|,
\]
that is, the maximum number of copies of a vertex. Bounding the width is crucial for controlling congestion.

The construction in \cite{abs-2511-08551} gives near-linear size but does not control the width. For our purposes this is not enough. We therefore design a new construction, inspired by multiplicative weight updates, that penalizes vertices that are copied too often. The result is a path cover of width $n^{o(1)}$ with parameters
\[
d_{\cov}=n^{o(1)}\cdot d_{\diam}.
\]
The full construction is given in \Cref{sec:pathcover}. This construction is of independent interest. One open problem is to get the width to be polylogarithmic.
\footnote{In \cite{abs-2511-08551}, they only construct a path cover satisfying $d_{\cov}<d_{\diam}$. In comparison, this relation $d_{\cov}=n^{o(1)}\cdot d_{\diam}$ might look surprising at first glance. However, we show that this is easy to achieve by concatenating many copies of the same path cover into a single one (as shown in \Cref{lem:concatenating}). }

\paragraph{One-layer hierarchy.}
To illustrate the main idea, first assume that the input graph $G$ itself is already an $(h,s)$-length $\phi$-expander with respect to $\deg_G$, for every scale $h$.

For each dyadic scale $h'$, we construct a DAG $\widehat G^{(h')}$ that is responsible for routing flow paths of cost roughly $h'\cdot n^{o(1)}$. The final DAG projection is the disjoint union over all dyadic scales.

We start from a bounded-width path cover 
$G_{\cov}^{(h')}$
with diameter parameter $h_{\diam}=h'$, covering parameter $h_{\cov}=h'\cdot n^{o(1)}$, and width $\omega=n^{o(1)}$. Similar to the idea of \cite{Cohen00},
we look at its SCCs and classify them as follows.
\begin{description}
    \item[Light SCCs.]  
    If $|S|<n/\sigma$, where $\sigma=n^{o(1)}$ is a parameter, we recurse on $G[S]$ and replace $S$ by its recursively constructed DAG projection.

    \item[Heavy SCCs.]  
    If $|S|\ge n/\sigma$, we do not recurse. Instead, we use the expander property: since $G[S]$ has strong diameter at most $h'$, any sufficiently local demand inside $S$ can be routed at low congestion and average cost $h's$. We therefore replace $S$ by a simple \emph{star gadget}: one copy of $S$ for entering the gadget, one central hub vertex $v_S$, and one copy of $S$ for exiting it. The edges incident to the hub have capacity $\deg_G(\cdot)$ and length $h's$.
\end{description}

After replacing every SCC in this way, the SCC structure disappears and we obtain a DAG $D^{(h')}$. The crossing SCC edges can also be defined naturally in $D^{(h')}$.

The final step is to add \emph{bridges} between heavy SCCs. We duplicate $D^{(h')}$ into two copies and, for every pair of heavy SCCs $(S_1,S_2)$, we add $O(\log n)$ edges from the hub of $S_1$ to the hub of $S_2$. The $i$-th bridge has capacity equal to the maximum flow that can be sent from $S_1$ to $S_2$ in $G$ with average cost at most $r_i=(1+\eps)^i$. 

\paragraph{Forward mapping.}
Given a feasible flow $F$ in $G$, we lift each flow path $P$ to the appropriate scale $\widehat G^{(h')}$ according to its cost. Because the path cover lifts all paths up to length $h_{\cov}$, the path $P$ lifts to a path in $G_{\cov}^{(h')}$, which decomposes into pieces inside SCCs.

For a light SCC, we can simply recurse. So the case that all SCCs are light is easy to solve. 

If the path passes through at least on heavy SCC, let the first one be $S_1$ and last one by $S_2$. For each $S_1,S_2$, we route through the corresponding star gadget. We connect the two hub vertices $v_{S_1},v_{S_2}$ using the appropriate bridge int the following way: We choose a bridge with length approximately equal to the cost of the corresponding subpath from $S_1$ to $S_2$ up to a factor of $(1+\eps)$, while the additional star cost is only $h's$, which can be made negligible compared to the total cost of the path by parameter choice. Thus, by replacing the path between the two heavy SCCs by the bridge, the total loss is only a $(1+O(\eps))$ factor.

Feasibility follows because star capacities are chosen according to vertex degree, and bridge capacities are chosen according to the maximum amount of flow that can be routed at the corresponding average-cost scale.

\paragraph{Backward mapping.}
Conversely, a flow in the DAG can be mapped back to $G$:
\begin{itemize}
    \item the light-SCC part is handled recursively: Remember that each $G[S]$ where $S$ is a light has been replaced by a recursive DAG projection construction, such a construction itself also gives a backward mapping algorithm;
    \item the star gadgets are implemented using the expander property inside the corresponding heavy SCC: Remember that a SCC in an expander has the property that demand can be routed efficiently, and
    \item each bridge is realized by the definition of its capacity because the bridge capacity is defined as the maximum capacity that can be routed using the lengths of the bridge.
\end{itemize}
This gives the backward projection.

\paragraph{Multiple layers.}
For a full LC-expander hierarchy, the construction recurses through the hierarchy levels. At a high level, instead of building a star over each heavy cluster, we replace each heavy cluster by a lower layer DAG projection, then build the star and bridges. The analysis follows a similar flavour but is much more complicated. We defer the full proof to \Cref{sec:expandertoembedding}. 

\paragraph{Connection to previous works \cite{haeupler2025parallel}.} The construction has several similar ideas to the \emph{LC-flow shortcut} construction in \cite{haeupler2025parallel}. In their work, a shortcut on \emph{undirected} graph is constructed so that every flow has a corresponding low-step flow on the graph after adding this shortcut. Our construction of DAG projection is \emph{not a shortcut} and is working on \emph{directed graphs}. This leads to several difficulties compared to the shortcut construction.

The construction of \cite{haeupler2025parallel} uses \emph{neighborhood cover} on an undirected graph, but there is no such object that covers every neighbor in a directed graph. There is an object called \emph{path cover} with bounded width that we have introduced, but it is weaker than coving neighborhoods, and specifically, it does not allow the argument of covering one low-length path by one SCC - a critical requirement needed for \cite{haeupler2025parallel}. Thus, a more careful analysis of cumulative SCC covering along a path is needed.

\subsubsection{Fast construction}

The existence proof above is not yet algorithmic; the main bottleneck is that it seems to require exactly the approximate MCF oracle we are trying to obtain. To be precise, there are two places that require an MCF oracle:
\begin{enumerate}
    \item constructing the LC-expander hierarchy usually requires a MCF oracle for the cut-matching game, and
    \item computing the capacities of the bridges between heavy SCCs requires at least an approximate MCF.
\end{enumerate}

This results in a 'chicken-and-egg' situation: we want to compute DAG projection to get a MCF algorithm, but it requires MCF as a subroutine. 

A similar issue also appears in \cite{haeupler2026dag}. For their purpose, they wish to construct a DAG projection that only preserves distance so that it can help the algorithm for single-source shortest path (SSSP), but the construction requires (1) low-diameter decomposition (LDD), which requires an SSSP oracle, (2) computing bridges similar to our case, which requires an SSSP oracle. Their solution is to introduce an \emph{$h$-length} version of the projection problem: instead of preserving all distances, the projection only needs to preserve distances up to scale $h$. They can also define the $h$-length version of LDD and SSSP. Then they try to reduce the $h$-length problem to a smaller-scale $h/z$-length problem for a suitable parameter $z=n^{o(1)}$, and repeat this $O(\log n)$ times to get the final algorithm.

We apply the same philosophy here. We define:
\begin{itemize}
    \item an $h$-length DAG projection, which only needs to preserve flows of average cost at most $h$;
    \item an $h$-length approximate MCF problem; and
    \item an $h$-length LC-expander hierarchy.
\end{itemize}
The formal definitions are more subtle and appear in \Cref{sec:DAGprojections}.

At a high level, we obtain the following cycle of reductions:
\begin{enumerate}
    \item reduce $h$-length DAG projection to $h$-length LC-expander hierarchy plus $h$-length approximate MCF (in \Cref{sec:expandertoembedding});
    \item reduce $h$-length LC-expander hierarchy to $h$-length approximate MCF (in \Cref{sec:LCED});
    \item reduce $h$-length approximate MCF to $h$-length DAG projection assuming the DAG MCF oracle (see \Cref{lem:MCFtoDAG});
    \item reduce $h$-length DAG projection to $(h/z)$-length DAG projection using concatenation (see \Cref{lem:concatenating}).
\end{enumerate}

Taken together, these reductions decrease the length scale from $h$ to $h/z$, and the recursion terminates once $h$ becomes small enough to handle directly.

A natural question is why we do not simply apply the last reduction repeatedly, since it already decreases the length scale. The reason is that the concatenation lemma \Cref{lem:concatenating} enlarges the size of the DAG by a factor of $z$ each time. Repeating this for $O(\log_z n)$ levels would lead to a polynomial blow-up. In contrast, the four-step recursion above avoids this: the parameter $z$ enters the \emph{running time} of the recursion: the concatenation makes the algorithm for computing approximate MCF larger, but not increasing the final size of the DAG projection.

\paragraph{LC-expander hierarchy algorithm.}
We also briefly mention the algorithmic novelties in \Cref{sec:LCED}. Previous work gave existence results for directed LC-expanders \cite{HaeuplerLSW25}, and algorithmic constructions only in undirected graphs \cite{HaeuplerHT24}. In order to combine the previous ideas into a directed LC-expander hierarchy algorithm, two important ingredients were missing before our work:
\begin{enumerate}
    \item a \emph{fat low-diameter decomposition} in directed graph, namely a decomposition whose clusters are well separated in round-trip distance compared to their diameter; and
    \item a suitable directed cut-matching game with length constraints.
\end{enumerate}
We obtain the first by a reduction to its undirected variant in \Cref{sec:FatLDD}. The second contains more details, and we leave the full proof to \Cref{sec:LCED}.

\subsection{Weighted Push-Relabel Algorithm for MCF on DAGs}\label{subsec:overviewdagalgo}

Given the first two ingredients, it remains to solve approximate PD-MCF on DAGs. We prove the following in \Cref{sec:dag-mcf}.

\begin{lemma}\label{lem:overview DAG algo}
    There exists an \textbf{augmenting-path-based} algorithm on DAGs that returns a $(3,1+\eps)$-approximate PD-MCF solution.\footnote{In fact, the weighted push-relabel algorithm on DAGs returns a $(3,1+\eps)$-strongly-approximate PD-MCF solution, which is stronger than what the rest of the framework requires.}
\end{lemma}

The algorithm extends the weighted push-relabel method for maximum flow on DAGs from \cite{BernsteinBST24} to the minimum-cost flow setting.

As in classical push-relabel for MCF \cite{goldberg1990finding}, we maintain a potential function $\phi:V\to \mathbb{R}$. For an edge $(u,v)$, its reduced length is
\[
\ell_\phi(u,v)\stackrel{\mathrm{def}}{=}\ell(u,v)-(\phi(u)-\phi(v)).
\]
The algorithm alternates between:
\begin{itemize}
    \item \emph{relabel} steps, which decrease vertex potentials, and
    \item \emph{push} steps along \emph{admissible} edges, i.e., edges whose reduced length is sufficiently negative.
\end{itemize}

The classical Goldberg--Tarjan algorithm takes $\tilde O(mn)$ time essentially because an edge can become admissible up to $\tilde O(n)$ different times. Our key idea is to define admissibility using a carefully chosen weight function $w:E\to \mathbb{R}_{>0}$.

Following \cite{BernsteinBST24}, on a DAG one can assign weights so that:
\begin{enumerate}
    \item every flow $f$ of interest satisfies
    \[
    \frac{w(f)}{\val(f)}=O(n),
    \qquad\text{where } w(f)\stackrel{\mathrm{def}}{=}\sum_{e\in E} f(e)\,w(e),
    \]
    and
    \item
    \[
    \sum_{e\in E}\frac1{w(e)}=\tO{n}.
    \]
\end{enumerate}

We then declare an edge $e$ admissible when its reduced length is roughly $-\Theta(w(e))$. Saturating such an edge changes the relevant potentials by at least $w(e)$. Since the total potential range is $O(n)$, edge $e$ can be saturated only
$O\!\left(\frac{n}{w(e)}\right)$
times. Summing over all edges gives
\[
\sum_{e\in E} O\!\left(\frac{n}{\,w(e)}\right)
=
O\!\left(n\sum_{e\in E}\frac1{w(e)}\right)
=
O(n^2),
\]
which is the desired bound.

Although this idea is morally similar to \cite{BernsteinBST24}, extending it from max flow to min-cost flow is not automatic. Two main difficulties arise:
\begin{enumerate}
    \item In max flow, admissible paths move monotonically down the level structure. In min-cost flow, the relevant notion of level is given by potentials, and admissible paths need not be monotone in any obvious sense.
    \item In max flow, once the source label becomes small enough, one can directly certify that the current flow is good. In the min-cost setting this argument is more delicate, and we exploit the special structure of DAGs to make it work.
\end{enumerate}

These issues are resolved in \Cref{sec:dag-mcf}, completing the third ingredient of the framework.
\section{Preliminaries}\label{sec:preliminray}

We always use $n$ and $m$ to denote the number of vertices and edges of the original graph in the relevant context (for example, for each main theorem, the original graph is its input graph). Throughout the paper, all input graphs have length (cost) and capacity functions, and we assume that all lengths and capacities are \emph{integers} upper bounded by $N = \poly(n)$ and lower bounded by $-N$. We further assume that capacities are positive.

\paragraph{Directed Graphs.} Let $G=(V,E)$ be a directed graph. In this paper, $G$ can be a multigraph containing parallel edges.
Let $\ell_G : E \rightarrow \bbR$ denote the edge length (or \emph{cost}, in the minimum-cost-flow terminology) function of $G$. A path from vertex $v$ to vertex $w$ is called a $(v,w)$-path. For any path $P$, let $\ell_G(P) = \sum_{e \in P} \ell_G(e)$. The \emph{distance} between vertices $v$ and $w$ is $\dist_G(v, w) = \min_{P: (v,w)\textrm{-path}} \ell_G(P)$.
A \emph{ball} of radius $r$ around a vertex $v$ is $\ball_G(v,r) = \{ w \mid \dist_G(v,w) \le r\}$.
We further define $\dist_G(v,w) = \infty$ if vertex $w$ cannot be reached from vertex $v$.

Let $U_G : E \rightarrow \reals_{>0}$ denote the edge capacity function of $G$. 
Sometimes we use $U(e)$ to denote the capacity of edge $e$ when $G$ is explicitly mentioned.
We further define the directed degrees of a vertex $v$ in $G$, denoted by $\degi_G(v)$ and $\dego_G(v)$, as $\sum_{(w,v) \in E}U_G(w,v)$ and $\sum_{(v,w) \in E}U_G(v,w)$, respectively.
Namely, they represent the total capacities of edges incoming to $v$ and outgoing from $v$.
We use $\deg_G(v) = \degi_G(v) + \dego_G(v)$ to denote the total capacity of all edges incident to $v$.

\paragraph{Flows.}
A \emph{flow/routing} in $G$ is a function $F:{\cal P}\to \mathbb{R}_{\geq 0}$ (where ${\cal P}$ denotes the set of simple paths in $G$) that assigns each simple path $P$ in $G$ a flow value $F(P) \geq 0$. We define $P$ to be a \emph{flow path} of $F$ if $F(P) > 0$. Further, $P$ is a \emph{$(v,w)$-flow path} of $F$ if $P$ is both a $(v,w)$-path and a flow path of $F$. Let $\path(F)$ denote the set of all paths $P$ with $F(P) > 0$. The \emph{value} of $F$ is denoted by $\vvalue(F) = \sum_{P\in\path(F)}F(P)$.

The \emph{edge representation} of a flow is a function from edges to the amount of flow passing through each edge, i.e., $F(e) = \sum_{P: P \ni e} F(P)$. Since a flow-path representation can usually have superlinear size, whenever we store a flow, we store it in edge representation. The \emph{cost} of a flow is defined as $\cost(F)=\sum_{e\in E}F(e)\cdot \ell(e)$. The \emph{congestion of $F$ on an edge $e$} is $\conge_F(e) = \frac{F(e)}{U_G(e)}$.
We define a flow to be feasible if its congestion does not exceed one.

It is possible to perform a \emph{flow-path decomposition} of any flow given in edge representation to obtain one of its path representations, as stated in the following theorem.

\begin{theorem}[Lemma C.1 of \cite{BernsteinBLST25}]\label{thm:flowpathdecomposition}
    Given a flow $F$ on an $m$-edge graph $G$ with edge capacities and lengths, there is an $O(m\log m)$-time algorithm that computes a representation ${(L_i,\lambda_i,s_i,t_i)}_{i\in[k]}$ such that $F$ can be decomposed into $k$ flow paths, each with value $\lambda_i$ from $s_i$ to $t_i$ and with length $L_i$. 

    Moreover, given any subset $I\subseteq [k]$, there is an algorithm that computes an edge representation routing $\lambda_i$ amount of flow from $s_i$ to $t_i$ for every $i\in I$ such that the total cost is $\sum_{i\in I}L_i\cdot \lambda_i$.
\end{theorem}

A flow is called \emph{$h$-length bounded} if there exists a path decomposition of the flow such that the length of every flow path is at most $h$. The triple $(\lambda_i,s_i,t_i)_{i\in[k]}$ is called an \emph{implicit flow-path decomposition} if the flow can be decomposed into $k$ flow paths, each from $s_i$ to $t_i$ with value $\lambda_i$.

\medskip\noindent{\textbf{Demands.}}
A (single-commodity) demand $D:V \rightarrow \bbR$ assigns a value to each vertex such that $\sum_{v\in V}D(v)=0$. 
We say the demand $D$ is
\begin{itemize}
    \item \emph{symmetric $h$-length} if for every $u,v$ with $D(u),D(v)\not=0$, we have $\dist_G(u,v)\le h$ and $\dist_G(v,u)\le h$,
    \item \emph{$A$-respecting}: if for any $v\in V$, $|D(v)|\leq A(v)$. 
    \item \emph{$h$-length symmetrically constrained} if there exists a partition of $D$ into $D=D_1+D_2+...D_k$ such that for every $i\in[k]$, $D_i$ is symmetric $h$-length. We call the partition $D_1,D_2,...,D_k$ a \emph{witness} that $D$ is an $h$-length symmetrically constrained demand. 
    \item \emph{$A$-respecting $h$-length symmetrically constrained}: For the witness sequence of demands $D_1,...,D_k$, if for any $v\in V$, $\sum_{i\in[k]}|D_i(v)|\le A(v)$, then we say $D$ is $A$-respecting $h$-length symmetrically constrained. Notice that this is different from saying $D$ is $A$-respecting, since different $D_i$ with different signs might cancel out.
    \item \emph{feasible} if there is a feasible flow routing this demand.
\end{itemize}

We write the size of a demand as $|D|=\sum_{v\in V,D(v)>0}D(v)$.

We call $D$ an \emph{$s$-$t$ demand} if only $D(s)>0,D(t)<0$ are nonzero. 

\paragraph{Moving cuts.} We now give the definition of moving cuts.

\begin{definition}[Length-Constrained Moving Cut]\label{def:movingcut}
An $h$-length moving cut $C: E \mapsto \bbR_{\ge 0}$ assigns to each edge $e$ a non-negative
integer length increase $\ell_C=h\cdot C(e)$.%
We use $\leng(C)=h$ to denote the length parameter of $C$ (as an inherited parameter of the moving cut).
The \emph{size} of $C$ is defined as
\[
|C|=\sum_{e\in E} U(e)\cdot C(e).
\]
We define the degree of a moving cut at a vertex $v$ to be
\[
\deg_{C}(v) = \sum_{e\ni v}U(e)\cdot C(e).
\]
The graph $G-C$ denotes the graph using $\ell_{G}+h\cdot C$ as the length function.
We write $C'\preceq C$ if $C'(u,v)\le C(u,v)$ for all $(u,v)\in E$.
\end{definition}

\subsection{Minimum cost flow.}\label{subsec:MCF}
Given a graph $G=(V,E)$ with edge capacities and lengths and a single-commodity demand $D$,
the \emph{minimum cost flow (MCF)} problem asks for a flow routing $D$ with minimum cost.

Minimum cost flow has the following primal-dual LP.

\begin{minipage}[t]{0.45\textwidth}
\begin{align*}
\text{(P)}\quad &\min \sum_{e\in E} \ell(e)\cdot f(e) \\
\text{s.t.}\quad &\sum_{e\in\delta^+(u)} f(e) - \sum_{e\in\delta^-(u)} f(e) = D(u)
  &&\forall\,u\in V, \\
&0 \le f(e) \le U(e)
  &&\forall\,e\in E.
\end{align*}
\end{minipage}%
\hfill
\begin{minipage}[t]{0.45\textwidth}
\begin{align*}
\text{(D)}\quad &\displaystyle\max \sum_{u\in V} (-D(u))\cdot \phi(u)
   - \sum_{e\in E} U(e)\cdot \tl(e) \\
\text{s.t.}\quad &\phi(v) - \phi(u)\le \ell(e)+\tl(e)
  \quad \forall\,e=(u,v)\in E, \\
&\tl(e) \ge 0\quad 
  \forall\,e\in E.
\end{align*}
\end{minipage}

\newcommand{\Dual}{\mathrm{Dual}}
We can view $\tl:E\to\bbR_{\ge 0}$ as a length increase. We can view $\phi$ as a distance
function (from source to sink). The goal of the dual is to enlarge the distance from source
to sink by applying a small length increase. We write
\[
\Dual(\phi,\tl)=\sum_{u\in V} (-D(u))\cdot \phi(u) - \sum_{e\in E} U(e)\cdot \tl(e).
\]

We say a primal-dual triple $(f,\phi,\tl)$ is a $(\kappa,\lambda)$-\PDMCF{} if
(i) $f$ has congestion at most $\kappa$, and
(ii) $\cost(f)\le \lambda\cdot \Dual(\phi,\tl)$.

The problem of $(\kappa,\lambda)$-\PDMCF{} is: given a graph $G$ with positive, polynomially
bounded integer capacities and lengths, find a $(\kappa,\lambda)$-\PDMCF{}. We denote an oracle solving the $(\kappa,\lambda)$-\PDMCF{} by $\OPDMCF{\kappa,\lambda}$. 

The goal of this paper is to reduce the exact minimum-cost-flow problem on general graphs to $\OPDMCF{\kappa,\lambda}$ on DAGs, and to solve $\OPDMCF{\kappa,\lambda}$ on DAGs. 

When an algorithm makes oracle calls, we say that they are $f$-efficient if 
\begin{itemize}
    \item the cumulative number of vertices in the graphs on which the oracle runs is at most $f\cdot n$,
    \item the cumulative number of edges in the graphs on which the oracle runs is at most $f\cdot m$,
    \item the additional running time is at most $f\cdot m$.
\end{itemize}

\subsection{Projection map.}
We will extensively use the concept of \emph{projections}, introduced in \cite{abs-2511-08551,haeupler2026dag}, throughout this paper.

A graph $G'=(V',E')$ can be assigned a projection map $\pi: V' \to V \cup \{\bot\}$ (where $V$ may be the vertex set of another graph $G=(V,E)$). We write $\pi^{-1} : V \to 2^{V'}$ for the inverse mapping, where $\pi^{-1}(u) = \{ v \in V' \mid \pi(v) = u \}$. For a set $S \subseteq V$, we write $\pi^{-1}(S) = \bigcup_{u \in S} \pi^{-1}(u)$.

For a set $S \subseteq V'$, we write
\[
\pi(S) = \{ \pi(s) \mid s \in S \} \setminus \{ \bot \}.
\]
For a path $(v_1,...,v_k)$ in $G'$, we write $\pi(p)$ for the (potentially valid) path in another graph given by $(\pi(v_1),...,\pi(v_k))$.

\newcommand{\flpi}{\hat\pi^{-1}}
When $G'$ is a DAG, for each $v\in V$, write $\vf,\vl$ for the first and last vertex $v'$ in the topological order of $G'$ with $\pi(v')=v$, respectively. For a demand $D'$ on $G'$, we write $D=\pi(D')$ for the demand $D(v)=\sum_{\pi(v')=v}D'(v')$. For a demand $D$ on $G$, we write the lifted demand from first to last vertex on $G'$ as $D'=\flpi(D)$ where $D'(\vf)=D(v)$ for all $D(v)>0$ and $D'(\vl)=D(v)$ for all $D(v)<0$, and all other nodes have $D'(v)=0$. 

In \cite{abs-2511-08551}, they introduced a useful notion of \emph{path cover}, which we state as follows.

\begin{definition}\label{def:pathcover}
    A \emph{$h_{\cov}$-path cover with diameter $h_\diam$} of a graph $G=(V,E)$ is a graph $G'$ along with a projection map $\pi$ such that
    \begin{itemize}
        \item for every $(u,v)\in E(G')$, we have $(\pi(u),\pi(v))\in E(G)$ with the same edge length, and $\pi(u)\not=\bot$ for every $u\in V(G')$,
        \item for every path $p$ of length at most $h_\cov$ in $G$, there is a path $p'$ in $G'$ with $\pi(p')=p$,
        \item for every strongly connected component (SCC) $C$ of $G'$, $G'[C]$ has diameter at most $h_\diam$.
    \end{itemize}
    The size of the path cover is $|E(G')|$. We say $G'$ is \emph{induced} if for every $(u,v)\in E(G')$ we have $(\pi(u),\pi(v))\in E(G)$, and every SCC of $G'$ maps to an induced subgraph of $G$ under $\pi$. The \emph{width} of the path cover is defined as $\max_{v\in V}|\pi^{-1}(v)|$.
\end{definition}

In \cite{abs-2511-08551}, they only show how to construct a path cover with bounded size, but with unbounded width. We will show in \Cref{sec:pathcover} how to extend their idea to obtain a bounded-width path cover. We summarize the result as follows.
\begin{restatable}{lemma}{pathcover}\label{lem:pathcover}
There exists a deterministic algorithm that, given a directed graph $G=(V,E)$ with non-negative edge weights, a diameter parameter $d$, a parameter $\eps\in (\log\log n/\log n,1)$ and a stretch parameter $q\in \bbN_+$, 
constructs a $qd$-path cover $H$ of $G$ such that
\begin{enumerate}
    \item every SCC of $H$ has strong diameter $d\cdot \exp(O(1/\eps))$,
    \item the width is at most $q\cdot n^{O(\eps)}$,
\end{enumerate}
in $q\cdot m^{1+o(\eps)}$ time.
\end{restatable}
\section{Boosting Approximate Minimum Cost Flow}\label{sec:boosting}

In this section we will show how to get an exact minimum cost flow by oracle calls to approximate minimum cost flow.

\boosting*
In \Cref{subsec:exactflowwithapproximatecost}, we will show that if we have an oracle solving approximate primal-dual to a \emph{MCF LP}, then there is a simple way to get an exact algorithm. This boosting holds even when the MCF LP only routes a small fraction of the required demands.

However, a $(\kappa,1+\eps)$-\PDMCF{} oracle does not directly give a solution to the MCF LP. In \Cref{subsec:flowboosting}, we will show how to use MWU to turn the $(\kappa,1+\eps)$-\PDMCF{} oracle into a MCF LP solution and thus get the exact algorithm.

\begin{proof}[Proof of \Cref{thm:boosting}]
    By combining \Cref{lem:weak-to-strong-via-cost-scaling,lem:weakLPboosting}, we get the required complexity.
\end{proof}

\subsection{Exact Flow with Approximate Cost}\label{subsec:exactflowwithapproximatecost}
\paragraph{Primal-dual pair of minimum cost flow.} We call following primal dual LP for minimum cost flow the \emph{MCF LP}, as defined in the preliminary.

\begin{minipage}[t]{0.45\textwidth}
\begin{align*}
\text{(P)}\quad &\min \sum_{e\in E} \ell(e)\cdot f(e) \\
\text{s.t.}\quad &\sum_{e\in\delta^+(u)} f(e) - \sum_{e\in\delta^-(u)} f(e) = D(u)
  &&\forall\,u\in V,\quad \\
&0 \le f(e) \le U(e)
  &&\forall\,e\in E
\end{align*}
\end{minipage}%
\hfill
\begin{minipage}[t]{0.45\textwidth}
\begin{align*}
\text{(D)}\quad &\displaystyle\max \sum_{u\in V} (-D(u))\cdot \phi(u)
   - \sum_{e\in E} U(e)\cdot \tilde\ell(e) \\
\text{s.t.}\quad &\phi(v) - \phi(u)\le \tilde\ell(e)+\ell(e)
  \quad \forall\,e=(u,v)\in E, \\
&\tilde\ell(e) \ge 0\quad 
  \forall\,e\in E.
\end{align*}
\end{minipage}

\paragraph{Residual minimum cost flow problem.} Given a primal-dual pair $f,\phi,\tilde\ell$, we define the residual problem $(G_{\res},D_{\res})=(G,D)_{f,\phi}$ as follows (notice that $\tilde\ell$ should also be in the dual solution, but we omit it here because the residual problem only depends on $f$ and $\phi$). 

Let $G_f$ be the residual graph of $G$ on $f$. \footnote{Notice that $G_f$ might contain parallel edges: for an edge $e\in E$, if $f$ sends $x$ amount of flow on $e$, then the residual graph will create a reversed edge with cost $-\ell_G(e)$ and capacity $x$, which might be in parallel with some existing edges with different costs.}
Let $(G_f)_{\phi}$ be the graph $G_f$ applying $\phi$ as a potential to change the edge length (cost). To be precise, for every edge $(u,v)\in E((G_f)_{\phi})$, we have $\ell_{(G_f)_\phi}(u,v)=\ell_G(u,v)-(\phi(v)-\phi(u))$ (for convenience, if $(u,v)\not\in E$ but $(u,v)\in E(G_f)$, we define $\ell_G(u,v)=-\ell_G(v,u)$).

Suppose $E_{\neg}$ contains all edges in $(G_f)_\phi$ that has negative length, then $G_{\res}$ is defined as reversing all edges $(u,v)\in E_{\neg}$ (and change the weight to $-\ell_{(G_f)_\phi}(u,v)$). Notice that $G_{\res}$ has only non-negative lengths. $D_{\res}$ is defined as the demand such that every edge $(u,v)\in E_{\neg}$ contribute $U(u,v)$ unit source demands $+1$ to $v$ and unit sink demands $-1$ to $u$. 

The following lemma shows that the exact minimum cost flow can be solved by finding a flow plus an optimal flow for the residual problem.

\begin{lemma}\label{lem:residual}
    Suppose $f,\phi,\tilde\ell$ is a feasible primal-dual pair for the minimum cost flow problem $(G,D)$, and $f^*_{\res}$ is a minimum cost flow of $(G,D)_{f,\phi}$, then $f+f^*_{\res}$ is a minimum cost flow of $(G,D)$.
\end{lemma}
\begin{proof}
    Since $f^*_{\res}$ is a minimum cost flow of $(G,D)_{f,\phi}$, we have that $(G_f)_{f^*_{\res}}=G_{f+f^*_\res}$ has no negative cycle. Notice that $f+f^*_{\res}$ is a feasible flow routing $D$ according to definition, thus, $f+f^*_{\res}$ is optimal.
\end{proof}

The following lemma shows that the residual problem has a small optimal value if the primal-dual pair is a good approximation. 

\begin{lemma}\label{lem:smallresidual}
    Suppose $f,\phi,\tilde\ell$ is a feasible primal-dual pair of $(G,D)$ with primal-dual value difference $\Delta$, then the optimal cost of the minimum cost flow instance $(G,D)_{f,\phi}$ is at most $\Delta$. More precisely, we have that
\[
\sum_{\substack{
e \in (G_f)_\phi \\
\ell_{(G_f)_\phi}(e) < 0
}}\left|\ell_{(G_f)_\phi}(e)\right|\le \Delta
\]
\end{lemma}
\begin{proof}
    We use the same notation $G_f,(G_f)_\phi,E_{\neg},G_{\res},E_{\res}$ as in the definition of residual problem. For convenience, we write $\ell(\cdot),U(\cdot)$ as the length and capacity function of the original graph $G$. We write $\ell_\phi(u,v)$ as the length function of $(G_f)_\phi$. We write $U_f(\cdot )$ as the capacity function on the residual graph $G_f$. 
    
    By definition, we can route $D_{\res}$ by sending $U_f(v,u)$ units of flow for each edge $(u,v)\in E_{\neg}$. Thus, the optimal cost of $(G_{\res},D_{\res})$ is at most $\sum_{e\in E_{\neg}}(-U_f(e)\cdot \ell_\phi(e))$. For each $e\in E_{\neg}$, there are two cases,
    \begin{enumerate}[label=(\roman*)]
        \item if $e=(u,v)\in E$, then since $\phi,\tilde\ell$ is a feasible dual solution, we have $-\ell_\phi(e)=-(\ell(e)-\phi(v)+\phi(u))\le \tilde\ell(e)$,
        \item if $e=(v,u)\not\in E$ but $e\in E(G_f)$ (which means $f$ sends flows by $(u,v)$ in $G$ so $(u,v)\in E$), then we have $-\ell_\phi(v,u)=\ell(u,v)-(\phi(v)-\phi(u))$. Notice that in this case we must have $U_f(v,u)=f(u,v)=f(e)$. 
    \end{enumerate}
    Thus, we have

    \begin{align*}
        \sum_{e\in E_{\neg}}(-U_f(e)\cdot \ell_\phi(e)) & \le \sum_{e\in E}U(e)\cdot \tilde\ell(e)+\sum_{(u,v)\in E}f(u,v)\cdot (\ell(u,v)-\phi(v)+\phi(u))\\
        & = \sum_{e\in E}f(e)\cdot \ell(e)+\sum_{(u,v)\in E}f(u,v)\cdot (\phi(u)-\phi(v))+\sum_{e\in E}U(e)\cdot \tilde\ell(e)\\
        & = \sum_{e\in E}f(e)\cdot \ell(e)+\sum_{u\in V}\phi(u)\cdot \left(\sum_{e\in\delta^+(u)}f(e)-\sum_{e\in\delta^-(u)}f(e)\right)+\sum_{e\in E}U(e)\cdot \tilde\ell(e)\\
        & = \sum_{e\in E}f(e)\cdot \ell(e)-\left(\sum_{u\in V}\phi(u)\cdot (-D(u))-\sum_{e\in E}U(e)\cdot \tilde\ell(e)\right)
    \end{align*}
    Notice that this is exactly the difference between primal and dual values of $f,\phi,\tilde\ell$.
\end{proof}

\Cref{lem:smallresidual} implies that, if we have an algorithm that gives $1\pm\eps$ approximate primal dual pair for minimum cost flow, then their is an exact algorithm using $O(\log n)$ calls of the approximation algorithm whenever $\eps<0.5$.

In many situations, we cannot route $D$ perfectly. The following lemma shows how to deal with this case.

\newcommand{\OPT}{\mathsf{OPT}}

\begin{lemma}[Boosting fractional $(1\pm\eps)$ primal--dual to exact MCF]
\label{lem:ApxflowApxcostBoosting}

Suppose there is an algorithm $\cA$ which, given a non-negative minimum-cost flow instance $(G,D)$,
returns
(i) a feasible flow $f$ routing a sub-demand $D'\preceq D$ with $|D'|\ge |D|/\kappa$ (where $\kappa\ge 1$), and
(ii) a dual pair $(\phi,\tilde\ell)$
such that $(f,\phi,\tilde\ell)$ is a $(1\pm\eps)$ approximate primal--dual pair for routing $D'$ in $G$ with $0\le \eps\le 0.1$.
Then there is an algorithm that computes an \emph{exact} minimum-cost flow routing all of $D$ in $G$
using at most $\tO{\kappa}$-efficient calls to $\cA$.
\end{lemma}

\begin{proof}
We describe a recursive procedure $\cB(G,D)$ that returns an exact optimal flow routing $D$ in $G$. We assume the demand $D$ contains only one source $s$ and one sink $t$ for technical reasons. This assumption is w.o.l.g. since we can add a super source $s$ and super sink $t$ connecting to every vertex with positive/negative demand with the corresponding capacities and length $0$. 

The algorithm initially sets $G^{(0)}=G$ and $D^{(0)}=D$. Starting from $i=0$, the algorithm will repeat the following procedure to compute flows.

For each $i$, the algorithm makes one call to $\cA$ to compute a flow $f_i$ routing $\bar{D}^{(i)}\preceq D^{(i)}$ and an approximate dual $(\phi_i,\tilde\ell_i)$. Set $G^{(i+1)}=\left((G^{(i)}_{f_i})_{\phi_i}\right)_{\ge 0}$, i.e., the residual graph after routing the flow $f_i$ and applying the potential function $\phi_i$, then round each negative weight edge up to $0$. Set $D^{(i+1)}=D^{(i)}-\bar{D}^{(i)}$. 

We repeat the above operation until $D^{(i)}$ has sufficiently small size (less than $1/\poly(n)$), in which case we let $f_i$ be an arbitrary flow routing $D^{(i)}$ and end the iterations.

Let $f^*=\sum_{i\ge 0}f_i$ and $\phi^*=\sum_{i\ge 0}\phi_i$. Consider the residual instance $(G_{\res},D_{\res})=(G,D)_{f^*,\phi^*}$, we recursively call $\cB(G_{\res},D_{\res})$ to find a flow $f'$ routing $D_{\res}$ on $G_{\res}$. The final flow we return is $f^*+f'$. 

The base case of the recursion is when the residual instance $\cB(G_{\res},D_{\res})$ can be routed by cost sufficiently small (less than $1/\poly(n)$), then we can let $f'$ be an arbitrary flow routing $D_{\res}$, and use a standard rouding lemma\footnote{See Lemma 8.10 of \cite{BrandLNPSS0W20}} to get an exact optimal solution.

\paragraph{Correctness.} The correctness for the non-base case follows from \Cref{lem:residual} since $f^*,\phi^*$ is a feasible primal and dual solution and $f'$ is an optimal solution for $(G_{\res},D_{\res})$. In the base case, the flow is a $(1+1/\poly(n))$ approximation to the optimal solution, and the correctness is by the rounding lemma in \cite{BrandLNPSS0W20}.

\paragraph{Complexity.} The complexity relies on the following lemma.

\begin{lemma}\label{lem:sumofflows}
    Let the optimal cost for routing $D$ on $G$ be $\OPT_G(D)$. We have that
    \[\sum_{i\ge 0}\ell_{G^{(i)}}(f_i)\le 2\cdot \OPT_G(D)\]
\end{lemma}
\begin{proof}
    We will prove by induction that 
    \begin{equation}\label{eq:optsize}
        \OPT_{G^{(i)}}(D^{(i)})\le \OPT_G(D)-0.9\cdot \sum_{j<i}\ell_{G^{(j)}}(f_j)
    \end{equation}
    Notice that if \Cref{eq:optsize} is correct, then we are done since $\OPT_{G^{(i)}}(D^{(i)})\ge 0$ ($G^{(i)}$ always contains positive edge lengths).

    When $i=0$, we have that $\OPT_{G^{(i)}}(D^{(i)})=\OPT_G(D)$, so \Cref{eq:optsize} holds.

    Suppose $i>0$. We have that 
    \[\OPT_{G^{(i-1)}}(D^{(i-1)})\le \OPT_G(D)-0.8\cdot \sum_{j<i-1}\ell_{G^{(j)}}(f_j)\]
    by induction hypothesis. Let $\hat{f}$ be the optimal flow routing $D^{(i-1)}$ in $G^{(i-1)}$. Recall that $f_{i-1}$ routes $\bar{D}^{(i-1)}$ in $G^{(i-1)}$. Thus, there is a flow $f=\hat{f}-f_{i-1}$ routing $D^{(i)}=D^{(i-1)}-\bar{D}^{(i-1)}$ in the residual graph $G^{(i-1)}_{f_{i-1}}$, with cost at most $\OPT_{G^{(i-1)}}(D^{(i-1)})-\ell_{G^{(i-1)}}(f_{i-1})$. Now consider the cost of $f$ on the graph $(G^{(i-1)}_{f_{i-1}})_{\phi_{i-1}}$, we claim that the cost does not increase after applying $\phi_{i-1}$: since we are working on single source $s$ and single sink $t$, the cost will not increase if $\phi_{i-1}(t)\ge \phi_{i-1}(s)$, which can be guaranteed w.o.l.g. if we set $\phi_{i-1}$ to be the distance function with source $s$ and non-negative edge weights $\ell_{G^{(i-1)}}+\tilde\ell_{i-1}$. This adjustment to $\phi_{i-1}$ makes the difference between $\phi(t)$ and $\phi(s)$ the maximum, and does not change $\tilde\ell_{i-1}$, so the dual solution does not decrease. Finally, we have that $G^{(i)}=\left((G^{(i-1)}_{f_{i-1}})_{\phi_{i-1}}\right)_{\ge 0}$, for which the total cost increase for $f$ is at most $2\eps\cdot \ell_{G^{(i-1)}}(f_{i-1})$ according to \Cref{lem:smallresidual}. Thus, we get \Cref{eq:optsize}.
\end{proof}

We also need the following lemma to bound the cost for the recursive call.

\begin{lemma}\label{lem:recursivecost}
    Let $E_{\neg}$ denote the negative edges of $G_{f^*,\phi^*}$, we have that
    \[\sum_{e\in E_{\neg}}\left|\ell_{G_{f^*,\phi^*}}(e)\right|\le 3\eps \cdot \sum_{i\ge 0}\ell_{G^{(i)}}(f_i)\]
\end{lemma}
\begin{proof}
    Let $G^{(k)}$ be the graph in the last iteration, we have that $G^{(k)}$ only contains non-negative cost edges. For each edge $e$ of $G^{(k)}$, the difference 
    \[U_{G^{(k)}}(e)\cdot |\ell_{G^{(k)}}(e)-\ell_{G_{f^*,\phi^*}}(e)|\]
    comes from a rounding up operation that rounds some negative edges to $0$ in some iterations. Consider the $i$-th iteration, the total amound of rounding up depends on the total size of the negative edges
    \[\sum_{\ell_{G^{(i)}_{\phi_i}}(e)<0}U_{G^{(i)}_{\phi_i}}(e)\cdot |\ell_{G^{(i)}_{\phi_i}}(e)|\]
    which according to \Cref{lem:smallresidual}, is at most 
    \[3\eps \cdot \ell_{G^{(i)}}(f_i)\]
    Summing over all iterations gives the lemma.
\end{proof}

\Cref{lem:recursivecost} combines with \Cref{lem:sumofflows} gives an upper bound on the optimal cost of the recursive instance to be $6\eps\cdot \OPT_G(D)$. Thus, the recursion has depth at most $O(\log _{1/\eps}n)$ given that $\OPT_G(D)=\poly(n)$ as our assumption. Moreover, the algorithm makes at most $O(\kappa\log n)$ calls to $\cA$ since each call to $\cA$ reduce the remaining demand $D^{(i)}$ by a fraction of $(1-1/\kappa)$. Thus, the total number of calls to $\cA$ is $\tO{\kappa}$ as long as $\eps<0.1$.

\end{proof}

\subsection{Flow Boosting}\label{subsec:flowboosting}
In this section, we show how to use flow boosting to obtain a $(1+\eps)$ approximation to the
following primal-dual form of minimum cost flow, using oracle calls to
$(\kappa,(1+\eps))$-\PDMCF{}.
Notice that the following formalization is weaker than the MCF LP stated in the last section
(for example, an approximation to this LP does not solve max flow exactly while the strong
version does). We call it the \emph{MBCF LP}. For convenience, we only consider $s$--$t$
flow, and let $\cF_{s,t}$ contain all $s$--$t$ feasible flows. Remember the definition of moving cut \Cref{def:movingcut}. We consider $C$ as a moving cut of length $h$.

\begin{align*}
\max \qquad & \val(F) \\
\text{s.t.}\qquad & \cost(F)\le B \\
& F\in \cF_{s,t}
\end{align*}
\begin{align*}
\min \qquad & \sum_{e\in E} U(e)\cdot C(e) + B/h \\
\text{s.t.}\qquad & \dist_{G-C}(s,t)\ge h \\
& C\ge0,\;h>0\\
&C\text{ is a }h\text{-length moving cut}
\end{align*}

We prove the following lemma.
\begin{lemma}\label{lem:weakLPboosting}
There is an algorithm that, given a graph $G$ with non-negative edge capacities and lengths,
computes a $(1+\eps)$-approximate primal-dual pair of the MBCF LP, making
$\tO{\kappa^2/\eps^3}$-efficient calls to a $(\kappa,o(\eps))$-\PDMCF{} oracle.
\end{lemma}

One benefit of the MBCF LP is that it admits the following flow boosting template using
multiplicative weight update (MWU). The following flow boosting template is from
\cite{HaeuplerH0RS24}.

\begin{theorem}[Theorem 9.2 of \cite{HaeuplerH0RS24}]\label{thm:boosting-template}
Let $G=(V,E)$ be a graph with capacity function $U_G$ and cost function $\ell_G$, and let
$B\ge0$ be a cost budget. Let $\mathcal F$ be any convex set of flows in $G$ containing at least
one capacity-respecting flow, and let $\eps,s,\kappa\ge0$ be parameters. Suppose an algorithm
is given an oracle $\mathcal O$ that, given any integral edge length function
$\ell:E\to\{1,2,\ldots,O(m^{1/\eps}N/\eps)\}$, computes the edge representation of a
(not necessarily capacity-respecting) flow $F\in\mathcal F$ such that
\begin{itemize}
\item \textup{Length slack $s$:} $\sum_{e\in E}F(e)\cdot\ell(e)\le s\cdot\sum_{e\in E}F^*(e)\cdot\ell(e)$
for any flow $F^*\in\mathcal F$ that is also capacity-respecting, and
\item \textup{Congestion slack $\kappa$:} $F(e)\le\kappa U(e)$ for all $e\in E$, i.e., $F(e)/\kappa$ is
capacity-respecting.
\end{itemize}
Then, there is a deterministic algorithm that outputs the edge representation of a flow
$\bar F\in\mathcal F$ and scalar $\lambda\ge0$ such that
\begin{enumerate}
\item \textup{Feasibility:} The flow $\lambda\bar F$ is capacity-respecting with cost at most $B$, and
\label{item:feasibility}
\item \textup{Approximation factor $\approx s$:} Let $\lambda^*$ be the maximum value such that there
exists flow $F^*\in\mathcal F$ where $\lambda^*F^*$ is capacity-respecting with cost at most $B$.
Then, $\lambda\ge\frac{1-O(\eps)}s\lambda^*$.
\label{item:approx-factor}
\item \textup{Dual solution $h,C$:} A length value $h>1$ and an $h$-length moving cut $C$ such that
(i) for every capacity-respecting flow $F\in\mathcal F$ we have $\cost_{G-C}(F)\ge h$,
(ii) $|C|+B/h\le (1+\eps)\lambda$.
\end{enumerate}

The algorithm takes $\tO{\lambda \kappa\eps^{-2}}$ calls to the oracle and
$\tilde O(\lambda \kappa\eps^{-2}m)$ time outside of the oracle calls.
\end{theorem}

Now we are ready to prove \Cref{lem:weakLPboosting}.

\begin{proof}[Proof of \Cref{lem:weakLPboosting}]
Let $P$ be the optimal value for the MBCF LP. We assume we know an approximation
$\tP\in[P,2P]$, which can be guessed by powers of $2$. We scale down the capacities of every
edge and $B$ by a factor of $\tP\cdot \eps'$ where $\eps'$ is a sufficiently small parameter
which will be decided later. Now it remains to solve the MBCF LP on this scaled down graph
to get $F,C,h$ where $F$ has value roughly $1/\eps'$. After which we can get the solution to
the original graph by scaling up $F$ and the capacity $U$ both by $\tP\cdot \eps'$, which makes
the primal solution $\tP\cdot \eps'$ times larger and the dual solution
$\sum_{e\in E}U(e)\cdot C(e)+B/h$ also $\tP\cdot \eps'$ times larger, so it remains a
$(1+\eps)$-approximate primal-dual pair. Now we focus on the scaled down graph and assume that
the optimal solution has value roughly $1/\eps'$.

We first apply \Cref{thm:boosting-template} by setting $\mathcal F$ to contain all $s$--$t$ flows with
value $1$. To obtain the oracle $\mathcal O$ for \Cref{thm:boosting-template}, we apply the
$(\kappa,1+o(\eps))$-\PDMCF{} algorithm on $G$ and the single-commodity demand with $1$ unit of
demand from $s$ to $t$ (although the algorithm only works on integer-weighted graphs, we can
scale up the fractional capacity to make it integral). It will return a flow $F$ with length at
most $s=1+o(\eps)$ times the cost of the optimal capacity-respecting flow in $\mathcal F$, and
congestion at most $\kappa$, so it is a valid oracle.

After calling \Cref{thm:boosting-template}, we obtain $\bar{F}$ and $h,C$. Notice that $\bar{F}$
is a convex combination of flows in $\mathcal F$, so $\bar{F}$ has value $1$. Since $\lambda \bar{F}$ is
capacity-respecting with cost at most $B$ and we know that the largest such flow has value
$1/\eps'$, we have that $\lambda\le 1/\eps'$. Moreover, $\lambda\bar{F}$ is a
$(1-o(\eps))/(1+o(\eps))$ approximation to the optimal capacity-respecting flow with cost at most
$B$, by \Cref{thm:boosting-template}\Cref{item:approx-factor}. We conclude that $\lambda\bar{F}$
is a $(1-\eps)$-approximate primal solution to the MBCF LP.

It remains to compute a dual solution for the MBCF LP. Recall that
\Cref{thm:boosting-template} returns $h,C$ where $C$ is an $h$-length cut satisfying
(i) for every capacity-respecting flow $F$ with value $1$ we have $\cost_{G-C}(F)\ge h$, and
(ii) $|C|+B/h\le (1+o(\eps))\lambda$.

We apply the $(\kappa,(1+o(\eps)))$-\PDMCF{} algorithm on $G-C$ and the demand sending $\kappa$
from $s$ to $t$, and denote the returned primal-dual triple by $(f,\phi,\tilde\ell)$. Since $f$ has
congestion at most $\kappa$, $f/\kappa$ is a feasible flow with value $1$, which by property (i)
has cost at least $h$. Thus, $\cost_{G-C}(f)\ge h\cdot \kappa$. Moreover, by the definition of
$(\kappa,(1+o(\eps)))$-\PDMCF{}, the dual value satisfies
\[
\kappa\cdot (\phi(t)-\phi(s))-\sum_{e\in E}U(e)\cdot \tilde\ell(e)\ge h\cdot \kappa/(1+o(\eps)).
\]

Let $h'=h/(1+o(\eps))$. We have $\kappa\cdot (\phi(t)-\phi(s))\ge h'\cdot \kappa$, so
$\phi(t)-\phi(s)\ge h'$.

Let $C'$ be an $h'$-length cut defined as follows:
(i) if $\dist_{G-C}(s,t)\ge h'$, then $C'=0$; otherwise,
(ii) let $C'= (\tilde\ell/h')/T$ where $T$ is the largest value such that when we interpret $C'$ as an
$h'$-length cut, we get $\dist_{G-C-C'}(s,t)=h'$
(such a $T$ must exist since $\dist_{G-C}(s,t)<h'$ and $\dist_{G-C+\tilde\ell}(s,t)\ge \phi(t)-\phi(s)\ge h'$).
Moreover, we must have $\phi(t)-\phi(s)\le h'\cdot T$ since scaling up $C'$ by $T$
while keeping the edge length of $G$ the same does not increase the distance to more than $T$
times larger. Thus, we get
\[
|C'|=\sum_{e\in E}U(e)\cdot \tilde\ell(e)/(Th')\le
\frac{\kappa(\phi(t)-\phi(s))-h\cdot \kappa/(1+o(\eps))}{Th'}\le \kappa.
\]

Let $C^*$ be an $h'$-length cut defined as $C^*=C\cdot (h/h')+C'$ so that $G-C^*=(G-C)-C'$.
It follows from $|C|+B/h\le (1+o(\eps))\lambda$ that
\[
|C^*|+B/(h')\le (1+o(\eps))\cdot \bigl((1+o(\eps))\lambda +\kappa\bigr).
\]

Notice that $\lambda$ is roughly $1/\eps'$. By setting $\eps'=o(\eps/\kappa)$, we get that
\[
|C^*|+B/h'\le (1+\eps)\lambda
\qquad\text{and}\qquad
\dist_{G-C^*}(s,t)\ge h',
\]
so $C^*,h'$ is a valid $(1+\eps)$-approximate dual.

According to \Cref{thm:boosting-template}, the algorithm makes
$\tO{\lambda\kappa\eps^{-2}}=\tO{\kappa^2\eps^{-3}}$ calls to the oracle, and an additional
$\tO{\kappa^2\eps^{-3}m}$ time.
\end{proof}

\subsection{Converting MCF solution: From MBCF to MCF}

Now we show how to convert solutions to the MBCF LP into the one needed in the previous section.

\begin{lemma}\label{lem:weak-to-strong-via-cost-scaling}
Suppose there is an algorithm $\cA$ that, for every cost bound $B$, returns a primal-dual pair
$(f,C,h)$ for the MBCF LP with approximation ratio $(1+\eps)$, where $\eps<1/\log^9 n$.
Then there is an algorithm that makes $2^{O(\log_{1/\eps^{1/3}} n)}$ calls to $\cA$ and computes
an exact minimum cost flow routing $D$ in $G$, assuming that there exists a feasible flow routing
$D$ in $G$.
\end{lemma}

\begin{proof}
We describe an oracle $\cA'$ that can be plugged into \Cref{lem:ApxflowApxcostBoosting}.
For convenience, we first assume $D$ is a single-source single-sink demand from $s$ to $t$. 
\paragraph{Budget schedule and the index $a$.}
Let
\[
\alpha \stackrel{\mathrm{def}}{=} \eps^{1/3}
\qquad\text{and}\qquad
\beta \stackrel{\mathrm{def}}{=} \eps^{2/3}.
\]
We will call $\cA$ on a geometric sequence of budgets. Define
\[
B_i \stackrel{\mathrm{def}}{=} (1+\alpha)^{i-2}
\qquad\text{for } i=0,1,2,\ldots,I,
\]
where $I$ is the smallest index such that $B_I \ge B_{\max}$ and
\[
B_{\max}\stackrel{\mathrm{def}}{=} \sum_{e\in E} U(e)\cdot \ell(e).
\]
Note that $B_{\max}$ is a valid a priori upper bound on the cost of \emph{any} feasible
(capacity-respecting) flow.

For each $i\in\{0,1,\ldots,I\}$, run $\cA$ with budget $B_i$ and obtain $(f_i,C_i,h_i)$.
Let $q_i \stackrel{\mathrm{def}}{=} \val(f_i)$.

Since $B_I\ge B_{\max}$ and the instance is feasible, the MBCF optimum at budget $B_I$ equals
$|D|$, hence by the approximation guarantee of $\cA$ we have
\[
q_I \ge \frac{1}{1+\eps}\cdot |D| > |D|/2.
\]
Moreover, we can assume $q_0\le |D|/2$. Otherwise, it implies that with a cost less than $1/4$, we can already route half of the demands. We can assume edge lengths and capacities, the demand $D$ are all non-negative integral by scaling. Thus, this implies that by routing through only $0$ length edge, we can already route at least $3/4$ fraction of the demands, which finishes the oracle.

Define
\[
i_0 \stackrel{\mathrm{def}}{=} \max\{\, i\in\{0,1,\ldots,I-1\} : q_i \le |D|/2\,\}.
\]
This $i_0$ is well-defined and satisfies $q_{i_0}\le |D|/2$ and $q_{i_0+1}>|D|/2$.

Now choose an index
\[
a \in \arg\max_{i\in\{i_0,i_0+1,\ldots,I-1\}} \frac{q_{i+1}}{q_i},
\]
with the convention that $\frac{q_{i+1}}{q_i}=+\infty$ if $q_i=0$.
Then $q_{a+1}>|D|/2$ (since $a\ge i_0$) and moreover we claim that
\begin{equation}\label{eq:ratio-claim}
q_{a+1} \;\ge\; (1+\beta)\,q_a.
\end{equation}

To see \eqref{eq:ratio-claim}, note that
\[
\frac{q_I}{q_{i_0}}
=\prod_{i=i_0}^{I-1}\frac{q_{i+1}}{q_i}
\le \left(\max_{i\in\{i_0,\ldots,I-1\}} \frac{q_{i+1}}{q_i}\right)^{I-i_0}
= \left(\frac{q_{a+1}}{q_a}\right)^{I-i_0}.
\]
Rearranging gives
\[
\frac{q_{a+1}}{q_a}
\;\ge\;
\left(\frac{q_I}{q_{i_0}}\right)^{\!1/(I-i_0)}
\;\ge\;
\left(\frac{(1+\eps)^{-1}|D|}{|D|/2}\right)^{\!1/I}
\;\ge\;
\left(\frac{2}{1+\eps}\right)^{\!1/I}.
\]
Since $B_{\max}$ is polynomially bounded (lengths and capacities are polynomially bounded integers),
we have $I=\Theta(\log B_{\max}/\alpha)=O(\log n/\alpha)$.
Therefore,
\[
\left(\frac{2}{1+\eps}\right)^{1/I}
\ge
\exp\!\left(\frac{\ln(3/2)}{I}\right)
\ge
1+\Omega\!\left(\frac{1}{I}\right)
\ge
1+\Omega\!\left(\frac{\alpha}{\log n}\right).
\]
Finally, since $\eps<1/\log^9 n$, we have $\alpha=\eps^{1/3}\le 1/\log^3 n$ and thus
$\alpha/\log n \ge \eps^{2/3}=\beta$ for all sufficiently large $n$.
This implies \eqref{eq:ratio-claim}.

\paragraph{Output primal and dual.}
We set the primal flow output by $\cA'$ to be
\[
f \stackrel{\mathrm{def}}{=} f_{a+1},
\]
so $\val(f)=q_{a+1}>|D|/2$ and $\cost(f)\le B_{a+1}$.

It remains to produce a dual $(\phi,\tilde\ell)$ for an $s$--$t$ demand of value $\val(f)=q_{a+1}$
satisfying $\cost(f)\le 1.1\Dual(\phi,\tilde\ell)$.

Let
\[
\tilde\ell(e) \stackrel{\mathrm{def}}{=} h_a\cdot C_a(e)\qquad\forall e\in E,
\]
and let $\phi(\cdot)$ be the shortest-path distance from $s$ in the graph with length function
$\ell+\tilde\ell$ (equivalently, the distance function in $G-C_a$).
Then $(\phi,\tilde\ell)$ is feasible for the MCF dual, and
\[
\phi(t)-\phi(s) = \dist_{G-C_a}(s,t) \ge h_a
\]
by dual feasibility of $(C_a,h_a)$ in the MBCF LP.

For the $s$--$t$ demand of value $q_{a+1}$, the MCF dual objective equals
\[
\Dual(\phi,\tilde\ell)
= q_{a+1}\cdot (\phi(t)-\phi(s)) - \sum_{e\in E} U(e)\cdot \tilde\ell(e)
\ge q_{a+1}\cdot h_a - h_a\cdot |C_a|.
\]
Using \eqref{eq:ratio-claim} and the MBCF LP primal-dual guarantee of $\cA$ at budget $B_a$,
\[
|C_a| + \frac{B_a}{h_a} \le (1+\eps)\, q_a,
\]
we obtain
\[
\Dual(\phi,\tilde\ell)
\ge (1+\beta) q_a h_a - |C_a| h_a
\ge (1+\beta)\bigl(q_a h_a - |C_a|h_a\bigr)
\ge \frac{1+\beta}{1+\eps}\, B_a.
\]
Since $\cost(f)\le B_{a+1}=(1+\alpha)B_a$, it follows that
\[
\cost(f)
\le (1+\alpha)B_a
\le \frac{(1+\alpha)(1+\eps)}{1+\beta}\,\Dual(\phi,\tilde\ell).
\]
With $\alpha=\eps^{1/3}$ and $\beta=\eps^{2/3}$, the multiplicative factor above is
$1+O(\eps^{1/3})$, hence it is at most $1.1$ for all sufficiently large $n$
(using $\eps<1/\log^9 n$). Therefore,
\[
\cost(f)\le 1.1\,\Dual(\phi,\tilde\ell),
\]
as required. This completes the construction of $\cA'$.

\paragraph{General demands.}
For a general demand vector $D$, let $|D| \stackrel{\mathrm{def}}{=}\sum_{u: D(u)>0}D(u)$.
Add a super-source $s$ and edges $(s,u)$ of capacity $D(u)$ and length $0$ for each $u$ with
$D(u)>0$, and add a super-sink $t$ and edges $(v,t)$ of capacity $-D(v)$ and length $0$ for each
$v$ with $D(v)<0$. Then routing $D$ in $G$ is equivalent to routing $|D|$ units from $s$ to $t$
in the augmented graph; restricting any $s$--$t$ flow to the original edges gives a feasible
routing of $D$. The primal/dual objects produced above transfer to the original instance by
restriction (primal) and by restricting $(\phi,\tilde\ell)$ to the original vertices/edges.

Finally, plugging $\cA'$ into \Cref{lem:ApxflowApxcostBoosting} yields the claimed algorithm and
the stated bound on the number of calls to $\cA$.
\end{proof}
\section{Flow-Cost-Preserving DAG Projections}\label{sec:DAGprojections}
Below, we start by introducing the definition of flow-cost-preserving DAG projections, state the main construction algorithm for our DAG projection in \Cref{thm:DAGprojectionconstruction}, which will be proved later via the constructions in \Cref{sec:LCED,sec:expandertoembedding}.
Instead, we show how to use our DAG projection to reduce computing a minimum-cost flow on general graphs to solving it on DAGs.

We take the definition of DAG projections from \cite{haeupler2026dag} and extend it to preserve both distance and congestions. We call it \emph{\Proj{}}.

\begin{definition}\label{def:overviewDAGprojection}
    A \Proj{} of $G = (V,E)$ is a DAG $G'$ along with a projection map $\pi: V(G')\to V$. We call it $(\kappa,\lambda)$-\Proj{}, or say it has congestion slack $\kappa$ and length slack $\lambda$ if
\begin{enumerate}
\item For every demand $D'$ routable in $G'$ with congestion $1$ and cost $B$, $\pi(D')$ is routable in $G$ with congestion $\kappa$ and cost $B$ and,
\item For every demand $D$ routable in $G$ with congestion $1$ and cost $B$, $\flpi(D)$ is routable in $G'$ with cost $\lambda\cdot B$ and congestion $1$. 
\end{enumerate}
\end{definition}

When $h=\infty$, write \emph{$h$-length \Proj{}} simply as \emph{\Proj{}}.

Often, we want not only the existential property above, but also an explicit and efficient way to project flows and moving cuts from $D$ back to $G$. This is given by the following definitions. We also add a $h$-length bound for technical reasons.

\begin{definition}[Projection Algorithm]\label{def:projectionalgorithm}
    Let $G'$ be a DAG projection of $G = (V,E)$. A \emph{$h$-length projection algorithm} with congestion slack $\kappa$ and length slack $\lambda$ associated with $G'$ (or $(\kappa,\lambda)$-projection algorithm) is an algorithm that, given either
    \begin{itemize}
        \item a flow $F'$ on $G'$ with congestion $1$ and cost $B$, returns a flow $F$ on $G$ with congestion $\kappa$ and cost $B$ such that $\pi(\Dem(F')) = \Dem(F)$; or
        \item a moving cut $C'$ on $G'$, returns a moving cut $C$ on $G$ with the same size and length $\leng(C')/\lambda$ such that for every $u,v\in V$ with $\dist_{G-C}(u,v)\le h$, we have $\dist_{G'-C'}(\first{u},\last{v})\le \lambda\cdot \dist_{G-C}(u,v)$.
    \end{itemize}
    If $h=\sum_{e\in E}U(e)\cdot \ell(e)$ is sufficiently large, we simply say that the algorithm is a \emph{projection algorithm} without the $h$-length bound.

    When we say $D$ is associated with a projection algorithm, we assume the projection is surjective, i.e. $\pi^{-1}(v) \neq \emptyset$ for every $v \in V(G)$.
\end{definition}

Note that the second property of the projection algorithm is stronger than the second property in the definition of \Proj{}. This stronger form is by design and is necessary for our reduction algorithm later in the paper.

\newcommand{\OProj}[1]{\cO^{\mathrm{Proj}}_{#1}}
\newcommand{\OProjAlg}[1]{\cO^{\mathrm{ProjAlg}}_{#1}}
We denote the oracle of computing a flow-cost-preserving DAG projection associated with a $h$-length $(\kappa,\lambda)$-projection algorithm as $\OProj{h,\kappa,\lambda}$. With a bit of overload of notations, we also use $\OProj{h,\kappa,\lambda}$ to denote its associated $h$-length $(\kappa,\lambda)$-projection algorithm.

In \Cref{sec:expandertoembedding}, we will show the following main theorem about computing a DAG projection.

\begin{restatable}{theorem}{projectionmain}\label{thm:DAGprojectionconstruction}
    There is an algorithm given a directed graph $G=(V,E)$ with non-negative polynomially bounded integral edge capacities and lengths, parameters $\eps_f\in(0,1/\log n),\delta\in(1/\log n,1)$ and $\kappa_f\in \bbN_+$, with $n^{\delta}\cdot (1/\eps_f)^{\poly(1/\delta)}$-efficient oracle calls to $\OPDMCF{\kappa_f,1+\eps_f}$, output a DAG projection of size 
        \[m^{1+\delta}\cdot (1/\eps_f)^{\poly(1/\delta)}\]
    with a $(\kappa,\lambda)$-projection algorithm where
        \[\kappa=\kappa_f^{1/\delta}\cdot n^{\delta}\cdot (1/\eps_f)^{\poly(1/\delta)}  \qquad\qquad \lambda=1+\eps_f\cdot (1/\delta)^{O(1/\delta)}\]

    The projection algorithm takes time $m^{1+\delta}\cdot (1/\eps_f)^{\poly(1/\delta)}$. 
\end{restatable}

The power of a DAG projection above is that it allows us to reduce approximate primal-dual
minimum-cost flow on a general graph to approximate one on a DAG.
\begin{restatable}{lemma}{MCFtoDAG}\label{lem:MCFtoDAG}
Let $G$ be a directed graph with positive polynomially bounded integer edge
capacities and lengths, let $D$ be a demand vector on $G$, and let $G'$ be a
DAG equipped with a surjective projection map to $G$ and associated with a
$(\kappa',\lambda')$-projection algorithm. Then one can compute a
$(\kappa\kappa',\lambda\lambda')$-\PDMCF{} solution for $(G,D)$ using $O(1)$-efficient calls
to $\OPDMCF{\kappa,\lambda}$ and the associated projection algorithm. 
\end{restatable}
Combining with
\Cref{thm:boosting}, we can further reduce exact minimum cost flow computation from general graphs to approximate computation on DAGs.
\begin{theorem}\label{thm:MCFreductiontoDAGs}
    For $\kappa\in \bbN_+,\delta\in(1/\log n,1),\eps\in(0,1/\log n)$ with $\eps<\delta^{2/\delta}$, there is an algorithm solving the exact minimum cost flow problem with time complexity proportional to making $\kappa^{1/\delta}\cdot (1/\eps)^{\poly(1/\delta)}$-efficient calls to $\OPDMCF{\kappa,1+\eps}$.
\end{theorem}

\begin{proof}
    We wish to apply \Cref{thm:boosting}, which requires an oracle solving PD-MCF. To get that, we compute a flow-cost preserving DAG projection using \Cref{thm:DAGprojectionconstruction}, and apply \Cref{lem:MCFtoDAG} along with the $\OPDMCF{\kappa,1+\eps}$ oracle on DAGs to get the required $\OPDMCF{\kappa',1+\eps'}$ oracle on general graphs where
    \[\kappa'=\kappa^{O(1/\delta)}\cdot n^{\delta}\cdot (1/\eps)^{\poly(1/\delta)}\qquad\qquad\eps'=1+\eps\cdot (1/\delta)^{(1/\delta)}=1+O(\eps)\]

    and the oracle calls to $\OPDMCF{\kappa,1+\eps}$ are on graphs with total size $m^{1+\delta}\cdot (1/\eps)^{\poly(1/\delta)}$ (additional running time is subsumed by this). 

    Now we can apply \Cref{thm:boosting}. Notice that since $\eps<\delta^{2/\delta}$, we have that $O(1/\log(1/\eps))<\delta$. Thus, the total calls to $\OPDMCF{\kappa,1+\eps}$ on DAGs has size (after rescaling $\delta$ to hide constants in $\poly(1/\delta)$)
    \[m^{1+\delta}\cdot \kappa^{1/\delta}\cdot (1/\eps)^{\poly(1/\delta)}\]
\end{proof}

Let $c$ be a sufficiently small constant in $(0,1)$, by choosing $\delta=1/\log^{c}n$, we get the following corollary. 
\begin{corollary}\label{cor:simpleMCFreductiontoDAGs}
    For a sufficiently small constant $c\in (0,1)$ and
    \[\kappa<\exp(\log^{1-2c}n)\qquad\qquad \eps<\exp(-\log^cn\log\log n)\]
    There is an algorithm solving the exact minimum cost flow problem with time complexity proportional to making oracle calls to $\OPDMCF{\kappa,1+\eps}$ on DAGs with cumulative total size $m^{1+o(1)}$ and number of nodes $n^{1+o(1)}$.
\end{corollary}

\subsection{Reducing PD-MCF on General Graphs to DAGs}\label{subsec:MCFtoDAGs}
This subsection is dedicated for proving \Cref{lem:MCFtoDAG}. 

The next lemma isolates the two properties of a projection algorithm that we
will use later: it transfers primal flows from the DAG back to the original
graph, and it also transfers dual certificates in the reverse direction.
As an immediate consequence, any DAG equipped with such a projection algorithm
is indeed a flow-cost-preserving DAG projection.

\begin{lemma}[Projection algorithms transfer primal and dual certificates]
\label{lem:algorithmimpliesprojection}
Let $G'$ be a DAG together with a surjective projection map $\pi:V(G')\to V(G)$.
Suppose $G'$ is associated with a projection algorithm with congestion slack
$\kappa$ and length slack $\lambda$. Then the following hold.
\begin{enumerate}
    \item \textbf{Primal transfer.}
    For every flow $F'$ on $G'$ with congestion at most $1$ and cost $B$, the
    projection algorithm returns a flow $F$ on $G$ with
    \[
    \Dem(F)=\pi(\Dem(F')),
    \]
    congestion at most $\kappa$, and cost exactly $B$.

    \item \textbf{Dual transfer.}
    For every demand $D$ on $G$, and every feasible dual solution
    $(\phi,\tl)$ for the minimum-cost flow instance $(G',\flpi(D))$, one can
    construct a feasible dual solution $(\psi,C)$ for the minimum-cost flow
    instance $(G,D)$ such that
    \[
    \Dual(\psi,C)\ge \frac{1}{\lambda}\Dual(\phi,\tl).
    \]

    Consequently, $G'$ is a $(\kappa,\lambda)$-\Proj{} of $G$.
\end{enumerate}
\end{lemma}

\begin{proof}
The first item is exactly the primal part of the projection algorithm.

We now prove the second item. Fix any demand $D$ on $G$, and let
$(\phi,\tl)$ be any feasible dual solution for the instance $(G',\flpi(D))$.
Let
\[
P:=\{u\in V(G)\mid D(u)>0\},
\qquad
N:=\{v\in V(G)\mid D(v)<0\}.
\]
Since $\flpi(D)$ is supported only on first and last copies, we have
\[
\Dual(\phi,\tl)
=
\sum_{v\in N}(-D(v))\phi(\last{v})
-\sum_{u\in P}D(u)\phi(\first{u})
-\sum_{e\in E(G')}U(e)\tl(e).
\]

View $\tl$ as a moving cut $C'$ on $G'$ by setting $C'(e):=\tl(e)$ for every
$e\in E(G')$. By dual feasibility,
\[
\phi(y)-\phi(x)\le \dist_{G'-C'}(x,y)
\qquad \forall x,y\in V(G').
\]
In particular, for every $u\in P$ and $v\in N$,
\[
\phi(\last{v})-\phi(\first{u})
\le
\dist_{G'-C'}(\first{u},\last{v}).
\]

Apply the dual part of the projection algorithm to $C'$. It returns a moving
cut $C$ on $G$ with the same size,
\[
\sum_{e\in E(G)}U(e)C(e)=\sum_{e\in E(G')}U(e)C'(e)
=\sum_{e\in E(G')}U(e)\tl(e),
\]
and such that for every $u,v\in V(G)$,
\[
\dist_{G'-C'}(\first{u},\last{v})
\le
\lambda\cdot \dist_{G-C}(u,v).
\]
Hence, for every $u\in P$ and $v\in N$,
\[
\phi(\last{v})-\phi(\first{u})
\le
\lambda\cdot \dist_{G-C}(u,v).
\]

Now define
\[
\psi(x):=\min_{u\in P}
\left\{
\frac{\phi(\first{u})}{\lambda}+\dist_{G-C}(u,x)
\right\}.
\]
Since $G-C$ uses the length function $\ell+C$, the pair $(\psi,C)$ is dual
feasible for the minimum-cost flow LP on $G$.

Moreover, for every $u\in P$,
\[
\psi(u)\le \frac{\phi(\first{u})}{\lambda},
\]
while for every $v\in N$, the previous inequality implies that for every
$u\in P$,
\[
\frac{\phi(\first{u})}{\lambda}+\dist_{G-C}(u,v)
\ge
\frac{\phi(\last{v})}{\lambda},
\]
and therefore
\[
\psi(v)\ge \frac{\phi(\last{v})}{\lambda}.
\]

Using these bounds, we get
\begin{align*}
\Dual(\psi,C)
&=
\sum_{x\in V(G)}(-D(x))\psi(x)-\sum_{e\in E(G)}U(e)C(e)\\
&\ge
\frac{1}{\lambda}
\left(
\sum_{v\in N}(-D(v))\phi(\last{v})
-\sum_{u\in P}D(u)\phi(\first{u})
\right)
-\sum_{e\in E(G)}U(e)C(e)\\
&=
\frac{1}{\lambda}
\left(
\sum_{v\in N}(-D(v))\phi(\last{v})
-\sum_{u\in P}D(u)\phi(\first{u})
-\sum_{e\in E(G')}U(e)\tl(e)
\right)\\
&=
\frac{1}{\lambda}\Dual(\phi,\tl).
\end{align*}
This proves the second item.

It remains to prove the third item, namely that $G'$ is a
$(\kappa,\lambda)$-\Proj{} of $G$.

For the first property of \Cref{def:DAGprojection}, let $D'$ be any demand
routable in $G'$ with congestion $1$ and cost $B$, and let $F'$ be a
corresponding flow. By the first item, $F'$ projects to a flow $F$ on $G$ with
\[
\Dem(F)=\pi(\Dem(F'))=\pi(D'),
\]
the same cost $B$, and congestion at most $\kappa$. Hence $\pi(D')$ is
routable in $G$ with congestion $\kappa$ and cost $B$.

For the second property, let $D$ be any demand routable in $G$ with congestion
$1$ and cost $B$. Suppose for contradiction that $\flpi(D)$ is not routable in
$G'$ with congestion $1$ and cost at most $\lambda B$. By LP duality, there
exists a feasible dual solution $(\phi,\tl)$ for $(G',\flpi(D))$ with
\[
\Dual(\phi,\tl)>\lambda B.
\]
By the second item, this yields a feasible dual solution $(\psi,C)$ for
$(G,D)$ such that
\[
\Dual(\psi,C)>\frac{1}{\lambda}\cdot \lambda B=B.
\]
By weak duality, every feasible flow routing $D$ in $G$ has cost greater than
$B$, contradicting the assumption. Therefore $\flpi(D)$ is routable in $G'$
with congestion $1$ and cost at most $\lambda B$.

So $G'$ is a $(\kappa,\lambda)$-\Proj{} of $G$.
\end{proof}

Now, we are ready to prove \Cref{lem:MCFtoDAG}.

\begin{proof}[Proof of \Cref{lem:MCFtoDAG}]
Let
\[
D':=\flpi(D).
\]
Run the oracle $\OPDMCF{\kappa,\lambda}$ on the instance $(G',D')$, and let
$(\hat F,\hat\phi,\hat\tl)$ be the returned primal-dual triple. Then
$\hat F$ routes $D'$ in $G'$ with congestion at most $\kappa$, and
\[
\cost(\hat F)\le \lambda\cdot \Dual(\hat\phi,\hat\tl).
\]

By \Cref{lem:algorithmimpliesprojection}(1), projecting $\hat F$ back to $G$
gives a flow $F$ routing
\[
\pi(D')=\pi(\flpi(D))=D,
\]
with the same cost and congestion at most $\kappa\kappa'$. Hence
\[
\cost(F)=\cost(\hat F).
\]

By \Cref{lem:algorithmimpliesprojection}(2), the feasible dual solution
$(\hat\phi,\hat\tl)$ for $(G',D')$ can be transferred to a feasible dual
solution $(\psi,C)$ for $(G,D)$ satisfying
\[
\Dual(\psi,C)\ge \frac{1}{\lambda'}\Dual(\hat\phi,\hat\tl).
\]
Therefore,
\[
\cost(F)
=
\cost(\hat F)
\le
\lambda\cdot \Dual(\hat\phi,\hat\tl)
\le
\lambda\lambda'\cdot \Dual(\psi,C).
\]

Thus $(F,\psi,C)$ is a $(\kappa\kappa',\lambda\lambda')$-\PDMCF{} solution for
$(G,D)$, as claimed.
\end{proof}

\section{Length Constrained Expander Hierarchy Construction}\label{sec:LCED}

The goal of this section is to prove \Cref{lem:weakLCEDhierarchy_construct} on the construction of weak length-constrained expander hierarchy. This will be used for our DAG projection construction in the next section.
Before stating the lemma, required background and definitions are given below. 

\paragraph{Basic definitions.} The following definitions of directed length-constraint expansion are identical to \cite{HaeuplerLSW25}. We put them here for completeness. We first need to define multi-commodity demand only for the purpose of LC-expander, as follows.

A multi-commodity \emph{demand} $D:V\times V\rightarrow\mathbb{R}_{\ge0}$ assigns
a non-negative value $D(v,w) \ge 0$ to each ordered pair of vertices in $V$, specifying that vertex $v$ wants to send $D(v,w)$ units of demand to vertex $w$. We note that it is always safe to assume $D(u,u) = 0$.
The size of a demand is written as $|D|$ and is defined as $\sum_{v,w} D(v,w)$. For each vertex $v$, we let $D(v,\cdot) = \sum_{w}D(v,w)$ and $D(\cdot,v) = \sum_{w}D(w,v)$ denote the total demand starting from and ending at $v$, respectively. We say a demand $D$ is
\begin{itemize}
    \item \emph{$h$-length constrained}: if for each $v,w \in V$ with $D(v,w) > 0$, we have $\dist_G(v,w) \le h$;
    \item \emph{symmetric}: if for any pair of vertices $v,w \in V$, $D(v,w) = D(w,v)$;
    \item \emph{a sub-demand of $D'$}: if for any pair of vertices $v,w \in V$, $D(v,w) \le D'(v,w)$;
    \item \emph{$A$-respecting for some node-weighting $A$}: if for any $v\in V$, $\max\{D(v,\cdot),D(\cdot,v)\}\leq A(v)$.
\end{itemize}

We say a flow $F$ routes the demand $D$ if there exists a path decomposition of $F$ such that for every $D(v,w)\ge 0$, there is a flow path of $F$ from $v$ to $w$ with value $D(v,w)$. 

\begin{definition}[Node-Weighting]\label{def:node-weighting}
    Given a directed graph $G=(V,E)$ with edge capacities, a \emph{node-weighting} is a function $A:V\to \bbR_{\ge 0}$. Notice that $\deg_C$ is a node weighting. For a single-commodity demand $D$, we write $A_D$ as the node-weighting $A_D(v)=|D(v)|$.
\end{definition}

\begin{definition}[$h$-Length Separated Demand]
    For any multi-commodity demand $D$ and any $h$-length moving cut $C$, we define the amount of $h'$-length separated demand as the sum of demands between vertices that are $h'$-length separated by $C$. We denote this quantity with $\sep_{h'}(C,D)$, i.e.,
    \begin{align*}
        \sep_{h'}(C,D) = \sum_{u,v : \dist_{G-C}(u,v)>h'} D(u,v).
    \end{align*}
\end{definition}

\begin{definition}[$h$-Length Sparsity of a Cut $C$ for Demand $D$]\label{dfn:CDSparse}
For any multi-commodity demand $D$ and any $h$-length moving cut $C$ and $h'$ with $\sep_{h'}(C,D)>0$, the $h'$-length sparsity of $C$ with respect to $D$ is the ratio of $C$'s size to how much demand it $h'$-length separates i.e.,
\begin{align*}
    \spa_{h'}(C,D) = \frac{|C|}{\sep_{h'}(C,D)}.
\end{align*}
\end{definition}

Above we generalize the definition of length-constrained moving cut w.r.t arbitrary directed $h$-length demand. However, for the definition of a directed length-constrained expander, we restrict to symmetric $h$-length demands.

\begin{definition}[$(h,s)$-Length Sparsity of a Cut w.r.t.\ a Node-Weighting]\label{def:sparsity}
The $(h,s)$-length sparsity of any $h\cdot s$-length moving cut $C$ with respect to a node-weighting $A$ is defined as:
\begin{align*}
    \spa_{(h,s)}(C,A) = \min_{A\text{-respecting h-length symmetric demand}\ D} \spa_{h \cdot s} (C,D).
\end{align*}
\end{definition}

\paragraph{Sequence of sparse moving cuts.} We can upper bound the total size of a sequence of sparse moving cuts.

\begin{definition}[Sequence of Moving Cuts]\label{dfn:seq_moving_cuts}
Given a directed graph $G = (V, E)$, and node-weighting $A$, let $(C_1, C_2, \dots,C_n)$ be a sequence of $h \cdot s$ moving cuts, let $G - \sum_{j < i} C_j$ denote the graph that is applied with cuts from $C_1$ to $C_{i-1}$. 
We define $(C_1, C_2, \dots,C_n)$ as a sequence of $(h,s)$-length $\phi$-sparse moving cuts
if and only if the $(h,s)$-length sparsity of $C_i$ w.r.t $A$ in $G - \sum_{j < i} C_j$ is at most $\phi$.
\end{definition}

We will prove the following lemma in \Cref{subsec:unionofsparsecuts}.
\begin{restatable}{lemma}{cutsequence}\label{lem:existential_cut_sequence_different_length}
Let $A$ be a node-weighting on a graph $G$ with positive integral edge lengths and capacities.
Let $C_1,\ldots,C_k$ be a sequence of moving cuts.
Let $h_1,\ldots,h_k>1$, $\phi\in(0,1)$, $1\le \alpha\le \log n$, and
\[
s>\frac{8\log_2 n}{\alpha}.
\]
Suppose that for every \(i\in[k]\), the cut \(C_i\) is \((h_i,s)\)-length \(\phi\)-sparse with respect to \(A\) in the graph
\[
G-\sum_{j<i} C_j.
\]
Then
\[
\sum_{i=1}^k |C_i|
\le
O\!\bigl(2^{8\alpha+2}\phi\log^2 n\bigr)\cdot |A|.
\]
\end{restatable}
\paragraph{Decomposition.} We are only going to use weak expander decomposition in this paper, which is defined as follows.

\begin{definition}[Weak Length-Constraint Expander (LC-ED) Decomposition]\label{def:weakLC-ED}
    Given a graph $G$ with positive integral capacities and lengths, let $C$ be a $h$-length moving cut and $A$ be a node-weighting.
    
    We say $A$ is \emph{$(h,s)$-length $C$-constraint $\phi$-expanding} if every $A$-respecting $h$-length symmetrically constrained single-commodity demand $D$ on $G-C$ can be routed on $G$ by a flow with length bound $h\cdot s$ and congestion $1/\phi$. 

    Moreover, we say $A$ is \emph{$(\le h,s)$-length $C$-constraint $\phi$-expanding} with length slack $s$ if for every $h'\le h$, every $A$-respecting $h'$-length symmetrically constrained single-commodity demand $D$ on $G-C$ can be routed on $G$ by a flow with length bound $h'\cdot s$ and congestion $1/\phi$. 

    An \emph{associated routing algorithm} for $A$ is an algorithm with input a $A$-respecting $h$-length (or $h'$-length with $h'\le h$ in the second case) symmetrically constrained single-commodity demand $D$ on $G-C$ with a witness, output a flow routing the demand.
\end{definition}

We define a weak LC-ED hierarchy as follows. For a sequence of moving cuts $C_1,...,C_t$, we write $C_{> i}=\sum_{j> i}C_j$. 

\begin{definition}[Weak LC-ED hierarchy]\label{def:weakLC-EDhierarchy}
    Given a graph $G$ with positive integral capacities and lengths, a $\le h$-length $(\phi,s)$-expanding $t$-layer weak LC-ED hierarchy is a sequence of moving cuts $C_1,...,C_t$ such that $\sum_{i\in[t]}C_i(e)\ge1$ for all $e\in E$, and for every $h'\le h$, we interpret $C_{>i}$ as an $h'$-length moving cut, then $\deg_{C_i}$ is $\le h'$-length $C_{>i}$-constraint $(\phi,s)$-expanding.
\end{definition}

Now, we are ready to state the goal of this section.

\begin{restatable}{lemma}{weakLChierarchy}\label{lem:weakLCEDhierarchy_construct}
Let $G=(V,E)$ be a directed graph with positive integral capacities and lengths,
let $h\ge 1$, let $\eps\in(\log\log n/\log n,1),\phi\in\left(1/\poly(n),\ \frac{1}{\kappa\kappa' n^{c\eps}}\right)$ for some sufficiently large constant $c$, then there exists
$s=(\log n)^{O(1/\eps)}$
and an algorithm that outputs a $\le h$-length $(\phi,s)$-expanding
$t$-layer weak LC-expander hierarchy together with efficient associated routing algorithms, where
\[
t
=
O\!\left(
\frac{\log m}{\log(1/(\kappa\kappa' n^{O(\eps)}\phi))}
\right).
\]
The algorithm makes $\frac{m^{1+O(\eps)}}{\phi^3}$-efficient calls to $\OProj{(hs,\kappa',\lambda')}$ and $\OPDMCF{\kappa,\lambda}$ where $\lambda,\lambda'\in [1,2]$
\end{restatable}
\subsection{FatLDD for Round-Trip Distance}\label{subsec:FatLDD}

In this section, we will describe an important subroutine that extends the normal LDD on directed graphs to allow big separation between clusters, which we call FatLDD, stated in \Cref{lem:directedlddwithseparationFatLDD}.

\newcommand{\Rdist}{\widehat{\mathrm{dist}}}

The separation will be based on the symmetric round-trip distance, which is defined as follows.
\begin{definition}[Round-trip distance]\label{def:roundtrip}
    On a directed graph $G=(V,E)$ with length $\ell$, the \emph{round-trip} distance between two vertices $u,v$ is defined as
    \[\Rdist_G(u,v)=\dist_G(u,v)+\dist_G(v,u)\]
\end{definition}

We will use the existing FatLDD construction for undirected graphs on the graph $G'$.
We begin with the notions of \emph{clustering} and \emph{neighborhood cover}.

\begin{definition}[Clustering]\label{def:clustering}
Given a graph $G=(V,E)$ with distance function $\dist:V\times V\to \bbR_{\ge 0}$,
a \emph{clustering} $\cS=\{S_1,\dots,S_{|\cS|}\}$ is a collection of pairwise disjoint
vertex sets, called \emph{clusters}. We say that $\cS$ has \emph{diameter}
$h_{\diam}$ (with respect to $\dist$) if
\[
\dist(u,v)\le h_{\diam}
\qquad\text{for all }i\in[|\cS|]\text{ and all }u,v\in S_i.
\]
We say that $\cS$ has \emph{separation factor} $s$ if
\[
\dist(S_i,S_j)\stackrel{\mathrm{def}}{=}\min_{u\in S_i,\ v\in S_j}\dist(u,v)\ge s\cdot h_{\diam}
\qquad\text{for all }i\neq j.
\]

A \emph{separation witness} for separation factor $s$ is a collection of pairwise disjoint
vertex sets $S_1',\dots,S_{|\cS|}'$ such that
\[
\ball_{\dist}\!\left(S_i,\frac{s\,h_{\diam}}{2}\right)\subseteq S_i'
\qquad\text{for every }i\in[|\cS|].
\]
Intuitively, the sets $S_i'$ certify that different clusters are well separated. As usual,
one may replace the constant $1/2$ above by any smaller absolute constant by adjusting $s$
by a constant factor.
\end{definition}

\begin{definition}[Neighborhood Cover]
Given a graph $G=(V,E)$ with distance function $\dist:V\times V\to \bbR_{\ge 0}$,
a \emph{neighborhood cover} $\cN$ with \emph{width} $\omega$, \emph{covering radius}
$h_{\cov}$, \emph{diameter bound} $h_{\diam}$, and \emph{separation factor} $s$
(with respect to $\dist$) is a collection of $\omega$ clusterings
\[
\cN=\{\cS_1,\dots,\cS_\omega\},
\]
where each $\cS_i$ has diameter $h_i\le h_{\diam}$ and separation factor $s$, such that
for every vertex $v\in V$, there exists a cluster $S$ in some $\cS_i$ satisfying
\[
\ball_{\dist}(v,h_{\cov})\subseteq S.
\]
When we say that the cover comes \emph{with separation witnesses}, we mean that each
clustering $\cS_i$ is accompanied by a separation witness.
\end{definition}

We have the following lemma. The proof is deferred to \Cref{sec:FatLDD}.

\begin{restatable}{lemma}{FatLDD}\label{lem:directedlddwithseparationFatLDD}
For any $h_{\cov}>0$, $s>1$, and $\eps\in(0,1)$, and any directed graph
$G=(V,E)$ with edge lengths $\ell$, there is a randomized algorithm that computes
a neighborhood cover with respect to $\Rdist_G$, together with separation witnesses,
with
\begin{itemize}
    \item covering radius $h_{\cov}$,
    \item separation factor $s$,
    \item diameter bound
    \[
    h_{\diam}=\frac{1}{\eps}\cdot O(s\log^2 n)^{O(1/\eps)}\cdot h_{\cov},
    \]
    \item width
    \[
    \omega=n^{O(\eps)}\cdot \log n.
    \]
\end{itemize}
The algorithm runs in $\hO{m}$ time given that both $h_\diam/h_{\cov}$ and $\omega$ are in the order of $n^{o(1)}$.
\end{restatable}

\subsection{LC-Expander Decomposition}

In this section, we will show an algorithm that use cut-matching game to either find a cut that separates a large amount of demands, or certify that there is a large amount of node-weighting that is expanding. See \Cref{lem:LCexpanderdecomposition} for the main result of this section. We first provide some helpful lemmas.

The algorithm involves the following non-stop cut matching game from \cite{abs-2507-09729}.
\newcommand{\dd}{\mathbf{d}}
\newcommand{\Omatch}{\cO_{\mathrm{match}}}
\begin{definition}[Matching]\label{def:matching}
    Let $(P,Q)$ be a partition of $V$. A set of directed edges $M \subseteq P \times Q$ with capacities $U_M : M \to \bbR_{\ge 0}$ is called a \emph{$(P,Q)$-matching}. For a demand bound $\dd : V \to \bbR_{\ge 0}$ and an error $\eps \ge 0$, we say $M$ is \emph{$(\dd,\eps)$-perfect} if
    \begin{enumerate}
        \item $\vol_M \le \dd$ (i.e. $M$ does not send more than $\dd(v)$ out of any vertex), and
        \item $\vol_M(V) \ge (1 - \eps)\cdot \dd(V)$.
    \end{enumerate}
\end{definition}

\begin{lemma}[\cite{abs-2507-09729}]\label{lem:cutmatchinggame}
    Let $\dd : V \to \bbN$ be polynomially bounded and let $\eps = o(1/\log^2 n)$. Suppose we have an oracle $\Omatch$ that, for any partition $(P,Q)$ of $V$ and any $\dd' \le \dd$ with $\dd'(P) = \dd'(Q)$ and $\dd'(V) \ge \dd(V)/2$, returns a $(\dd',\eps)$-perfect $(P,Q)$-matching.

    Then there is a randomized algorithm that makes $O(\log^2 n)$ calls to $\Omatch$ and $O(m \log^2 n)$ additional work with $\tO{1}$ depth, and produces $\tilde{\dd} \le \dd$ with
    \[
        \tilde{\dd}(V) \ge \bigl(1 - O(\eps \log^2 n)\bigr) \cdot \dd(V)
    \]
    such that, with high probability, $\tilde{\dd}$ is $\Omega(1)$-expanding in the graph
    \[
        W := \bigcup_i M_i,
    \]
    where $M_i$ is the matching returned by the $i$-th call to $\Omatch$.
\end{lemma}

We also need the following lemma showing how to compute cut-match using calls to a length bounded MCF algorithm.

\begin{restatable}{lemma}{cutmatch}\label{lem:cutmatch}
    Given a directed graph $G$ with positive integral edge capacities and lengths, an $h$-length single-commodity demand $D$, and parameters $\phi,\eps\in(0,0.1)$ and $s>1$, there is an algorithm making $O\!\left(\frac{\log(1/\eps)}{\phi}\right)$-efficient oracle
    calls to $\OProj{hs,\kappa',\lambda'}$ and to $\OPDMCF{\kappa,\lambda}$, that either
    \begin{itemize}
        \item outputs an implicit flow-path decomposition of an $hs$-length flow routing some demand $D'\preceq D$ with congestion
        \[
        O\!\left(\frac{\kappa\kappa'\log(1/\eps)}{\phi}\right),
        \]
        such that
        \[
        |D'|\ge (1-\eps)\cdot |D|,
        \]
        or
        \item outputs a $\frac{hs}{4\lambda\lambda'}$-length moving cut $C$ that is $(h,\frac{s}{4\lambda\lambda'})$-length \(\phi\)-sparse with respect to \(A_D\) and has size at least
        \[
        \frac{\phi\eps}{3}\cdot |D|.
        \]
    \end{itemize}
\end{restatable}

\begin{lemma}\label{lem:LCexpanderdecomposition}
Let $G=(V,E)$ be a directed graph with positive integral capacities and integral edge lengths,
let $h\ge 1$ and $\lambda,\lambda'\in[1,2]$, let $\eps\in(\log\log n/\log n,1),\phi\in(1/\poly(n),1)$, 
and let
$A:V\to \bbR_{\ge 0}$ be a node-weighting.

Then there is a length slack 
\[s=(\log n)^{O(1/\eps)}\]
and an algorithm that makes $\frac{n^{O(\eps)}}{\phi^3}$-efficient
calls to $\OProj{(hs,\kappa',\lambda')}$ and $\OPDMCF{\kappa,\lambda}$ 
and the same additional time, and outputs a $hs$-length moving cut $C$
and a node-weighting $A'\preceq A$ such that:
\begin{enumerate}
    \item $A'$ is $(h,s)$-length $C$-constraint $\phi$-expanding with an efficient associated routing algorithm,
    \item both $|C|$ and $|A|-|A'|$ satisfy 
    \[
    |C|,|A|-|A'|\le \Delta\cdot \phi\cdot |A|,
\qquad\text{where}\qquad\Delta=\kappa\kappa'n^{O(\eps)}
    \]
\end{enumerate}
\end{lemma}
\begin{proof}
We first define the parameters used in the proof:
\[
s_{\sep}:=512\log n,\qquad
s_0:=256\log n,\qquad
\sigma:=\frac{s_0}{4\lambda\lambda'},
\]
\[
\omega:=n^{O(\eps)}\log n,\qquad
\Gamma:=\frac1\eps\cdot O(\log^3 n)^{O(1/\eps)},
\qquad
s:=c_s\Gamma\log^2 n,
\]
\[
\eps_{\rm cm}:=\frac{\phi}{100\omega\log^2 n},\qquad
\beta:=\frac{\phi}{100\omega},\qquad
L_\phi:=\log\!\frac{100\omega\log^2 n}{\phi},
\]
\[
\phi_0:=c_0\,\omega\kappa\kappa' L_\phi\cdot \phi,
\qquad
\Delta:=c_1\,\omega\kappa\kappa' L_\phi\log^2 n.
\]

Here $s_{\sep}$ is the separation factor used in the FatLDD cover, $s_0$ is the slack passed to
\Cref{lem:cutmatch}, and $\sigma$ is the resulting sparsity parameter of the cuts returned by
\Cref{lem:cutmatch}. The parameters $\omega$ and $\Gamma$ are respectively the width and diameter
blow-up of the FatLDD cover. The final slack $s$ is chosen large enough to absorb the
$O(\Gamma\log^2 n)$ routing-length blow-up. The parameter $\eps_{\rm cm}$ is the error used in the
cut-matching game, $\beta$ is the threshold for deciding whether a round is non-terminal, and
$\phi_0$ is the sparsity parameter given to \Cref{lem:cutmatch}. Finally, $\Delta$ is the loss
parameter that will upper bound both the total cut size and the total mass loss.

Since $\phi\in(1/\poly(n),1)$ and $\eps>\log\log n/\log n$, all polylogarithmic factors
are absorbed by $n^{O(\eps)}$, and thus
\[
s=(\log n)^{O(1/\eps)}
\qquad\text{and}\qquad
\Delta=\kappa\kappa' n^{O(\eps)}.
\]
\paragraph{The algorithm.}
Instead of maintaining a single moving cut during the execution, we maintain a cumulative
length-increase function
\[
\Lambda:E\to \bbR_{\ge 0},
\]
initially $\Lambda\equiv 0$.
At the beginning of each round, let $G^\circ$ denote the graph on the same vertex and edge sets as $G$,
with edge-length function
\[
\ell_{G^\circ}(e)=\ell_G(e)+\Lambda(e).
\]

Apply \Cref{lem:directedlddwithseparationFatLDD} to $G^\circ$ with parameters
\[
(h_{\cov},s,\eps)=(2h,s_{\sep},\eps),
\]
obtaining a neighborhood cover $\cN$ with respect to $\Rdist_{G^\circ}$,
with width $\omega$, and with each clustering $\cS\in\cN$ having diameter at most
\[
h_{\cS}\le \Gamma h.
\]

Fix a clustering $\cS\in\cN$ and a cluster $S\in\cS$.
Let $h_S\le \Gamma h$ be the diameter bound of $S$, and let $S'$ be its separation witness, so that
\[
\ball_{\Rdist_{G^\circ}}(S,s_{\sep}h_S/2)\subseteq S'.
\]
We run the cut-matching game of \Cref{lem:cutmatchinggame} on $S$ with demand bound
\[
\dd(v):=A(v)\qquad(v\in S).
\]

We implement the matching oracle $\Omatch$ as follows.
Given a partition $(P,Q)$ of $S$ and a vector $\dd'\le \dd$ with
\[
\dd'(P)=\dd'(Q)
\qquad\text{and}\qquad
\dd'(V)\ge \dd(V)/2,
\]
define the single-commodity demand
\[
D(v):=
\begin{cases}
\dd'(v) & v\in P,\\
-\dd'(v) & v\in Q,\\
0 & v\notin S.
\end{cases}
\]

We first prove the key geometric claim that paths of length at most $h_S\sigma$ between vertices of $S$
cannot leave $S'$ after adding any cut supported on $S'$.

\begin{claim}\label{clm:lced_short_paths_stay_inside}
Let $u,v\in S$, and let $C'$ be any moving cut supported on edges of $G^\circ[S']$.
Then every $(u,v)$-path in $G^\circ-C'$ of length at most $h_S\sigma$ is contained in $S'$.
Consequently, if
\[
\dist_{G^\circ[S']-C'}(u,v)>h_S\sigma,
\]
then also
\[
\dist_{G^\circ-C'}(u,v)>h_S\sigma.
\]
\end{claim}

\begin{proof}
Let $P$ be a $(u,v)$-path in $G^\circ-C'$ of length at most $h_S\sigma$, and let $x$ be any vertex on $P$.
Since $C'$ only adds nonnegative lengths, the prefix of $P$ from $u$ to $x$ has length at most $h_S\sigma$
already in $G^\circ$, and similarly the suffix from $x$ to $v$ has length at most $h_S\sigma$ in $G^\circ$.
Hence
\[
\dist_{G^\circ}(u,x)\le h_S\sigma
\]
and
\[
\dist_{G^\circ}(x,v)\le h_S\sigma.
\]
Since $u,v\in S$ and $S$ has round-trip diameter at most $h_S$, we also have
\[
\dist_{G^\circ}(v,u)\le \Rdist_{G^\circ}(u,v)\le h_S.
\]
Therefore
\[
\dist_{G^\circ}(x,u)\le \dist_{G^\circ}(x,v)+\dist_{G^\circ}(v,u)\le h_S\sigma+h_S,
\]
and thus
\[
\Rdist_{G^\circ}(x,u)\le h_S\sigma+(h_S\sigma+h_S)=(2\sigma+1)h_S.
\]
Because $\sigma\le 64\log n$ and $s_{\sep}=512\log n$, we have
\[
2\sigma+1 < s_{\sep}/2.
\]
Hence
\[
\Rdist_{G^\circ}(x,u) < s_{\sep}h_S/2,
\]
so
\[
x\in \ball_{\Rdist_{G^\circ}}(S,s_{\sep}h_S/2)\subseteq S'.
\]
Thus every vertex of $P$ lies in $S'$, proving the first part.

The second part follows immediately: if there were a $(u,v)$-path in $G^\circ-C'$ of length at most $h_S\sigma$,
then by the first part it would lie entirely in $S'$, contradicting
\[
\dist_{G^\circ[S']-C'}(u,v)>h_S\sigma. \qedhere
\]
\end{proof}

Now every positive vertex and every negative vertex of $D$ lie in $S$.
Since $S$ has round-trip diameter at most $h_S$, every positive vertex and every negative vertex are at directed
distance at most $h_S$ in $G^\circ$.
As $\sigma>1$, \Cref{clm:lced_short_paths_stay_inside} implies that every path of length at most $h_S$
between vertices of $S$ is contained in $S'$.
Hence $D$ is an $h_S$-length single-commodity demand on $G^\circ[S']$.

We therefore invoke \Cref{lem:cutmatch} on $G^\circ[S']$ with parameters
\[
(h,\phi,\eps,s)\leftarrow (h_S,\phi_0,\eps_{\rm cm},s_0).
\]
Since
\[
h_S s_0 \le \Gamma h \cdot 256\log n \le hs
\]
by the choice of $s$, every oracle call made by \Cref{lem:cutmatch} is a valid call to
$\OProj{(hs,\kappa',\lambda')}$.

If \Cref{lem:cutmatch} outputs an implicit flow-path decomposition
\[
F_S=(L_i,\lambda_i,s_i,t_i)_{i\in[k]},
\]
then each path goes from some vertex in $P$ to some vertex in $Q$.
We interpret the pairs $(s_i,t_i)$ with capacities $\lambda_i$ as the matching returned by $\Omatch$,
and store the corresponding path realizations in $G^\circ[S']$.

If \Cref{lem:cutmatch} instead outputs a moving cut $Q_S$, then we declare $S$ to be \emph{bad},
terminate the cut-matching game on $S$, and define
\[
\widetilde A_S(v):=0 \qquad\text{for all }v\in S.
\]

If $S$ is not bad, then \Cref{lem:cutmatchinggame} returns a node-weighting
\[
\widetilde A_S\preceq A|_S
\]
such that
\[
|\widetilde A_S|
\ge
\bigl(1-O(\eps_{\rm cm}\log^2 n)\bigr)\cdot |A|_S|,
\]
and such that the union $W_S$ of the matchings produced during the game is an $\Omega(1)$-expander
with respect to $\widetilde A_S$.

After this is done for every cluster in every clustering of $\cN$, we inspect the clusterings.
If there exists some clustering $\cS\in\cN$ such that the total $A$-mass of its bad clusters is at least
\[
\beta|A|,
\]
then for every bad cluster $S$ in this clustering we add the returned cut $Q_S$ to the cumulative
length increase, namely
\[
\Lambda(e)\gets \Lambda(e)+h_S\sigma\cdot Q_S(e)
\qquad\text{for all }e\in E,
\]
and start the next round.

Otherwise the algorithm terminates.
For every clustering $\cS_j\in\cN$, define
\[
A_j(v):=\widetilde A_{S_j(v)}(v),
\]
where $S_j(v)$ is the unique cluster of $\cS_j$ containing $v$.
Finally define
\[
A'(v):=\min_j A_j(v).
\]

\paragraph{Correctness (associated routing algorithm).}
We show that $A'$ is $(h,s)$-length $C$-constraint $\phi$-expanding and give an associated routing algorithm.

Let $D$ be any $A'$-respecting $h$-length symmetrically constrained single-commodity demand on the final graph
$G-\Lambda$, together with a witness
\[
D=D_1+\cdots+D_k,
\]
where each $D_i$ is symmetric $h$-length and the witness is cumulatively $A'$-respecting, namely
\[
\sum_{i=1}^k |D_i(v)|\le A'(v)
\qquad\text{for every }v\in V.
\]

For each witness demand $D_i$, choose any vertex $u_i$ in the support of $D_i$.
Since $D_i$ is symmetric $h$-length, for every vertex $v$ in the support of $D_i$ we have
\[
\dist_{G-\Lambda}(u_i,v)\le h
\qquad\text{and}\qquad
\dist_{G-\Lambda}(v,u_i)\le h.
\]
Hence
\[
\Rdist_{G-\Lambda}(u_i,v)\le 2h.
\]
Since the final FatLDD cover has covering radius $2h$, for every $i$ there exists a cluster
$S_i$ in some clustering $\cS_{j(i)}\in\cN$ such that
\[
\supp(D_i)\subseteq \ball_{\Rdist_{G-\Lambda}}(u_i,2h)\subseteq S_i.
\]

For each cluster $S$ appearing in this assignment, let
\[
D_S:=\sum_{i:\,S_i=S} D_i.
\]
Then the demands $D_S$ sum to $D$.
Moreover, each $D_S$ is supported on $S$, and for every $v\in S$,
\[
|D_S(v)|
=
\left|\sum_{i:\,S_i=S} D_i(v)\right|
\le
\sum_{i:\,S_i=S}|D_i(v)|
\le
\sum_{i=1}^k |D_i(v)|
\le
A'(v)
\le
\widetilde A_S(v).
\]
Thus $D_S$ is $\widetilde A_S$-respecting.

Because $S$ is not bad, the cut-matching game produced an $\Omega(1)$-expander $W_S$ with respect to
$\widetilde A_S$.
By standard routing on constant expanders, one can route $D_S$ inside $W_S$ with $O(1)$ congestion
and with paths using $O(\log n)$ expander edges.

Each edge of $W_S$ came from one successful call to \Cref{lem:cutmatch}, and is therefore realized by a path
in $G[S']$ of length at most
\[
h_S s_0\le \Gamma h\cdot 256\log n,
\]
and of congestion
\[
O\!\left(\frac{\kappa\kappa' \log(1/\eps_{\rm cm})}{\phi_0}\right)
=
O\!\left(\frac1{\omega\phi}\right)
\]
by the definition of $\phi_0$.
Replacing each expander edge by a flow path (this can be done cumulatively for all paths using
\Cref{thm:flowpathdecomposition} in nearly linear time) lifts the routing of $D_S$ from $W_S$
to the original graph.
Since the expander paths use only $O(\log n)$ edges, the lifted routing has length
\[
O(\Gamma h\log^2 n)\le hs
\]
by the choice of $s$.

Within any fixed clustering $\cS_j$, the witness sets $S'$ are pairwise disjoint, so routings inside different
clusters of $\cS_j$ do not interfere.
Across different clusterings there are only $\omega$ possibilities.
Hence after summing the lifted routings over all chosen clusters $S$, the total congestion is
\[
O(1)\cdot O\!\left(\frac1{\omega\phi}\right)\cdot \omega = O(1/\phi).
\]
After adjusting constants, this gives a routing with congestion at most $1/\phi$ and length bound $hs$.
Thus $A'$ is $(h,s)$-length $C$-constraint $\phi$-expanding.
\paragraph{Correctness (upper bound on $|A|-|A'|$).}
Fix one clustering $\cS_j$ in the final round.
By termination, the bad clusters in $\cS_j$ have total $A$-mass strictly less than $\beta|A|$.
For every good cluster $S\in\cS_j$, the cut-matching game loses only an
$O(\eps_{\rm cm}\log^2 n)$ fraction of the $A$-mass on $S$.
Therefore
\[
|A|-|A_j|
\le
\beta|A| + O(\eps_{\rm cm}\log^2 n)\cdot |A|
=
O\!\left(\frac{\phi}{\omega}\right)\cdot |A|.
\]
Since
\[
A'(v)=\min_j A_j(v),
\]
we have pointwise
\[
A(v)-A'(v)\le \sum_j (A(v)-A_j(v)).
\]
Summing over all vertices gives
\[
|A|-|A'|
\le
\sum_j (|A|-|A_j|)
\le
\omega\cdot O\!\left(\frac{\phi}{\omega}\right)\cdot |A|
=
O(\phi)\cdot |A|
\le
\Delta\phi\cdot |A|.
\]

\paragraph{Correctness (upper bound on $|C|$).}
Fix a bad cluster $S$ that appears in some round, and let $G^\circ$ denote the current graph in that round.
The invocation of \Cref{lem:cutmatch} on the induced subgraph $G^\circ[S']$ returns a moving cut $Q_S$ that is
\[
(h_S,\sigma)\text{-length }\phi_0\text{-sparse}
\]
with respect to $A_D$, where $A_D(v)=|D(v)|$ and $A_D\preceq A|_S\preceq A$.
By \Cref{def:sparsity}, there exists an $A_D$-respecting $h_S$-length symmetric demand $\widehat D$
on $G^\circ[S']$ such that
\[
\spa_{h_S\sigma}(Q_S,\widehat D;G^\circ[S'])\le \phi_0.
\]
Since $A_D\preceq A$, the demand $\widehat D$ is also $A$-respecting.
Moreover, $\widehat D$ is supported on pairs inside $S$, so by
\Cref{clm:lced_short_paths_stay_inside}, every pair that is $h_S\sigma$-separated by $Q_S$
inside $G^\circ[S']$ is also $h_S\sigma$-separated by $Q_S$ inside the whole graph $G^\circ$.
Hence
\[
\sep_{h_S\sigma}(Q_S,\widehat D;G^\circ)
\ge
\sep_{h_S\sigma}(Q_S,\widehat D;G^\circ[S']),
\]
and therefore
\[
\spa_{h_S\sigma}(Q_S,\widehat D;G^\circ)\le \phi_0.
\]
Equivalently,
\[
\spa_{(h_S,\sigma)}(Q_S,A;G^\circ)\le \phi_0.
\]
Thus every bad cut is sparse in the current whole graph with respect to the same global node-weighting $A$.

We now order all bad cuts added by the algorithm into a single sequence
\[
Q_1,Q_2,\ldots,Q_q
\]
as follows: we first order them by rounds, and within each round we order the bad clusters of the chosen
clustering arbitrarily.
For each $r$, let
\[
G_r:=G-\sum_{j<r}Q_j
\]
denote the graph obtained after applying the previous cuts in this order.

We claim that each $Q_r$ is $(h_r,\sigma)$-length $\phi_0$-sparse with respect to $A$ in $G_r$,
where $h_r$ is the diameter bound of the corresponding cluster.
This is immediate for cuts coming from different rounds, since $G_r$ is exactly the graph in which the cut is produced,
possibly with additional earlier cuts already applied.

The only point to check is when several bad cuts come from the same round.
Suppose $Q_r$ comes from a bad cluster $S_r$ with separation witness $S_r'$,
and let $Q_j$ be any earlier cut from the same round.
Then $Q_j$ comes from another bad cluster in the same clustering, so its witness set is disjoint from $S_r'$.
Since each cut is produced on the induced subgraph of its own witness set, $Q_j$ is supported on edges of $S_j'$
and therefore has disjoint support from $Q_r$.
More importantly, by \Cref{clm:lced_short_paths_stay_inside}, every path of length at most $h_r\sigma$
between vertices of $S_r$ is contained in $S_r'$.
Hence adding the earlier same-round cuts does not affect:
\begin{enumerate}
    \item the fact that the witness demand $\widehat D_r$ is $h_r$-length, and
    \item the value of $\sep_{h_r\sigma}(Q_r,\widehat D_r)$.
\end{enumerate}
Therefore the same witness demand still certifies that $Q_r$ is $(h_r,\sigma)$-length $\phi_0$-sparse in $G_r$.

We may therefore apply \Cref{lem:existential_cut_sequence_different_length} to the sequence
$Q_1,\ldots,Q_q$ with $\alpha=1$.
This yields
\[
\sum_{r=1}^q |Q_r|
\le
O(\phi_0\log^2 n)\cdot |A|.
\]
Substituting the definition of $\phi_0$ gives
\[
\sum_{r=1}^q |Q_r|
\le
O\!\left(
\omega\kappa\kappa' L_\phi\log^2 n
\right)\cdot \phi |A|
\le
\Delta\phi\cdot |A|.
\]

Let $\Lambda$ be the cumulative length increase added by all bad cuts during the algorithm, i.e.,
\[
\Lambda(e)=\sum_{r=1}^q h_r\sigma\cdot Q_r(e).
\]
We define the output cut $C$ to be the $hs$-length moving cut satisfying
\[
hs\cdot C(e)=\Lambda(e),
\qquad\text{equivalently}\qquad
C(e):=\frac{\Lambda(e)}{hs}
=\frac1{hs}\sum_{r=1}^q h_r\sigma\cdot Q_r(e).
\]
Since $h_r\le \Gamma h$ for every $r$ and $\sigma=O(\log n)$, while $s=c_s\Gamma\log^2 n$,
choosing $c_s$ sufficiently large ensures that
\[
h_r\sigma\le hs
\qquad\text{for every }r.
\]
Thus $C$ is indeed a valid $hs$-length moving cut. Finally,
\[
|C|
=
\sum_e U(e)C(e)
=
\frac1{hs}\sum_{r=1}^q h_r\sigma\cdot |Q_r|
\le
\sum_{r=1}^q |Q_r|
\le
\Delta\phi\cdot |A|.
\]
This proves the desired upper bound on $|C|$.

\paragraph{Complexity.}
We first show the number of rounds by the algorithm. Consider any non-terminal round.
By definition, there exists some clustering $\cS$ such that the bad clusters in $\cS$ have total $A$-mass at least
$\beta|A|$.
For every such bad cluster $S$, the corresponding call to \Cref{lem:cutmatch} is made on a vector $\dd'\le A|_S$
with $\dd'(V)\ge |A|_S|/2$.
Hence the returned cut has size at least
\[
\Omega(\phi_0\eps_{\rm cm}\cdot |A|_S|).
\]
Therefore the total cut added in this round has size at least
\[
\Omega(\phi_0\eps_{\rm cm}\beta)\cdot |A|.
\]
Since the total cut size over the whole execution is at most $\Delta\phi|A|$, the number of rounds is at most
\[
O\!\left(\frac{\Delta\phi}{\phi_0\eps_{\rm cm}\beta}\right)
=
O\!\left(\frac{\omega^2\log^4 n}{\phi^2}\right)
=
\frac{n^{O(\eps)}}{\phi^2}.
\]

Fix one round.
For any fixed clustering $\cS$, the witness sets $S'$ are pairwise disjoint, so the total size of all induced
subgraphs $G^\circ[S']$ over clusters $S\in\cS$ is $O(m)$.
Each cluster runs one cut-matching game, which makes $O(\log^2 n)$ calls to $\Omatch$.
Each such call invokes \Cref{lem:cutmatch}, and \Cref{lem:cutmatch} makes
\[
O\!\left(\frac{L_\phi}{\phi_0}\right)
\]
calls to $\OProj{(hs,\kappa',\lambda')}$ and $\OPDMCF{\kappa,\lambda}$.

Hence over one clustering, the total size of all oracle instances is
\[
O\!\left(
m\log^2 n\cdot \frac{L_\phi}{\phi_0}
\right),
\]
and over all $\omega$ clusterings in the cover, one round contributes total oracle-instance size
\[
O\!\left(
\omega m\log^2 n\cdot \frac{L_\phi}{\phi_0}
\right)
=
O\!\left(\frac{m\log^2 n}{\phi}\right).
\]
Multiplying by the number of rounds gives total oracle-instance size
\[
O\!\left(
\frac{\omega^2 m\log^6 n}{\phi^3}
\right)
=
\frac{m\cdot n^{O(\eps)}}{\phi^3}.
\]
Under the standard regime $n=O(m)$, this is
\[
\frac{m^{1+O(\eps)}}{\phi^3}.
\]
The same bound also dominates the additional non-oracle work, including the FatLDD construction and the bookkeeping
inside the cut-matching games.

Finally, in one round there are at most $O(\omega n\log^2 n)$ invocations of \Cref{lem:cutmatch},
and each such invocation makes $O(L_\phi/\phi_0)$ calls to the two underlying oracles.
Multiplying again by the number of rounds yields
\[
O\!\left(
\frac{\omega^2 n\log^6 n}{\phi^3}
\right)
=
\frac{n^{1+O(\eps)}}{\phi^3}
\]
total calls to $\OProj{(hs,\kappa',\lambda')}$ and $\OPDMCF{\kappa,\lambda}$.
This proves the lemma.
\end{proof}

\subsection{LC-Expander Hierarchy}\label{subsec:LCexpanderHierarchy}

In this section, we prove the following lemma that shows how to constructed a LC-expander Hierarchy as defined in \Cref{def:weakLC-EDhierarchy}.

\weakLChierarchy*
\begin{proof}
Let
\[
\cH:=\{2^0,2^1,\dots,2^{\lceil \log_2 h\rceil}\}
\]
be the set of dyadic scales up to $h$.
Since all edge lengths are positive integers and $h\le \poly(n)$, we have
\[
|\cH|=O(\log n).
\]

We will construct cuts
\[
C_1,\ldots,C_t
\]
together with auxiliary residual cuts
\[
C'_1,C'_2,\ldots,C'_t.
\]
Initially, set
\[
C'_1(e):=1\qquad\text{for all }e\in E.
\]
For every $i\ge 1$, define
\[
A_i:=\deg_{C'_i}.
\]
Note that
\[
|A_i|=2|C'_i|.
\]

We also fix
\[
\Delta_0:=\kappa\kappa'n^{O(\eps)},
\]
where the hidden constant is large enough to absorb every logarithmic loss arising below,
including the $O(\log n)$ dyadic scales. Since
\[
\phi< \frac{1}{\kappa\kappa'n^{c\eps}},
\]
for a sufficiently large constant $c$ we may assume
\[
4\Delta_0\phi\le \frac12.
\]

\paragraph{The iterative step.}
Suppose $C'_i$ has already been defined.

If
\[
|A_i|\le 1,
\]
then we terminate the construction and set
\[
t:=i,
\qquad
C_t:=C'_t.
\]
We will show below that this last layer is trivially expanding.

Otherwise, assume
\[
|A_i|>1.
\]

\smallskip\noindent
\emph{Step 1: make $A_i$ expanding simultaneously at all scales.}
For every dyadic scale $\eta\in\cH$, apply \Cref{lem:LCexpanderdecomposition} to the graph $G$,
the length bound $\eta$, the parameter $\phi$, the parameter $\eps$, and the node-weighting $A_i$.
This returns an $\eta s$-length moving cut $D_{i,\eta}$ and a node-weighting
\[
A'_{i,\eta}\preceq A_i
\]
such that:
\begin{enumerate}
    \item $A'_{i,\eta}$ is $(\eta,s)$-length $D_{i,\eta}$-constraint $\phi$-expanding;
    \item
    \[
    |D_{i,\eta}|,\ |A_i|-|A'_{i,\eta}|
    \le \Delta_0\phi |A_i|.
    \]
\end{enumerate}

We unify these guarantees over all scales by defining
\[
A'_i(v):=\min_{\eta\in\cH} A'_{i,\eta}(v)
\qquad\text{for all }v\in V,
\]
and
\[
C''_i(e):=2\max_{\eta\in\cH} D_{i,\eta}(e)
\qquad\text{for all }e\in E.
\]
Since $|\cH|=O(\log n)$ and all logarithmic factors are absorbed into $n^{O(\eps)}$,
after enlarging the hidden constant in $\Delta_0$ if necessary we still have
\[
|C''_i|,\ |A_i|-|A'_i|\le \Delta_0\phi |A_i|.
\]

We claim that $A'_i$ is $(\le h,s)$-length $C''_i$-constraint $\phi$-expanding.
Indeed, let $h'\le h$, and choose $\eta\in\cH$ such that
\[
h'\le \eta\le 2h'.
\]
If we interpret $C''_i$ as an $h'$-length moving cut, then for every edge $e$,
\[
h' C''_i(e)\ge 2h' D_{i,\eta}(e)\ge \eta D_{i,\eta}(e).
\]
Therefore every demand that is $h'$-length in $G-C''_i$ is also $\eta$-length in $G-D_{i,\eta}$.
Since $A'_i\preceq A'_{i,\eta}$, the routing algorithm for $A'_{i,\eta}$ routes every
$A'_i$-respecting $h'$-length symmetrically constrained demand with length bound at most
\[
\eta s\le 2h's.
\]
Absorbing the constant factor $2$ into the hidden constant in the definition of $s$ proves the claim.

\smallskip\noindent
\emph{Step 2: realize most of $A'_i$ by a subcut of $C'_i$.}
Let $C_i\preceq C'_i$ be a maximal subcut satisfying
\[
\deg_{C_i}\preceq A'_i.
\]
which can be found by checking each edge and increase it $C_i$ value until one of its end point get saturated in $A'_i$. Let
\[
R_i:=C'_i-C_i
\]
be the leftover cut.

Since $\deg_{C_i}\preceq A'_i$ and $A'_i$ is $(\le h,s)$-length $C''_i$-constraint $\phi$-expanding,
the same routing algorithm certifies that $\deg_{C_i}$ is also $(\le h,s)$-length
$C''_i$-constraint $\phi$-expanding.

We now bound the remainder.
By maximality of $C_i$, every edge $e$ with $R_i(e)>0$ has at least one endpoint $v$ with
\[
\deg_{C_i}(v)=A'_i(v).
\]
Charge the quantity $U(e)\cdot R_i(e)$ to one such saturated endpoint.
For any fixed vertex $v$, the total charge sent to $v$ is at most
\[
\deg_{R_i}(v)=A_i(v)-\deg_{C_i}(v)\le A_i(v)-A'_i(v).
\]
Summing over all vertices yields
\[
|R_i|\le |A_i|-|A'_i|\le \Delta_0\phi |A_i|.
\]

\smallskip\noindent
\emph{Step 3: define the next residual cut.}
Set
\[
C'_{i+1}:=C''_i+R_i.
\]
Then
\[
|C'_{i+1}|
\le
|C''_i|+|R_i|
\le
2\Delta_0\phi |A_i|
=
4\Delta_0\phi |C'_i|.
\]

\paragraph{Termination and the number of layers.}
Because $|C'_1|=\poly(n)$ and
\[
|C'_{i+1}|\le 4\Delta_0\phi |C'_i|,
\]
after
\[
t
=
O\!\left(
\frac{\log m}{\log(1/(4\Delta_0\phi))}
\right)
=
O\!\left(
\frac{\log m}{\log(1/(\kappa\kappa'n^{O(\eps)}\phi))}
\right)
\]
iterations we reach an index $t$ with
\[
|C'_t|\le \frac12,
\]
and hence
\[
|A_t|=2|C'_t|\le 1.
\]
This is exactly our stopping rule.

\paragraph{The last layer is expanding.}
We now verify that the final layer $C_t=C'_t$ is $(\le h,s)$-length $0$-constraint $\phi$-expanding.

Fix any $h'\le h$, and let $D$ be any $\deg_{C_t}$-respecting $h'$-length symmetrically constrained
single-commodity demand on $G$, with witness clustering
\[
\{V_1,\ldots,V_k\}.
\]
For each cluster $V_a$, choose a root $r_a\in V_a$.
Because $V_a$ has weak diameter at most $h'$, every vertex of $V_a$ can send flow to $r_a$
and receive flow from $r_a$ along paths of length at most $h'$.
Route the demand inside each cluster by first sending all positive mass to the root and then
sending it from the root to the negative vertices.
This produces a routing of length at most
\[
2h'\le h's
\]
since $s\ge 2$.

Moreover, since $D$ is $\deg_{C_t}$-respecting,
\[
|D|:=\sum_{v\in V}|D(v)| \le \sum_{v\in V}\deg_{C_t}(v)=|A_t|\le 1.
\]
Thus the total amount of flow routed in this construction is at most $|D|\le 1$.
As every edge has capacity at least $1$, the congestion on every edge is at most $1\le 1/\phi$.
Therefore $\deg_{C_t}$ is $(\le h,s)$-length $0$-constraint $\phi$-expanding.

\paragraph{Correctness of the hierarchy.}
We show that
\[
C_1,\ldots,C_t
\]
is a $\le h$-length $(\phi,s)$-expanding weak LC-ED hierarchy.

First, we prove the covering property
\[
\sum_{i=1}^t C_i(e)\ge 1
\qquad\text{for every }e\in E.
\]
In fact, we prove by backward induction that for every $i\in[t]$,
\[
\sum_{j\ge i} C_j \succeq C'_i.
\]
The base case $i=t$ is immediate since $C_t=C'_t$.
For the inductive step,
\[
\sum_{j\ge i} C_j
=
C_i+\sum_{j\ge i+1} C_j
\succeq
C_i+C'_{i+1}
=
C_i+C''_i+R_i
=
C'_i+C''_i
\succeq
C'_i.
\]
Applying this with $i=1$ and using $C'_1\equiv 1$ proves the required covering property.

Next, fix any $i<t$ and any $h'\le h$.
By the claim above, $\deg_{C_i}$
is $(\le h,s)$-length $C''_i$-constraint $\phi$-expanding.
Also, from the covering property just proved,
\[
C_{>i}:=\sum_{j>i} C_j \succeq C'_{i+1}\succeq C''_i.
\]
If both cuts are interpreted as $h'$-length moving cuts, then $C_{>i}$ imposes at least as much
length increase as $C''_i$.
Therefore every demand that is $h'$-length in $G-C_{>i}$ is also $h'$-length in $G-C''_i$.
So the same routing algorithm certifies that $\deg_{C_i}$ is $(\le h,s)$-length
$C_{>i}$-constraint $\phi$-expanding.

For $i=t$, the future cut is empty, and we already proved that $\deg_{C_t}$ is
$(\le h,s)$-length $0$-constraint $\phi$-expanding.
Hence $C_1,\ldots,C_t$ is a $\le h$-length $(\phi,s)$-expanding $t$-layer weak LC-ED hierarchy.

\paragraph{Associated routing algorithms.}
For every layer $i<t$, we simply use the routing algorithm inherited from the node-weighting $A'_i$
and the domination
\[
\deg_{C_i}\preceq A'_i.
\]
For the last layer $i=t$, we use the explicit root-routing algorithm described above.
All these routing algorithms are efficient.

\paragraph{Complexity.}
Each iteration performs $|\cH|=O(\log n)$ calls to \Cref{lem:LCexpanderdecomposition}.
By \Cref{lem:LCexpanderdecomposition}, one such call makes oracle calls to $\OProj{(hs,\kappa',\lambda')}$ and $\OPDMCF{\kappa,\lambda}$ on instances of total size
\[
\frac{m^{1+O(\eps)}}{\phi^3},
\]
and the same additional time.
Since
\[
t=O\!\left(
\frac{\log m}{\log(1/(\kappa\kappa'n^{O(\eps)}\phi))}
\right)
\le O(\log n),
\]
and all logarithmic factors are absorbed by $n^{O(\eps)}$, the overall complexity remains. 
\end{proof}
\section{DAG Projection Construction via LC Expander Hierarchy}\label{sec:expandertoembedding}

In this section, we will show the following lemma about constructing a $h$-length flow-cost-preserving DAG projection.
\begin{restatable}{lemma}{hlengthprojection}\label{lem:hlengthprojection}
There is an algorithm given a directed graph $G=(V,E)$ with non-negative polynomially bounded integral edge capacities and lengths, a length bound $h$, parameters $\eps,\eps_p,\eps_f\in (0,1/\log n),\delta\in(1/\log n,1)$ and $z,\kappa_p,\kappa_f\in \bbN_+$, with oracle calls to $\OProj{h/z,\kappa_p,1+\eps_p}$ and $\OPDMCF{\kappa_f,1+\eps_f}$, write
\[S(n,\delta,\eps)=n^{O(\delta)}\cdot (1/\eps)^{O(1/\delta^5)}\]

output a DAG projection of size 
\[m\cdot S(n,\delta,\eps)\]
with a $(\kappa,\lambda)$-projection algorithm where
\[\kappa=\kappa_f\kappa_p\cdot S(n,\delta,\eps)  \qquad\qquad \lambda=1+O(\eps+\eps_p+\eps_f)/\delta^2\]

The projection algorithm takes time $zm\cdot S(n,\delta,\eps)$ time. The construction algorithm makes $z\cdot S(n,\delta,\eps)$-efficient oracle calls.
\end{restatable}

\Cref{lem:hlengthprojection} will imply \Cref{thm:DAGprojectionconstruction} through recursive construction. Both of them will be proved in \Cref{subsec:combineeverything}.

\subsection{An Algorithm for Constructing the DAG Projection}
\label{subsec:algorithmforconstructingDAGprojection}

\newcommand{\AProj}{\textsc{Projection}}
\newcommand{\Aux}{\mathrm{aux}}
\newcommand{\dloc}{\mathrm{loc}}
\newcommand{\dstay}{\mathrm{stay}}
\newcommand{\ddesc}{\mathrm{desc}}
\newcommand{\dtr}{\mathrm{tr}}
\newcommand{\dout}{\mathrm{out}}
\newcommand{\dinm}{\mathrm{in}}

In this section we describe a recursive algorithm for constructing an
\(h\)-length flow-cost-preserving DAG projection.

We fix two global parameters:
\begin{itemize}
    \item a stretch parameter \(\eps\in(0,1/\log n)\), and
    \item a recursion parameter \(\delta\in(0,1)\).
\end{itemize}
Let
\[
\sigma := n^\delta.
\]
We also write
\[
s=(\log n)^{O(1/\delta)}
\qquad\text{and}\qquad
t=O(1/\delta)
\]
for the length slack and the number of layers guaranteed by
\Cref{lem:weakLCEDhierarchy_construct} for our chosen parameters.

The construction assumes oracle access to
\[
\OProj{hs,\kappa_p,1+\eps_p}
\qquad\text{and}\qquad
\OPDMCF{\kappa_f,1+\eps_f},
\]
where
\[
0\le \eps_p,\eps_f \le 0.1.
\]

Notice that $\OProj{hs,\kappa_p,1+\eps_p}$ is a stronger oracle than the $\OProj{h/z,\kappa_p,1+\eps_p}$ claimed in \Cref{lem:hlengthprojection}, but this is fine: we can use \Cref{lem:concatenating} to get the stronger oracle from the weaker one. This will be elaborated in \Cref{subsec:combineeverything}.

\paragraph{A single global weak LC-ED hierarchy.}
At the beginning of the algorithm we apply
\Cref{lem:weakLCEDhierarchy_construct} to the original graph \(G\),
with length parameter \(h\), with the lemma parameter \(\eps\) set to
\(\delta\), and with
\[
\phi:=\frac{1}{\kappa_f\kappa_p\,n^{c\delta}}
\]
for a sufficiently large constant \(c\).
This yields a \(\le h\)-length \((\phi,s)\)-expanding \(t\)-layer weak
LC-ED hierarchy
\[
C_1,\ldots,C_t.
\]

For a vertex set \(S\subseteq V\), when we write
\[
G[S]-C_{>i},
\qquad\text{where }C_{>i}:=\sum_{j>i}C_j,
\]
we mean the induced subgraph \(G[S]\) equipped with the length function
obtained by restricting the moving cut \(C_{>i}\) to the edges of \(G[S]\).

The construction is recursive. For every \(S\subseteq V\) and
\(i\in\{0,1,\ldots,t\}\), let
\[
\AProj(S,i)
\]
denote the recursive routine, and note that the top-level call is
\[
\AProj(V,t).
\]

Intuitively, \(\AProj(S,i)\) constructs a DAG projection for the induced
subproblem on \(S\), after the cut \(C_{>i}\) has already been applied.

\paragraph{Base cases.}
The recursive routine \(\AProj(S,i)\) has the following base cases.
\begin{itemize}
    \item If \(i=0\), return the edgeless graph on the vertex set \(S\).
    \item If \(|S|=1\), return the single-vertex graph.
\end{itemize}

We now describe \(\AProj(S,i)\) for \(i>0\) and \(|S|>1\).

\paragraph{Path covers at all dyadic scales.}
For each dyadic scale \(h'\le h\), let
\[
G_{S,i}^{(h')}:=G[S]-C_{>i},
\]
where \(C_{>i}\) is interpreted as an \(h'\)-length moving cut.
Let
\[
q=O(\log h)=O(\log n)
\]
denote the number of dyadic scales.

Applying \Cref{lem:pathcover} to \(G_{S,i}^{(h')}\) with appropriate parameters, we can get an induced
\[
\frac{h's}{\eps}\text{-path cover}
\]
with SCC diameter at most \(h'\), such that the width is at most $\omega=(s/\eps)\cdot \exp(O(1/\delta^3))\cdot n^{O(\delta^3)}$.

We denote this path cover by
\[
P^{(h')}.
\]

The graph \(P^{(h')}\) is not necessarily a DAG. However, every SCC of
\(P^{(h')}\) maps to an induced subgraph of \(G_{S,i}^{(h')}\), and every
such SCC has weak diameter at most \(h'\) in \(G_{S,i}^{(h')}\).
For notational convenience, if \(C\) is an SCC of \(P^{(h')}\), then we
also write \(C\) for its image inside \(G\).

\paragraph{Light and heavy SCCs.}
Let \(C\) be an SCC of \(P^{(h')}\).
\begin{itemize}
    \item If
    \[
    |V(C)|\le |S|/\sigma,
    \]
    then \(C\) is called \emph{light}, and we recursively call
    \[
    \AProj(V(C),t)
    \]
    to obtain a DAG projection \(D_C\).

    \item If
    \[
    |V(C)|>|S|/\sigma,
    \]
    then \(C\) is called \emph{heavy}, and we recursively call
    \[
    \AProj(V(C),i-1)
    \]
    to obtain a DAG projection \(D_C\).
\end{itemize}

\paragraph{Assembling one DAG layer.}
For each SCC \(C\) of \(P^{(h')}\), let \(D_C\) be the recursively
constructed DAG projection of \(C\).
Since \(D_C\) is a DAG equipped with a projection map to \(V(C)\), for
every \(u\in V(C)\) we denote by
\[
\first{u}_C
\qquad\text{and}\qquad
\last{u}_C
\]
the first and last lifted copies of \(u\) in the topological order of
\(D_C\).

We now replace each SCC \(C\) of \(P^{(h')}\) by its DAG projection \(D_C\).
For every edge \((u,v)\in E(P^{(h')})\) whose endpoints lie in two
different SCCs, say \(u\in C_u\) and \(v\in C_v\) with \(C_u\neq C_v\),
we add an edge
\[
\bigl(\last{u}_{C_u},\,\first{v}_{C_v}\bigr)
\]
with the same length and capacity as \((u,v)\).
Since the condensation of \(P^{(h')}\) is a DAG and each \(D_C\) is a DAG,
the resulting graph is again a DAG. We denote it by
\[
D^{(h')}.
\]

\paragraph{Duplicating the DAG.}
Create two disjoint copies \(D^{(h')}_1\) and \(D^{(h')}_2\) of
\(D^{(h')}\).
If \(x\in V(D^{(h')})\), then \(x_1\) and \(x_2\) denote its copies in
the two DAGs.

For each heavy SCC \(C\), create two fresh interface vertices
\[
v_C,\qquad v'_C.
\]
For every \(u\in V(C)\), add the interface edges
\[
\bigl((\last{u}_C)_1,\,v_C\bigr)
\qquad\text{and}\qquad
\bigl(v'_C,\,(\first{u}_C)_2\bigr),
\]
each of length \(h's\), and each of capacity \(\deg_{C_i}(u)\).

\paragraph{Auxiliary graph for heavy SCCs.}
Construct an auxiliary graph \(G^{\Aux,(h')}_{S,i}\) as follows.
Start from \(G[S]-C_{>i}\).
For every heavy SCC \(C\), add a fresh star vertex \(s_C\), and for every
\(u\in V(C)\) add the two zero-length edges
\[
(s_C,u)
\qquad\text{and}\qquad
(u,s_C),
\]
each with capacity \(\deg_{C_i}(u)\).

Invoke the oracle \(\OProj{hs,\kappa_p,1+\eps_p}\) on
\(G^{\Aux,(h')}_{S,i}\) and obtain a DAG projection of this auxiliary graph.

For every ordered pair of heavy SCCs \((A,B)\) and every power-of-two demand
value \(2^j\), apply \Cref{cor:PDMCFtoDAG} to the projected auxiliary graph
with an \(s_A\)-\(s_B\) demand of value \(2^j\).
For each such call, store both:
\newcommand{\pow}{\mathrm{pow}}
\begin{itemize}
    \item the returned flow \(g^{\pow}_{A,B,j}\) on \(G^{\Aux,(h')}_{S,i}\), and
    \item the returned moving cut \(C^{\pow}_{A,B,j}\) on \(G^{\Aux,(h')}_{S,i}\).
\end{itemize}

Let
\[
B_0:=\sum_{e\in E(G)} U(e)\cdot |\ell(e)|,
\qquad
R:=O(\log n/\eps).
\]
The specific value of $R$ will be set later. For every ordered pair of heavy SCCs \((A,B)\) and every \(r\in [R]\), let
\[
L_r:=(1+\eps)^{r-1}.
\]
Among all stored power-of-two flows from \(s_A\) to \(s_B\), let
\[
\mu_{A,B,r}
\]
be the largest value \(2^j\) for which the returned flow has average cost
at most \(L_r\). If no such value exists, set \(\mu_{A,B,r}:=0\). We let $R$ to be the first time $\mu_{A,B,R}$ reach the maximum value, i.e., the maximum amount of $s_A-s_B$ flow is at most $\mu_{A,B,R}$ under the power $2$ scales. 
In the case of there exists a flow with value $2^j$ and average cost at most $L_r$, denote the stored flow $g^\pow_{A,B,j}$ by $g^\star_{A,B,r}$, and the store cut $C^\pow_{A,B,j+1}$ by $C^\star_{A,B,r}$ in case $r<R$. Notice that $C^\star_{A,B,r}$ corresponds to a flow with average cost at least $L_r$, and size at most $2$ times the flow value $g^\star_{A,B,r}$. 
\paragraph{Bridge edges between heavy SCCs.}
For every ordered pair of heavy SCCs \((A,B)\) and every \(r\in [R]\) with
\(\mu_{A,B,r}>0\), add a bridge edge
\[
(v_A,v'_B)
\]
of length \(L_r\) and capacity
\[
2\mu_{A,B,r}/\eps.
\]

\paragraph{The DAG at scale \(h'\).}
Let \(\widehat D^{(h')}\) denote the DAG obtained by taking the union of
the two copies \(D^{(h')}_1,D^{(h')}_2\), all interface vertices, and all
bridge edges added above.

\paragraph{The final output and the projection map.}
The output of \(\AProj(S,i)\) is the disjoint union
\[
D^*_{S,i}:=\bigsqcup_{h'} \widehat D^{(h')}.
\]

Fix any topological order of \(D^*_{S,i}\).
For every \(u\in S\), let
\[
x_1\prec x_2\prec \cdots \prec x_k
\]
be all vertices of \(D^*_{S,i}\) whose projection is \(u\), listed in this
topological order.
Add zero-length, infinite-capacity edges
\[
(x_1,x_2),\ (x_2,x_3),\ \ldots,\ (x_{k-1},x_k).
\]

The projection map \(\pi\) on \(D^*_{S,i}\) is defined as follows.
\begin{itemize}
    \item On every recursively constructed DAG \(D_C\), use the recursively
    returned projection map.
    \item On the two copies \(D^{(h')}_1,D^{(h')}_2\), inherit the same
    projection map.
    \item For every heavy SCC \(C\), define
    \[
    \pi(v_C)=\pi(v'_C):=\bot.
    \]
\end{itemize}
\subsection{Size and Complexity of the Construction}\label{subsec:sizeandcomplexity}
We now calculate the explicit size of the output DAG and the non-oracle work of
the recursive construction.

\begin{lemma}[Size of the recursive projection]\label{lem:aproj_size}
For the top-level recursive call \(\AProj(V,t)\), the returned DAG satisfies
\[
|D^*_{V,t}|
\le m^{1+O(\delta)}\cdot (1/\eps)^{O(1/\delta^5)}\]
\end{lemma}

\begin{proof}
We will upper bound the size of the graph \(D^*_{S,i}\) returned by a general
recursive call \(\AProj(S,i)\), write down the recurrence, and then solve it.

Fix one dyadic scale \(h'\le h\).

First, the path cover \(P^{(h')}\) has size \(\tO{\omega^2 |E(G[S])|}\). After
duplicating the DAG and adding the inter-SCC edges, this remains
\(\tO{\omega^2 |E(G[S])|}\).

Second, the total number of lifted copies of original vertices in
\(P^{(h')}\) is at most \(\omega |S|\), so the interface vertices and
interface edges contribute at most \(\tO{\omega |S|}\) size.

Third, every heavy SCC contains more than \(|S|/\sigma\) original vertices,
while the total multiplicity of original vertices over all SCC images at this
scale is at most \(\omega |S|\). Hence the number of heavy SCCs at this scale
is at most \(\omega\sigma\). Therefore the bridge edges at this scale
contribute at most
\[
O((\omega\sigma)^2R).
\]

Thus the explicit graph material created in the current call at one scale is
\[
\tO{\omega^2 |E(G[S])|+\omega |S|+(\omega\sigma)^2R}.
\]
Summed over all \(q\) dyadic scales, this becomes
\[
q\cdot \tO{\omega^2 |E(G[S])|+\omega |S|+(\omega\sigma)^2R}.
\]

It remains to account for the recursive SCC gadgets. At one fixed scale,
every original vertex of \(S\) appears in at most \(\omega\) SCC images. Each light SCC \(C\) satisfies
\(|V(C)|\le |S|/\sigma\), and its recursive child is \(\AProj(V(C),t)\). Each
heavy SCC \(C\) satisfies \(|V(C)|\le |S|\), and its recursive child is
\(\AProj(V(C),i-1)\). Therefore, if we let \(\Lambda(\eta,i)\) denote the
 maximum size of \(D^*_{S,i}\) over all \(|S|\le \eta\), for \(\eta>1,i>0\),
\[
\Lambda(\eta,i)
\le
q\cdot \tO{\omega^2 m+\omega \eta+(\omega\sigma)^2R}
+
q\omega \cdot\left(
\sum_{S'\text{ is light}}\Lambda(|S'|,t)
+\Lambda(\eta,i-1)
\right).
\]

Solving this recurrence and substituting
\[
\sigma=n^\delta,\qquad
t=O(1/\delta),\qquad
q=O(\log n),\qquad
\]
\[
\omega=(s/\eps)\cdot \exp(O(1/\delta^3))\cdot n^{O(\delta^3)},\qquad
s=(\log n)^{O(1/\delta)},\qquad
R=O(\log n/\eps),
\]
gives
\[
\Lambda(n,t)
\le
m^{1+O(\delta)}\cdot
\eps^{-O(1/\delta^5)}\cdot
\]
which proves the lemma (remember that $\eps<1/\log n$, and $m\ge n$ w.o.l.g.).
\end{proof}

\begin{lemma}[Construction complexity]\label{lem:aproj_time}
The total non-oracle running time of constructing the top-level DAG
\(\AProj(V,t)\) and the cumulative graph size the oracles $\OProj{hs,\kappa_p,1+\eps_p}$, $\OPDMCF{\kappa_f,1+\eps_f}$ runs on is at most $m^{1+O(\delta)}\cdot (1/\eps)^{O(1/\delta^5)}$.
\end{lemma}

\begin{proof}
The proof is identical in form to \Cref{lem:aproj_size}. Let \(T(\eta,i)\)
denote the maximum non-oracle running time of a recursive call \(\AProj(S,i)\)
over all \(|S|\le \eta\).

For one fixed dyadic scale \(h'\), the non-oracle time in the current call
consists of constructing the path cover, computing SCCs, duplicating the DAG,
adding the interface and bridge edges, and adding the final zero-length trivial
edges. This is linear, up to polylogarithmic factors, in the explicit graph
material created in this call. Hence
\[
T(\eta,i)
\le
q\cdot \tO{\omega^2 m+\omega \eta+(\omega\sigma)^2R}
+
q\omega \eta\cdot
\left(
\sum_{S\text{ is light}}T(|S|,t)
+T(\eta,i-1)
\right)
\]

This is the same recurrence as above, so solving it and substituting the same
parameter bounds yields the lemma as claim.

Moreover, the total cumulative graph size for the oracle call $\OProj{hs,\kappa_p,1+\eps_p}$ is on the auxiliary graphs over all layers of recursion and all dyadic scales, with size $G[S]$ each, so the total size is no larger than $|D^*_{V,t}|$. The oracle call $\OPDMCF{\kappa_f,1+\eps_f}$ is also on the auxiliary graph over all layers of recursion and all dyadic scales with size $G[S]$, but repeat for all pairs of large clusters ($\omega\sigma^2$ many of them, which is at most a factor of $n^{O(\delta)}\cdot \exp(1/\delta^3)$).

\end{proof}

\subsection{Projection Algorithm for Flows}\label{subsec:projectionflows}

In this section, we prove that the construction above admits a projection
algorithm for flows.

\begin{lemma}\label{lem:flowprojection}
Let \(D^*_{V,t}\) be the DAG projection returned by \(\AProj(V,t)\).
Then there is an algorithm that, given a feasible flow \(F^D\) on
\(D^*_{V,t}\), computes a flow \(F\) on \(G\) such that
\[
\Dem(F)=\pi(\Dem(F^D)),
\qquad
\cost(F)\le \cost(F^D),
\]
and the congestion of \(F\) is at most
\[
\kappa=
\kappa_f\kappa_p\cdot n^{O(\delta)}\cdot (1/\eps)^{O(1/\delta^5)}
\]
Moreover, the running time is at most $\tO{|D^*_{V,t}|}$.
\end{lemma}

The proof is based on a recursive local reduction. The non-trivial part is
that the demands created at a recursive call \(\AProj(S,i)\) are \emph{not}
routed inside \(G[S]\). Instead, each call produces a local flow on \(G[S]\)
together with a collection of deferred demands. These deferred demands are
later routed \emph{globally on \(G\)}, layer by layer, using the weak
LC-expander hierarchy of the original graph.

\begin{definition}[Global heavy-SCC families]\label{def:globalheavyfamilies}
Consider the full recursion tree of the recursive call \(\AProj(S,i)\), and
collect all heavy SCC vertex sets generated by all recursive calls in this
subtree. We denote this family by
\[
\mathcal H_{S,i}.
\]
For every \(A\in\mathcal H_{S,i}\), let
\[
j(A)\in\{1,\ldots,t\}
\qquad\text{and}\qquad
\hat h(A)\le h
\]
denote respectively the second parameter of the recursion call involves $A$, and the dyadic scale at which the
heavy SCC \(A\) is generated.

Let $\chi_{S,i}$ be the maximum number of times a vertex $v$ being included in sets of $\mathcal H_{S,i}$.
\end{definition}

\begin{lemma}\label{lem:overlapbound_local}
For every recursive call \(\AProj(S,i)\), we have
\[
\chi_{S,i}\le (qw)^{O(t\log_\sigma |S|)}.
\]

\end{lemma}

\begin{proof}
We prove the first bound by induction on the recursion. For convenience, write
\(\Xi(|S|,i)\) for an upper bound on \(\chi_{S,i}\).

The base cases are immediate. If \(i=0\), then \(\AProj(S,0)\) returns the
edgeless graph on \(S\), so no heavy SCC is generated and \(\chi_{S,0}=0\).
If \(|S|=1\), again no heavy SCC is generated.

Now consider a non-base call \(\AProj(S,i)\) with \(|S|>1\) and \(i>0\).
For each dyadic scale \(h'\le h\), the path cover \(P^{(h')}\) has width
\(w\). Since there are \(q\)
dyadic scales, each vertex is included in at most $q\omega$ many clusters.

For every light SCC \(C\), the recursive child is
\(\AProj(V(C),t)\), and by definition of light,
\[
|V(C)|\le |S|/\sigma.
\]
Hence, by induction, the maximum duplication of vertices in the recursive tree rooted at \(C\)
is at most
\[
\Xi(\lceil |S|/\sigma\rceil,t)
\]

Similarly, for every heavy SCC \(C\), the recursive child is
\(\AProj(V(C),i-1)\). Since \(|V(C)|\le |S|\), the induction hypothesis gives
a upper bound
\[
\Xi(|S|,i-1)
\]

Therefore
\[
\chi_{S,i}
\le
qw\cdot\Bigl(1+\Xi(\lceil |S|/\sigma\rceil,t)+\Xi(|S|,i-1)\Bigr).
\]
Solving this recurrence gives
\[
\chi_{S,i}\le (qw)^{O(t\log_\sigma |S|)}.
\]
The partition can be constructed recursively together with the construction of
\(D^*_{S,i}\).
\end{proof}

We next define the recursive congestion bound for the local-flow part.

The following lemma performs the recursive local reduction. 
For each deferred demand \(D\) created from a heavy SCC \(A\), we define
the \emph{budget} as
\[
b_A(D):=\hat h(A)s\cdot |D|,
\]

\begin{lemma}[Recursive flow mapping]\label{lem:localflowreduction}
Fix a recursive call \(\AProj(S,i)\), and let \(D^*_{S,i}\) be the DAG
returned by this call. There is an algorithm that, given a feasible flow
\(F^D\) on \(D^*_{S,i}\), outputs
\begin{itemize}
    \item a flow \(F^{\dloc}\) on \(G[S]\), and
    \item for every heavy SCC \(A\in \mathcal H_{S,i}\), two deferred demands
    \(D_A^{\dout},D_A^{\dinm}\) supported on \(V(A)\),
\end{itemize}
such that:
\begin{enumerate}
    \item each \(D_A^{\dout},D_A^{\dinm}\) is \(\deg_{C_{j(A)}}\)-respecting and symmetric
    \(\hat h(A)\)-length in \(G-C_{>j(A)}\);
    \item
    \[
    \Dem(F^{\dloc})
    +
    \sum_{A\in\mathcal H_{S,i}}
    \bigl(D_A^{\dout}+D_A^{\dinm}\bigr)
    =
    \pi(\Dem(F^D));
    \]
    \item
    \[
    \cost(F^{\dloc})
    +
    \sum_{A\in\mathcal H_{S,i}}
    \bigl(b_A(D_A^{\dout})+b_A(D_A^{\dinm})\bigr)
    \le
    \cost(F^D);
    \]
    \item the congestion of \(F^{\dloc}\) is at most \(\Gamma(|S|,i)\) where 
    \[\Gamma(|S|,i)\le (\omega\sigma)^2R\kappa_f\kappa_p/\eps
\cdot
\bigl(O(\omega^2)\bigr)^{\,i+t\lceil\log_\sigma |S|\rceil}\]
\end{enumerate}
Moreover, the reduction runs in time \(\tO{|D^*_{S,i}|}\).
\end{lemma}

\begin{proof}
We induct on \((|S|,i)\) lexicographically. The base cases \(i=0\) and
\(|S|=1\) are trivial.

Write
\[
D^*_{S,i}
=
\Bigl(\bigsqcup_{\hat h\le h}\widehat D^{(\hat h)}\Bigr)
\cup E_{\mathrm{triv}},
\]
where \(E_{\mathrm{triv}}\) denotes the set of zero-length trivial edges added
between copies of the same original vertex.

For every dyadic scale \(\hat h\le h\), let \(F^{D,(\hat h)}\) be the
restriction of \(F^D\) to the edge set of \(\widehat D^{(\hat h)}\), i.e., we
simply ignore all trivial edges and keep only the edges inside the
\(\hat h\)-scale gadget.

Since the gadgets \(\widehat D^{(\hat h)}\) are edge-disjoint, we have
\[
F^D\!\mid_{E(D^*_{S,i})\setminus E_{\mathrm{triv}}}
=
\sum_{\hat h\le h} F^{D,(\hat h)}.
\]
Moreover, every trivial edge \(e=(x,y)\in E_{\mathrm{triv}}\) satisfies
\(\pi(x)=\pi(y)\), so it contributes zero to the projected demand. Therefore
\[
\sum_{\hat h\le h}\pi(\Dem(F^{D,(\hat h)}))
=
\pi(\Dem(F^D)).
\]
Also, all trivial edges have length \(0\), hence
\[
\sum_{\hat h\le h}\cost(F^{D,(\hat h)})\le \cost(F^D).
\]

We now process each scale independently and then sum the results.

\paragraph{Processing one fixed scale \(\hat h\).}
Fix one dyadic scale \(\hat h\le h\).

Let \(F^{D,(\hat h)}_{\dstay}\) be the restriction of \(F^{D,(\hat h)}\) to
the edges inside the two copies \(D^{(\hat h)}_1,D^{(\hat h)}_2\), including
the inter-SCC edges of \(D^{(\hat h)}\), but excluding interface and bridge
edges.

For every SCC \(C\) of \(P^{(\hat h)}\), and for each copy of the recursive
gadget \(D_C\), let \(F_C^D\) be the restriction of
\(F^{D,(\hat h)}_{\dstay}\) to that copy. If \(C\) is light, apply the
induction hypothesis to the recursive call \(\AProj(V(C),t)\); if \(C\) is
heavy, apply it to \(\AProj(V(C),i-1)\).

Every inter-SCC edge
\[
\bigl(\last{u}_{C_u},\first{v}_{C_v}\bigr)
\]
maps directly to the original edge \((u,v)\in E(G[S])\) with the same length
and capacity.

Summing the recursively returned local flows over all copies of all SCC
gadgets, together with the direct images of all inter-SCC edges, gives a flow
\(F^{(\hat h)}_{\dstay,\dloc}\) on \(G[S]\), and a collection of descendant
deferred demands \(\mathcal D^{(\hat h)}_{\ddesc}\), such that
\[
\Dem(F^{(\hat h)}_{\dstay,\dloc})
+
\sum_{D\in \mathcal D^{(\hat h)}_{\ddesc}} D
=
\pi(\Dem(F^{D,(\hat h)}_{\dstay})),
\]
and
\[
\cost(F^{(\hat h)}_{\dstay,\dloc})
+
\sum_{D\in \mathcal D^{(\hat h)}_{\ddesc}} b_{A(D)}(D)
\le
\cost(F^{D,(\hat h)}_{\dstay}),
\]
where \(A(D)\) denotes the heavy SCC that created the deferred demand \(D\).

Its congestion is at most
\[
2\omega^2\max\!\left\{
1,\Gamma(\lceil |S|/\sigma\rceil,t),\Gamma(|S|,i-1)
\right\},
\]
because in each of the two copies an original edge can be represented by at
most \(\omega^2\) lifted edges (each node is copied at most $\omega$ times, and the path cover is \emph{induced}).

\paragraph{Handling the interface and bridge edges.}
We next handle the transfer gadget.

For every heavy SCC \(A\) created in the current call at scale \(\hat h\), and
every \(u\in V(A)\), define
\[
p_A(u):=F^{D,(\hat h)}\bigl((\last{u}_A)_1,v_A\bigr),
\qquad
q_A(u):=F^{D,(\hat h)}\bigl(v'_A,(\first{u}_A)_2\bigr).
\]
For every ordered pair \((A,B)\) and every bridge bucket \(r\), let
\[
x_{A,B,r}:=F^{D,(\hat h)}(v_A,v'_B).
\]
Since \(F^{D,(\hat h)}\) is feasible and the bridge edge \((v_A,v'_B)\) has
capacity \(2\mu_{A,B,r}/\eps\), we have
\[
x_{A,B,r}\le 2\mu_{A,B,r}/\eps.
\]

We now scale the stored flow by the amount of bridge flow actually used:
\[
\bar g_{A,B,r}:=
\begin{cases}
\dfrac{x_{A,B,r}}{\mu_{A,B,r}}\,g^\star_{A,B,r}, & x_{A,B,r}>0,\\[2mm]
0, & x_{A,B,r}=0.
\end{cases}
\]
Then \(\bar g_{A,B,r}\) has value \(x_{A,B,r}\), average cost at most \(L_r\),
and congestion at most \(2\kappa_f\kappa_p/\eps\).

Define
\[
\alpha_{A,B,r}(u):=\bar g_{A,B,r}(s_A,u),
\qquad
\beta_{A,B,r}(v):=\bar g_{A,B,r}(v,s_B),
\]
and aggregate over all bridge types incident to \(A\):
\[
\alpha_A(u):=\sum_{B,r}\alpha_{A,B,r}(u),
\qquad
\beta_A(u):=\sum_{B,r}\beta_{B,A,r}(u).
\]
By flow conservation at the interface vertices \(v_A,v'_A\),
\[
\sum_u p_A(u)=\sum_u \alpha_A(u),
\qquad
\sum_u q_A(u)=\sum_u \beta_A(u).
\]

Define deferred demands
\[
D_A^{\dout}(u):=p_A(u)-\alpha_A(u),
\qquad
D_A^{\dinm}(u):=\beta_A(u)-q_A(u).
\]
These are single-commodity demands summing to zero, supported on \(V(A)\).
Moreover, for every \(u\in V(A)\),
\[
|D_A^{\dout}(u)|\le p_A(u)+\alpha_A(u)\le \deg_{C_i}(u),
\qquad
|D_A^{\dinm}(u)|\le q_A(u)+\beta_A(u)\le \deg_{C_i}(u),
\]
because both the interface edges and the star edges adjacent to \(u\) have
capacity \(\deg_{C_i}(u)\). Since \(A\) is an SCC of \(P^{(\hat h)}\), it has diameter at most \(\hat h\) in \(G[S]-C_{>i}\). Thus
\(D_A^{\dout},D_A^{\dinm}\) are \(\deg_{C_i}\)-respecting and
symmetric \(\hat h\)-length in \(G-C_{>i}\), witnessed by the
clustering consisting of \(V(A)\) together with singleton clusters outside
\(A\). Since \(j(A)=i\) and \(\hat h(A)=\hat h\), this proves Item~(1) for
heavy SCCs created in the current call; the corresponding claim for descendant
heavy SCCs holds by induction.

Let
\[
M^{(\hat h)}:=\sum_{A,B,r}\bar g_{A,B,r}\big|_{E(G[S])},
\]
that is, the sum of the restrictions of all scaled auxiliary flows to the
original edges of \(G[S]\). Since there are at most \(\omega\sigma\) heavy SCCs
in the current call at the fixed scale \(\hat h\), there are at most
\[
(\omega\sigma)^2R
\]
bridge types. Therefore \(M^{(\hat h)}\) has congestion at most
\[
2(\omega\sigma)^2R\kappa_f\kappa_p/\eps.
\]

Furthermore,
\[
\Dem(M^{(\hat h)})
+
\sum_{A\in\mathcal H^{\mathrm{cur}}_{S,i,\hat h}}
\bigl(D_A^{\dout}+D_A^{\dinm}\bigr)
=
\pi\bigl(\Dem(F^{D,(\hat h)}_{\dtr})\bigr),
\]
where \(F^{D,(\hat h)}_{\dtr}\) denotes the restriction of \(F^{D,(\hat h)}\)
to the interface and bridge edges, and
\[
\mathcal H^{\mathrm{cur}}_{S,i,\hat h}\subseteq \mathcal H_{S,i}
\]
denotes the heavy SCCs created in the current call at scale \(\hat h\).

For the cost, each \(\bar g_{A,B,r}\) has average cost at most \(L_r\), hence
\[
\cost(M^{(\hat h)})
\le
\sum_{A,B,r}x_{A,B,r}L_r
=
\cost\bigl(F^{D,(\hat h)}\!\mid_{\text{bridge edges}}\bigr).
\]
Also,
\[
|D_A^{\dout}|\le \sum_u p_A(u),
\qquad
|D_A^{\dinm}|\le \sum_u q_A(u),
\]
so
\[
b_A(D_A^{\dout})\le \hat h s\sum_u p_A(u),
\qquad
b_A(D_A^{\dinm})\le \hat h s\sum_u q_A(u).
\]
These are bounded by the costs of the outgoing and incoming interface-edge
families of \(A\), since every such interface edge has length \(\hat h s\).
Hence
\[
\cost(M^{(\hat h)})
+
\sum_{A\in\mathcal H^{\mathrm{cur}}_{S,i,\hat h}}
\bigl(b_A(D_A^{\dout})+b_A(D_A^{\dinm})\bigr)
\le
\cost(F^{D,(\hat h)}_{\dtr}).
\]

Now define
\[
F^{(\hat h)}_{\dloc}:=
F^{(\hat h)}_{\dstay,\dloc}+M^{(\hat h)},
\]
and
\[
\mathcal D^{(\hat h)}
:=
\mathcal D^{(\hat h)}_{\ddesc}
\cup
\{D_A^{\dout},D_A^{\dinm}\mid A\in \mathcal H^{\mathrm{cur}}_{S,i,\hat h}\}.
\]
Then
\[
\Dem(F^{(\hat h)}_{\dloc})+\sum_{D\in\mathcal D^{(\hat h)}}D
=
\pi(\Dem(F^{D,(\hat h)})),
\]
and
\[
\cost(F^{(\hat h)}_{\dloc})+\sum_{D\in\mathcal D^{(\hat h)}}b_{A(D)}(D)
\le
\cost(F^{D,(\hat h)}).
\]
Its congestion is at most
\[
\Gamma(|S|,i)\le 
2\omega^2\max\!\left\{
1,\Gamma(\lceil |S|/\sigma\rceil,i),\Gamma(|S|,i-1)
\right\}
+
2(\omega\sigma)^2R\kappa_f\kappa_p/\eps.
\]
Solving the recurrence gives the lemma.
\paragraph{Summing over all scales.}
Finally set
\[
F^{\dloc}:=\sum_{\hat h\le h}F^{(\hat h)}_{\dloc},
\]
and take the union of all deferred-demand collections over all scales. Since
there are \(q\) dyadic scales, the congestion multiplies by at most \(q\).
Therefore
\[
\cong(F^{\dloc})\le \Gamma(|S|,i)
\]
after adjusting the constant hidden in the definition of \(\Gamma\).
The demand and cost identities add over different scales, because
\[
\sum_{\hat h\le h}\pi(\Dem(F^{D,(\hat h)}))=\pi(\Dem(F^D))
\]
and
\[
\sum_{\hat h\le h}\cost(F^{D,(\hat h)})\le \cost(F^D).
\]
This proves the lemma.
\end{proof}
Now we are ready to prove \Cref{lem:flowprojection}.
\begin{proof}[Proof of \Cref{lem:flowprojection}]
Apply \Cref{lem:localflowreduction} to the top-level call \(\AProj(V,t)\) and
the input feasible flow \(F^D\). This gives a local flow \(F^{\dloc}\) on \(G\)
and, for every heavy SCC \(A\in\mathcal H_{V,t}\), two deferred demands
\(D_A^{\dout},D_A^{\dinm}\), such that
\[
\Dem(F^{\dloc})
+
\sum_{A\in\mathcal H_{V,t}}
\bigl(D_A^{\dout}+D_A^{\dinm}\bigr)
=
\pi(\Dem(F^D)),
\]
\[
\cost(F^{\dloc})
+
\sum_{A\in\mathcal H_{V,t}}
\bigl(b_A(D_A^{\dout})+b_A(D_A^{\dinm})\bigr)
\le
\cost(F^D),
\]
and
\[
\cong(F^{\dloc})\le \Gamma(n,t).
\]

For every layer \(j\in[t]\) and dyadic scale \(\hat h\le h\), define the
aggregate deferred demands
\[
\Delta^{\dout}_{j,\hat h}
:=
\sum_{\substack{A\in \mathcal H_{V,t}\\ j(A)=j,\ \hat h(A)=\hat h}}
D_A^{\dout},
\qquad
\Delta^{\dinm}_{j,\hat h}
:=
\sum_{\substack{A\in \mathcal H_{V,t}\\ j(A)=j,\ \hat h(A)=\hat h}}
D_A^{\dinm}.
\]

By Item~(1) of \Cref{lem:localflowreduction}, every
\(D_A^{\dout},D_A^{\dinm}\) in the above sum is \(\deg_{C_j}\)-respecting and
symmetric \(\hat h\)-length in \(G-C_{>j}\). Since each vertex belongs to at
most \(\chi_{V,t}\) sets in \(\mathcal H_{V,t}\), the aggregate demands
\(\Delta^{\dout}_{j,\hat h},\Delta^{\dinm}_{j,\hat h}\) are
\(\chi_{V,t}\deg_{C_j}\)-respecting and symmetric \(\hat h\)-length in
\(G-C_{>j}\).

Therefore
\[
\frac{1}{\chi_{V,t}}\Delta^{\dout}_{j,\hat h},
\qquad
\frac{1}{\chi_{V,t}}\Delta^{\dinm}_{j,\hat h}
\]
are \(\deg_{C_j}\)-respecting and symmetric \(\hat h\)-length. Since
\(\deg_{C_j}\) is \((\le h,s)\)-length \(C_{>j}\)-constraint
\((\phi,s)\)-expanding, the associated routing algorithm outputs flows
routing these scaled demands with congestion at most \(1/\phi\) and cost at most
\[
\hat h s\cdot \Bigl|\frac{1}{\chi_{V,t}}\Delta^{\dout}_{j,\hat h}\Bigr|,
\qquad
\hat h s\cdot \Bigl|\frac{1}{\chi_{V,t}}\Delta^{\dinm}_{j,\hat h}\Bigr|.
\]
Scaling the two routings back by \(\chi_{V,t}\), we obtain flows
\(R^{\dout}_{j,\hat h}\) and \(R^{\dinm}_{j,\hat h}\) routing
\(\Delta^{\dout}_{j,\hat h}\) and \(\Delta^{\dinm}_{j,\hat h}\), each with
congestion at most \(\chi_{V,t}/\phi\), and cost at most
\[
\hat h s\cdot |\Delta^{\dout}_{j,\hat h}|,
\qquad
\hat h s\cdot |\Delta^{\dinm}_{j,\hat h}|.
\]

Since
\[
|\Delta^{\dout}_{j,\hat h}|
\le
\sum_{\substack{A\in \mathcal H_{V,t}\\ j(A)=j,\ \hat h(A)=\hat h}}
|D_A^{\dout}|,
\qquad
|\Delta^{\dinm}_{j,\hat h}|
\le
\sum_{\substack{A\in \mathcal H_{V,t}\\ j(A)=j,\ \hat h(A)=\hat h}}
|D_A^{\dinm}|,
\]
we get
\[
\cost(R^{\dout}_{j,\hat h})
+
\cost(R^{\dinm}_{j,\hat h})
\le
\sum_{\substack{A\in \mathcal H_{V,t}\\ j(A)=j,\ \hat h(A)=\hat h}}
\bigl(b_A(D_A^{\dout})+b_A(D_A^{\dinm})\bigr).
\]

Now define
\[
F
:=
F^{\dloc}
+
\sum_{j,\hat h}R^{\dout}_{j,\hat h}
+
\sum_{j,\hat h}R^{\dinm}_{j,\hat h}.
\]
The deferred demands cancel exactly, so
\[
\Dem(F)=\pi(\Dem(F^D)).
\]
Also, summing the above cost bounds over all \((j,\hat h)\) and using
\Cref{lem:localflowreduction}, we get
\[
\cost(F)\le \cost(F^D).
\]

For the congestion,
\[
\cong(F)\le \Gamma(n,t)+\frac{2tq\chi_{V,t}}{\phi}.
\]
By \Cref{lem:localflowreduction},
\[
\Gamma(n,t)\le
(\omega\sigma)^2R\kappa_f\kappa_p/\eps
\cdot
\bigl(O(\omega^2)\bigr)^{\,t+t\lceil\log_\sigma n\rceil},
\]
and by \Cref{lem:overlapbound_local},
\[
\chi_{V,t}\le (qw)^{O(t\log_\sigma n)}
\]
Substituting the correct parameters yeilds the lemma.

Finally, the local reduction takes time \(\tO{|D^*_{V,t}|}\), and for each of
the \(2tq\) aggregate demands we make one call to the associated routing
algorithm. It can be verified that the associated routing algorithm takes time subsumed by \(\tO{|D^*_{V,t}|}\).
This proves the lemma.
\end{proof}
\subsection{Projection Algorithm for Moving Cuts}
\label{subsec:projectionmovingcuts}

In this subsection we prove that the recursive construction above admits a
projection algorithm for moving cuts.  

\begin{lemma}\label{lem:cutprojectionmain}
    Let $D^*_{V,t}$ be the DAG returned by $\AProj(V,t)$. There is an algorithm that, given a moving cut $C^D$ on $D^*_{V,t}$, return a moving cut $C$ on $G$ with the same size and length, such that for every $u,v\in V$ with $\dist_{G-C}(u,v)\le h$, we have $\dist_{D^*_{V,t}-C^D}(\first{u},\last{v})\le (1+\eps^*)\cdot \dist_{G-C}(u,v)$ where
    \[\eps^*=O\left(\eps+\eps_p+\eps_f\right)/\delta^2\]
\end{lemma}

It is technically more convenient to
work with \emph{length increases} instead of moving cuts.

For a graph \(H\), a \emph{length increase} on \(H\) is a nonnegative function
\(\tilde\ell:E(H)\to \bbR_{\ge 0}\). We write \(H+\tilde\ell\) for the graph
obtained by replacing the length of every edge \(e\) by
\(\ell_H(e)+\tilde\ell(e)\), and define the size of \(\tilde\ell\) by
\[
|\tilde\ell|:=\sum_{e\in E(H)} U_H(e)\cdot \tilde\ell(e).
\]

We define
\[
\Lambda_i:=(1+\eps/s)^{t-i},
\]
and fix a sufficiently large absolute constant \(c_0\).  We also define
\[
\eta(S,i)
:=
\Bigl(1+c_0(\eps+\eps_p+\eps_f)\Bigr)^{
(t+1)\lceil \log_\sigma |S|\rceil+i
}.
\]

\begin{lemma}\label{lem:cutprojection_lengthincrease}
Let \(D^*_{S,i}\) be the DAG returned by \(\AProj(S,i)\).  There is an
algorithm that, given a length increase \(\tilde\ell^D\) on \(D^*_{S,i}\),
returns a length increase \(\tilde\ell\) on \(G[S]\) such that:
\begin{enumerate}
    \item \(\displaystyle |\tilde\ell|\le |\tilde\ell^D|\),
    \item for every dyadic \(h'\le h\), if \(C_{>i}\) is interpreted as an
    \(h'\)-length moving cut and
    \[
    G_{S,i}^{(h')}:=G[S]-C_{>i},
    \]
    then for every \(u,v\in S\) satisfying
    \[
    \dist_{G_{S,i}^{(h')}+\tilde\ell}(u,v)\le
    \Lambda_i\cdot \frac{h's}{\eps},
    \]
    we have
    \[
    \dist_{D^*_{S,i}+\tilde\ell^D}(\first{u},\last{v})
    \le
    \eta(S,i)\cdot
    \dist_{G_{S,i}^{(h')}+\tilde\ell}(u,v).
    \]
\end{enumerate}
Moreover, the algorithm runs in \(\tilde O(|E(D^*_{S,i})|)\) time.
\end{lemma}

Notice that \Cref{lem:cutprojectionmain} is a directly corollary of \Cref{lem:cutprojection_lengthincrease} by interpretating $\tilde\ell,\tilde\ell^D$ as moving cuts, and set $S=V,i=t,h'=h$ (w.o.l.g. we assume $h$ is also powers of $2$) so that $G^{(h')}_{S,i}=G$.
\begin{proof}
We prove the lemma by induction on the quantity
\[
M(S,i):=(t+1)\lceil \log_\sigma |S|\rceil+i.
\]

\paragraph{Base cases.}
If \(i=0\), then \(D^*_{S,0}\) is edgeless, so there is nothing to do.
If \(|S|=1\), again there are no nontrivial edges and the claim is vacuous.

We now assume \(i>0\) and \(|S|>1\).

\paragraph{Construction of \(\tilde\ell\).}
We define \(\tilde\ell\) by summing three kinds of contributions.

\medskip
\noindent\textbf{(a) Recursive edges.}
For every SCC \(C\) appearing in every scale \(P^{(h'')}\), consider the
subgraph \(D_C\subseteq D^*_{S,i}\).
Restrict \(\tilde\ell^D\) to \(E(D_C)\), apply the inductive algorithm to
this restriction, and obtain a length increase \(\tilde\ell_C\) on \(G[V(C)]\).
We add all these contributions to \(\tilde\ell\).

\medskip
\noindent\textbf{(b) Interface edges.}
Let
\[
e=\bigl((\last{u}_C)_1,v_C\bigr)
\qquad\text{or}\qquad
e=\bigl(v'_C,(\first{u}_C)_2\bigr)
\]
be an interface edge, and let
\[
\alpha_e:=\tilde\ell^D(e).
\]
For every original edge \(f\) incident to \(u\) with \(C_i(f)>0\), add
\[
C_i(f)\cdot \alpha_e
\]
to \(\tilde\ell(f)\).

\medskip
\noindent\textbf{(c) Bridge bundles.}
Fix a scale \(h''\), and fix an ordered pair \((A,B)\) of heavy SCCs at this
scale.  Let
\[
e_r=(v_A,v'_B)
\]
denote the bridge edge of original length \(L_r\), and write
\[
\alpha_r:=\tilde\ell^D(e_r).
\]
Define
\[
\beta_{A,B}:=\min_r \bigl(L_r+\alpha_r\bigr),
\]
and let \(r^*\) be the unique index such that
\[
L_{r^*}\le \beta_{A,B}<L_{r^*+1},
\]

If $\mu_{A,B,r^*-1}=0$ then do nothing. Otherwise notice that $r^*-1<R$, so consider the stored cut \(C^\star_{A,B,r^*-1}\), and let
\[
\tau_{r^*-1}:=
\min\!\left\{
h,\,
\frac{L_{r^*-1}}{(1+\eps_p)(1+\eps_f)}
\right\}.
\]
Interpret \(C^\star_{A,B,r^*-1}\) as a length increase
\[
\tilde\ell^{\Aux}_{A,B}:=\tau_{r^*-1}\cdot C^\star_{A,B,r^*-1}
\]
on the auxiliary graph \(G^{\Aux,(h'')}_{S,i}\).

We now transfer \(\tilde\ell^{\Aux}_{A,B}\) back to \(G[S]\) as follows.
\begin{itemize}
    \item If an auxiliary edge is an original edge of \(G[S]\), we add the
    same length increase to that edge.
    \item If an auxiliary edge is a star edge \((s_A,u)\) or \((u,s_A)\),
    then for every original edge \(f\) incident to \(u\), we add
    \[
    C_i(f)\cdot \tilde\ell^{\Aux}_{A,B}(s_A,u)
    \]
    to \(\tilde\ell(f)\).
    \item We do the same for star edges adjacent to \(s_B\).
\end{itemize}

Summing all recursive, interface, and bridge-bundle contributions gives
the final length increase \(\tilde\ell\).

\paragraph{Correctness: the size bound.}
We first bound \(|\tilde\ell|\).

For the recursive parts, \(\tilde\ell\) does not apply more size than
\(\tilde\ell^D\) by the inductive hypothesis.

For an interface edge
\[
e=\bigl((\last{u}_C)_1,v_C\bigr)
\quad\text{or}\quad
e=\bigl(v'_C,(\first{u}_C)_2\bigr),
\]
the size contributed by \(e\) to \(|\tilde\ell^D|\) is
\[
U_D(e)\cdot \alpha_e
=
\deg_{C_i}(u)\cdot \alpha_e.
\]
On the other hand, the translated contribution to \(\tilde\ell\) has size
\[
\sum_{f\ni u} U(f)\cdot C_i(f)\cdot \alpha_e
=
\deg_{C_i}(u)\cdot \alpha_e.
\]
Thus each interface edge preserves size exactly.

Now consider a bridge bundle \((A,B)\) at some scale \(h''\), and suppose
\(r^*>1\). Since
\[
L_{r^*}\le \beta_{A,B}<L_{r^*+1}=(1+\eps)L_{r^*},
\]
and since \(\beta_{A,B}\le L_{r^*-1}+\alpha_{r^*-1}\), we get
\[
\alpha_{r^*-1}\ge \beta_{A,B}-L_{r^*-1}\ge \eps L_{r^*-1}.
\]
Recall that the capacity of the bridge edge \(e_{r^*-1}\) is
\[
U_D(e_{r^*-1})=2\mu_{A,B,r^*-1}/\eps.
\]
Therefore the contribution of \(e_{r^*-1}\) to \(|\tilde\ell^D|\) is at least
\[
U_D(e_{r^*-1})\cdot \alpha_{r^*-1}
\ge
(2\mu_{A,B,r^*-1}/\eps)\cdot \eps L_{r^*-1}
=
2\mu_{A,B,r^*-1}L_{r^*-1}.
\]

The translated auxiliary cut contributes size exactly
\[
\bigl|\text{translation of }\tilde\ell^{\Aux}_{A,B}\bigr|
=
|\tilde\ell^{\Aux}_{A,B}|,
\]
because original auxiliary edges preserve size trivially, and each translated
star edge preserves size by the same computation as above. Moreover,
\[
|\tilde\ell^{\Aux}_{A,B}|
=
\tau_{r^*-1}\cdot |C^\star_{A,B,r^*-1}|
\le
L_{r^*-1}\cdot 2\mu_{A,B,r^*-1},
\]

Hence every translated bridge bundle contributes no more size than was already
paid by the corresponding bridge edge in \(\tilde\ell^D\). Summing everything,
we obtain
\[
|\tilde\ell|\le |\tilde\ell^D|.
\]

\paragraph{Correctness: the length bound.}
Fix a dyadic \(h'\le h\), interpret \(C_{>i}\) as an \(h'\)-length cut, and
let
\[
d:=\dist_{G_{S,i}^{(h')}+\tilde\ell}(u,v)
\le
\Lambda_i\cdot \frac{h's}{\eps}.
\]
Let \(h''\le h'\) be the smallest dyadic scale such that
\[
d\le \Lambda_i\cdot \frac{h''s}{\eps}.
\]
Then by minimality of \(h''\),
\[
h''s\le 2\eps d.
\]

Let \(P\) be a shortest \((u,v)\)-path in \(G_{S,i}^{(h'')}+\tilde\ell\). We
will construct a path from \(\first{u}\) to \(\last{v}\) in
\(\widehat D^{(h'')}+\tilde\ell^D\) of length at most
\((1+\eta(S,i))d\).

\noindent\textbf{(a) Decomposing the lifted path.}
Since removing \(\tilde\ell\) only decreases lengths, the same path \(P\) has
length at most \(d\le \Lambda_i h''s/\eps\) in \(G_{S,i}^{(h'')}\). Hence, by
the path-cover property, \(P\) lifts to a path in \(P^{(h'')}\).

Decompose the lifted path into maximal pieces
\[
Q_1,Q_2,\ldots,Q_z
\]
such that each \(Q_j\) is contained in a single SCC of \(P^{(h'')}\).
Let \(a_j,b_j\) denote the endpoints of \(Q_j\), and let \(S_j\) denote the
SCC containing \(Q_j\).

\noindent\textbf{(b) Lifting light pieces recursively.}
If \(S_j\) is light, let
\[
d_j:=\dist_{G[V(S_j)]-C_{>i}+\tilde\ell}(a_j,b_j).
\]
Choose the smallest dyadic \(\bar h_j\in\{1,2,4,\ldots,h\}\) such that
\[
d_j\le \frac{\bar h_j s}{\eps}.
\]
This exists because \(d_j\le d\le \Lambda_i h's/\eps\le 2hs/\eps\).

Since the child call is \AProj\((V(S_j),t)\), the inductive hypothesis gives
a path in \(D_{S_j}+\tilde\ell^D\) from \((\first{a_j})_{S_j}\) to
\((\last{b_j})_{S_j}\) of length at most
\[
(1+\eta(V(S_j),t))\cdot d_j.
\]

\noindent\textbf{(c) Lifting non-critical heavy pieces recursively.}
If \(S_j\) is heavy, write
\[
\gamma_j:=C_i(Q_j):=\sum_{e\in E(Q_j)} C_i(e).
\]
Define
\[
\Gamma_j^{\leftarrow}:=\sum_{\substack{j'\le j\\ S_{j'}\text{ heavy}}}\gamma_{j'},
\qquad
\Gamma_j^{\rightarrow}:=\sum_{\substack{j'\ge j\\ S_{j'}\text{ heavy}}}\gamma_{j'}.
\]
We say \(Q_j\) is \emph{left-non-critical} if \(\Gamma_j^{\leftarrow}\le 1\),
and \emph{right-non-critical} if \(\Gamma_j^{\rightarrow}\le 1\).

Suppose first that
\[
\sum_{\substack{1\le j\le z\\ S_j\text{ heavy}}}\gamma_j\le 1.
\]
Then every heavy piece is both left- and right-non-critical. For such a piece,
passing from \(G_{S,i}^{(h'')}\) to \(G_{S,i-1}^{(h'')}\) adds at most
\(\gamma_j h''\) along \(Q_j\). Hence
\[
\dist_{G[V(S_j)]-C_{>i-1}+\tilde\ell}(a_j,b_j)
\le
\dist_{G[V(S_j)]-C_{>i}+\tilde\ell}(a_j,b_j)+\gamma_j h''.
\]
Since
\[
\dist_{G[V(S_j)]-C_{>i}+\tilde\ell}(a_j,b_j)\le d\le \Lambda_i\frac{h''s}{\eps}
\]
and \(\gamma_j\le 1\), we obtain
\[
\dist_{G[V(S_j)]-C_{>i-1}+\tilde\ell}(a_j,b_j)
\le
\Lambda_i\frac{h''s}{\eps}+h''
=
\Lambda_{i-1}\frac{h''s}{\eps}.
\]
Therefore the inductive hypothesis applies to the heavy child
\(\AProj(V(S_j),i-1)\), and yields a path in \(D_{S_j}+\tilde\ell^D\) of
length at most
\[
(1+\eta(V(S_j),i-1))\cdot
\bigl(
\dist_{G[V(S_j)]-C_{>i}+\tilde\ell}(a_j,b_j)+\gamma_j h''
\bigr).
\]
Summing over all heavy pieces, the total additive loss from the
\(\gamma_j h''\) terms is at most
\[
h''\cdot
\sum_{\substack{1\le j\le z\\ S_j\text{ heavy}}}\gamma_j
\le h''.
\]
After concatenating all recursively lifted pieces with the inter-SCC edges
inside the first copy \(D^{(h'')}_1\), we obtain a path of length at most
\[
(1+\eta(S,i))d,
\]
provided \(c_0\) is chosen large enough to absorb the extra \(O(\eps)\) term
coming from \(h''s\le 2\eps d\). So in this case we are done.

We may therefore assume
\[
\sum_{\substack{1\le j\le z\\ S_j\text{ heavy}}}\gamma_j > 1.
\]

Let \(a\) be the smallest index such that \(S_a\) is heavy and
\[
\Gamma_a^{\leftarrow}>1,
\]
and let \(b\) be the largest index such that \(S_b\) is heavy and
\[
\Gamma_b^{\rightarrow}>1.
\]
Then every heavy piece strictly before \(Q_a\) is left-non-critical, and every
heavy piece strictly after \(Q_b\) is right-non-critical. By the same argument
as above, all pieces strictly before \(Q_a\) and strictly after \(Q_b\) can be
lifted recursively, and the total additive loss from the corresponding
\(\gamma_j h''\) terms is at most
\[
2h''.
\]

\noindent\textbf{(d) Choosing the exit and entry points.}
Let
\[
\Gamma_{a^-}:=\sum_{\substack{j<a\\ S_j\text{ heavy}}}\gamma_j.
\]
Traverse the edges of \(Q_a\) from \(a_a\) to \(b_a\), and let
\[
e_L=(p_L,q_L)
\]
be the first edge of \(Q_a\) such that
\[
\Gamma_{a^-}+C_i(Q_a[a_a,p_L])+C_i(e_L)>1.
\]
Thus
\[
\Gamma_{a^-}+C_i(Q_a[a_a,p_L])\le 1.
\]

Define the left candidate set
\[
X:=\Bigl(\bigcup_{j<a,\ S_j\text{ heavy}}V(Q_j)\Bigr)\cup V(Q_a[a_a,p_L]).
\]

Symmetrically, let
\[
\Gamma_{b^+}:=\sum_{\substack{j>b\\ S_j\text{ heavy}}}\gamma_j,
\]
and let
\[
e_R=(p_R,q_R)
\]
be the last edge of \(Q_b\) such that
\[
\Gamma_{b^+}+C_i(Q_b[q_R,b_b])+C_i(e_R)>1.
\]
Thus
\[
\Gamma_{b^+}+C_i(Q_b[q_R,b_b])\le 1.
\]
Define the right candidate set
\[
Y:=V(Q_b[q_R,b_b])\cup
\Bigl(\bigcup_{j>b,\ S_j\text{ heavy}}V(Q_j)\Bigr).
\]

For every pair \(x\in X\), \(y\in Y\), let \(A:=S(x)\) and \(B:=S(y)\), and define
\[
\Psi(x,y):=
\tilde\ell^D\bigl((\last{x}_A)_1,v_A\bigr)
+\tilde\ell^{\Aux}_{A,B}(s_A,x)
+\beta_{A,B}
+\tilde\ell^{\Aux}_{A,B}(y,s_B)
+\tilde\ell^D\bigl(v'_B,(\first{y}_B)_2\bigr).
\]
Choose a pair \((x,y)\in X\times Y\) minimizing \(\Psi(x,y)\), and let
\[
A:=S(x),\qquad B:=S(y).
\]

By construction of \(X\), the total \(C_i\)-mass on the left window
\[
W_L:=
\Bigl(\bigcup_{j<a,\ S_j\text{ heavy}}Q_j\Bigr)\cup Q_a[a_a,p_L]\cup \{e_L\}
\]
is strictly larger than \(1\). Every edge of \(W_L\) receives in
\(\tilde\ell\) at least
\[
C_i(e)\cdot
\Bigl(
\tilde\ell^D\bigl((\last{x}_A)_1,v_A\bigr)
+
\tilde\ell^{\Aux}_{A,B}(s_A,x)
\Bigr)
\]
from the combination of interface and left-star translations. Hence
\[
\tilde\ell(W_L)
>
\tilde\ell^D\bigl((\last{x}_A)_1,v_A\bigr)
+
\tilde\ell^{\Aux}_{A,B}(s_A,x).
\]
Similarly, for the right window
\[
W_R:=
\{e_R\}\cup Q_b[q_R,b_b]\cup
\Bigl(\bigcup_{j>b,\ S_j\text{ heavy}}Q_j\Bigr),
\]
we have
\[
\tilde\ell(W_R)
>
\tilde\ell^{\Aux}_{A,B}(y,s_B)
+
\tilde\ell^D\bigl(v'_B,(\first{y}_B)_2\bigr).
\]
Thus the extra interface lengths and the two relevant auxiliary-star lengths
are already charged to the \(\tilde\ell\)-length of the original path.

Moreover, each interface edge has original length \(h''s\), so the two
interface edges contribute an additive
\[
2h''s\le 4\eps d.
\]

Finally, by construction of \(X\) and \(Y\), the total \(C_i\)-mass on the
heavy pieces traversed before reaching \(x\), and after leaving \(y\), is at
most \(1\). Therefore the same non-critical-heavy-piece argument as above
lifts the prefix from \(u\) to \(x\) inside the first copy \(D^{(h'')}_1\),
and the suffix from \(y\) to \(v\) inside the second copy \(D^{(h'')}_2\),
with an additional additive loss at most \(2h''\).

\noindent\textbf{(e) The bridge edge.}

By \Cref{cor:PDMCFtoDAG}, in the auxiliary graph
\[
G^{\Aux,(h'')}_{S,i}+\tilde\ell^{\Aux}_{A,B},
\]
the distance from \(s_A\) to \(s_B\) is at least
\[
\tau_{r^*-1}
=
\min\!\left\{
h,\,
\frac{L_{r^*-1}}{(1+\eps_p)(1+\eps_f)}
\right\}.
\]

Now the subpath \(P[x,y]\), together with the two star edges
\[
(s_A,x),\qquad (y,s_B),
\]
forms an \(s_A\)-\(s_B\) walk in the auxiliary graph. Therefore
\[
\tau_{r^*-1}
\le
\dist_{G_{S,i}^{(h'')}+\tilde\ell}(x,y)
+
\tilde\ell^{\Aux}_{A,B}(s_A,x)
+
\tilde\ell^{\Aux}_{A,B}(y,s_B).
\]
But the two auxiliary-star terms were already charged to \(\tilde\ell(W_L)\)
and \(\tilde\ell(W_R)\) as argued already in (d) by the choosing of entry point, respectively. Hence they create no additional
uncharged loss in the final comparison.

Since
\[
\beta_{A,B}<L_{r^*+1}=(1+\eps)L_{r^*}\le (1+\eps)^2L_{r^*-1},
\]
we conclude that
\[
\beta_{A,B}
\le
(1+\eps)^2(1+\eps_p)(1+\eps_f)\cdot
\Bigl(
\dist_{G_{S,i}^{(h'')}+\tilde\ell}(x,y)
+
\tilde\ell^{\Aux}_{A,B}(s_A,x)
+
\tilde\ell^{\Aux}_{A,B}(y,s_B)
\Bigr).
\]
After charging the two star terms to the left and right windows, the bridge
contributes only a multiplicative \(1+O(\eps+\eps_p+\eps_f)\) loss on the
skipped middle part.

\noindent\textbf{(f) Concatenating all parts.}
We now concatenate:
\begin{itemize}
    \item the recursively lifted prefix from \(u\) to \(x\) inside the first
    copy \(D^{(h'')}_1\),
    \item the interface edge \(\bigl((\last{x}_A)_1,v_A\bigr)\),
    \item the cheapest bridge \((v_A,v'_B)\),
    \item the interface edge \(\bigl(v'_B,(\first{y}_B)_2\bigr)\),
    \item the recursively lifted suffix from \(y\) to \(v\) inside the second
    copy \(D^{(h'')}_2\).
\end{itemize}
All concatenations are valid because \(D^{(h'')}_1\) and \(D^{(h'')}_2\)
contain the lifted inter-SCC edges of \(P^{(h'')}\), and the final trivial
zero-length edges in \(D^*_{S,i}\) allow us to move from the global
\(\first{u}\) to the copy of \(\first{u}\) in \(\widehat D^{(h'')}\), and
from the copy of \(\last{v}\) in \(\widehat D^{(h'')}\) to the global
\(\last{v}\).

Each recursively lifted subpiece contributes at most the factor
\[
\max\bigl\{
\eta(V(S_j),t),\,
\eta(V(S_j),i-1)
\bigr\}
\le
\frac{\eta(S,i)}{1+c_0(\eps+\eps_p+\eps_f)}.
\]
The additive losses from the non-critical heavy parts and the two interface
edges are \(O(h''+h''s)=O((\eps+\eps_p+\eps_f)d)\), and the bridge contributes
a multiplicative \(1+O(\eps+\eps_p+\eps_f)\) factor on the skipped middle
segment. Choosing \(c_0\) large enough, all these losses are absorbed into the
definition of \(\eta(S,i)\), and we obtain
\[
\dist_{D^*_{S,i}+\tilde\ell^D}(\first{u},\last{v})
\le
\eta(S,i)\cdot d.
\]
This proves the second item of the lemma.

The running time is \(\tilde O(|E(D^*_{S,i})|)\), since each recursive edge,
interface edge, and bridge bundle is processed only once up to polylogarithmic
overhead.
\end{proof}

\subsection{Combine Everything}\label{subsec:combineeverything}
We first prove the following lemma.
\begin{lemma}\label{lem:concatenating}
Let $q\in\bbN$. Given a directed graph $G$ with non-negative edge capacities and lengths, suppose $\OProj{h,\kappa,\lambda}$ computes a $h$-length flow-cost-preserving DAG projection $G'$ of $G$ together with an efficient $(\kappa,\lambda)$-projection algorithm. Then there is an algorithm that computes a $qh$-length flow-cost-preserving DAG projection $\widehat G$ of size $O(q|G'|)$ associated with an efficient $(4\kappa,\lambda)$-projection algorithm. Its running time is $O(q)$ times that of the projection algorithm for $G'$.
\end{lemma}

\begin{proof}
Let $G'$ be the DAG projection returned by $\OProj{h,\kappa,\lambda}$, with projection map $\pi:V(G')\to V(G)$. For each $i\in[2q]$, let $G^{(i)}$ be a copy of $G'$, and for each vertex $x\in V(G')$ write $x^{(i)}$ for its copy in $G^{(i)}$.

We construct $\widehat G$ as follows.
\begin{itemize}
    \item For every internal edge of every copy $G^{(i)}$, keep the same length as in $G'$, but divide its capacity by $q$.
    \item For every $i\in[2q-1]$ and every edge $(u,v)\in E(G)$, add an edge from $(\last{u})^{(i)}$ to $(\first{v})^{(i+1)}$ with length $\ell_G(u,v)$ and capacity $U_G(u,v)/q$.
    \item For every $i\in[2q-1]$ and every $u\in V(G)$, add an edge from $(\last{u})^{(i)}$ to $(\first{u})^{(i+1)}$ with length $0$ and infinite capacity.
\end{itemize}
Let $\hat\pi$ be the natural projection map on $\widehat G$, namely $\hat\pi(x^{(i)})=\pi(x)$.

Since each $G^{(i)}$ is a DAG and all added edges go from layer $i$ to layer $i+1$, the graph $\widehat G$ is also a DAG. Its size is clearly $O(q|G'|)$.

We now define a $qh$-length $(4\kappa,\lambda)$-projection algorithm for $\widehat G$.

\medskip
\noindent\textbf{Primal part.}
Let $F$ be a flow on $\widehat G$ with congestion $1$ and cost $B$. For each $i\in[2q]$, let $F_i$ be the restriction of $F$ to the internal edges of $G^{(i)}$. Also, for each $i\in[2q-1]$, let $H_i$ be the restriction of $F$ to the inter-copy edges from layer $i$ to layer $i+1$ corresponding to original edges of $G$; we ignore the zero-length infinite-capacity edges, since they project to the same vertex and thus contribute zero projected demand.

View each $F_i$ as a flow on the original graph $G'$. Because capacities in $G^{(i)}$ were scaled down by a factor of $q$, each $F_i$ has congestion at most $1/q$ on $G'$. Applying the primal part of the projection algorithm for $G'$ to each $F_i$, we obtain a flow $f_i$ on $G$ with the same cost as $F_i$, congestion at most $\kappa/q$, and demand $\Dem(f_i)=\pi(\Dem(F_i))$.

Next, map each $H_i$ directly to a flow $g_i$ on $G$ by sending exactly the same edge-flow on the corresponding edge of $G$. Since those inter-copy edges have capacities scaled down by $q$, each $g_i$ has congestion at most $1/q$ on $G$, and the same cost as $H_i$.

Now define
\[
f:=\sum_{i=1}^{2q} f_i+\sum_{i=1}^{2q-1} g_i.
\]
By construction,
\[
\Dem(f)=\hat\pi(\Dem(F)).
\]
Also, $\cost(f)=\cost(F)=B$.

Finally, the congestion of $f$ is at most
\[
\sum_{i=1}^{2q}\frac{\kappa}{q}+\sum_{i=1}^{2q-1}\frac{1}{q}
\le 2\kappa+2
\le 4\kappa,
\]
where we use $\kappa\ge 1$. Hence the primal part has congestion slack $4\kappa$.

\medskip
\noindent\textbf{Dual part.}
Let $C$ be an $L$-length moving cut on $\widehat G$. Since the zero-length inter-copy edges have infinite capacity, any finite-size moving cut assigns value $0$ to them. Thus $C$ is supported only on the internal edges of the copies and on the inter-copy edges corresponding to edges of $G$.

For each $i\in[2q]$, let $C_i$ be the restriction of $C$ to the internal edges of $G^{(i)}$. For each $i\in[2q-1]$, let $D_i$ be the restriction of $C$ to the inter-copy edges from layer $i$ to layer $i+1$ corresponding to original edges of $G$.

For each $i$, view $C_i$ as a moving cut on the original graph $G'$. Since capacities in $G^{(i)}$ were scaled down by a factor of $q$, its size in $G'$ is $q|C_i|$. Applying the dual part of the projection algorithm for $G'$ to $C_i$, we obtain an $(L/\lambda)$-length moving cut $C_i^\star$ on $G$ with
\[
|C_i^\star|=q|C_i|.
\]

For each $i\in[2q-1]$, define a cut $D_i^\star$ on $G$ by assigning to each edge $(u,v)\in E(G)$ exactly the cut value that $D_i$ assigns to the inter-copy edge from $(\last{u})^{(i)}$ to $(\first{v})^{(i+1)}$. Then $D_i^\star$ is an $(L/\lambda)$-length moving cut on $G$, and because the capacities of those inter-copy edges were scaled down by $q$, we have
\[
|D_i^\star|=q|D_i|.
\]

Now define
\[
\widehat C:=\frac{1}{q}\left(\sum_{i=1}^{2q} C_i^\star+\sum_{i=1}^{2q-1} D_i^\star\right).
\]
Then $\widehat C$ is an $(L/\lambda)$-length moving cut on $G$, and
\[
|\widehat C|
=\frac{1}{q}\left(\sum_{i=1}^{2q}|C_i^\star|+\sum_{i=1}^{2q-1}|D_i^\star|\right)
=\sum_{i=1}^{2q}|C_i|+\sum_{i=1}^{2q-1}|D_i|
=|C|.
\]

It remains to prove the distance guarantee. Fix any $u,v\in V(G)$ with $\dist_G(u,v)\le qh$, and let $P$ be a shortest path from $u$ to $v$ in $G$.

We first claim that $P$ can be written as
\[
P=Q_1\circ e_1\circ Q_2\circ e_2\circ \cdots \circ e_{r-1}\circ Q_r,
\]
where each $Q_j$ is a (possibly empty) subpath of length at most $h$, each $e_j$ is a single edge, and $r\le q$. Indeed, greedily take $Q_1$ to be the longest initial subpath of $P$ of length at most $h$; if $Q_1\neq P$, let $e_1$ be the next edge, and continue. For every $j<r$, we then have $\ell(Q_j)+\ell(e_j)>h$. Summing over $j<r$ gives
\[
(r-1)h<\sum_{j<r}\bigl(\ell(Q_j)+\ell(e_j)\bigr)\le \ell(P)\le qh,
\]
so $r\le q$.

Let $s_j,t_j$ be the endpoints of $Q_j$, so $s_1=u$, $t_r=v$, and $e_j=(t_j,s_{j+1})$ for every $j<r$.

For each offset $a\in[q]$, construct a walk $W_a$ in $\widehat G-C$ as follows:
\begin{itemize}
    \item move from $(\first{u})^{(1)}$ to $(\first{u})^{(a)}$ using only zero-length inter-copy edges and zero-length paths inside the intermediate copies;
    \item for each $j=1,\ldots,r$, traverse inside copy $G^{(a+j-1)}$ a shortest path from $(\first{s_j})^{(a+j-1)}$ to $(\last{t_j})^{(a+j-1)}$;
    \item for each $j<r$, traverse the inter-copy edge from $(\last{t_j})^{(a+j-1)}$ to $(\first{s_{j+1}})^{(a+j)}$ corresponding to the edge $e_j=(t_j,s_{j+1})$;
    \item finally move from $(\last{v})^{(a+r-1)}$ to $(\last{v})^{(2q)}$ using only zero-length inter-copy edges and zero-length paths inside the later copies.
\end{itemize}
This is well-defined because $a+r-1\le 2q-1$.

We justify the zero-length paths inside one copy. Applying the dual guarantee for the base projection to the pair $(x,x)$ gives
\[
\dist_{G^{(i)}-C_i}\bigl((\first{x})^{(i)},(\last{x})^{(i)}\bigr)
\le \lambda\cdot \dist_{G-C_i^\star}(x,x)=0.
\]
Hence such a path exists and has length $0$.

For each $j$, since $\dist_G(s_j,t_j)\le \ell(Q_j)\le h$, the dual guarantee for the base projection yields
\[
\dist_{G^{(a+j-1)}-C_{a+j-1}}
\bigl((\first{s_j})^{(a+j-1)},(\last{t_j})^{(a+j-1)}\bigr)
\le
\lambda\cdot \dist_{G-C_{a+j-1}^\star}(s_j,t_j)
\le
\lambda\cdot \ell_{G-C_{a+j-1}^\star}(Q_j).
\]
Also, for each $j<r$, the inter-copy edge corresponding to $e_j$ has length
\[
\ell_G(e_j)+L\cdot D_{a+j-1}(e_j)
\le
\lambda\left(\ell_G(e_j)+\frac{L}{\lambda}D_{a+j-1}^\star(e_j)\right)
=
\lambda\cdot \ell_{G-D_{a+j-1}^\star}(e_j),
\]
where we use $\lambda\ge 1$ and $\ell_G(e_j)\ge 0$.

Therefore,
\[
\ell_{\widehat G-C}(W_a)
\le
\lambda\sum_{j=1}^r \ell_{G-C_{a+j-1}^\star}(Q_j)
+\lambda\sum_{j=1}^{r-1}\ell_{G-D_{a+j-1}^\star}(e_j).
\]

Now average over all $a\in[q]$. We get
\begin{align*}
\frac{1}{q}\sum_{a=1}^q \ell_{\widehat G-C}(W_a)
&\le
\lambda\sum_{j=1}^r\left(\ell_G(Q_j)+\frac{L}{\lambda}\cdot \frac{1}{q}\sum_{a=1}^q C_{a+j-1}^\star(Q_j)\right)\\
&\quad+\lambda\sum_{j=1}^{r-1}\left(\ell_G(e_j)+\frac{L}{\lambda}\cdot \frac{1}{q}\sum_{a=1}^q D_{a+j-1}^\star(e_j)\right)\\
&\le
\lambda\cdot \ell_{G-\widehat C}(P).
\end{align*}
Hence for some $a\in[q]$,
\[
\ell_{\widehat G-C}(W_a)\le \lambda\cdot \ell_{G-\widehat C}(P).
\]
Since $W_a$ is a walk from $(\first{u})^{(1)}$ to $(\last{v})^{(2q)}$, this implies
\[
\dist_{\widehat G-C}\bigl((\first{u})^{(1)},(\last{v})^{(2q)}\bigr)
\le \lambda\cdot \ell_{G-\widehat C}(P)
= \lambda\cdot \dist_{G-\widehat C}(u,v).
\]

Thus the dual part also holds. Therefore $\widehat G$ is associated with a $qh$-length $(4\kappa,\lambda)$-projection algorithm. By \Cref{lem:algorithmimpliesprojection}, $\widehat G$ is a $qh$-length flow-cost-preserving DAG projection of $G$.
\end{proof}
Now we are ready to prove the two lemmas stated at the begining of this seciton.
\hlengthprojection*
\begin{proof}
    We call the algorithm stated in \Cref{subsec:algorithmforconstructingDAGprojection} and apply \Cref{lem:aproj_size,lem:aproj_time,lem:flowprojection,lem:cutprojectionmain}. There are two things to notice in order to get \Cref{lem:hlengthprojection}.
    \begin{itemize}
        \item \Cref{subsec:algorithmforconstructingDAGprojection} makes oracle calls to $\OProj{hs,\kappa_p,1+\eps_p}$, i.e., a $hs$-length algorithm but not $h/z$. However, we can use \Cref{lem:concatenating} to get the oracle by concatenating $zs$ copies of the projection returned by $\OProj{h/z,\kappa_p,1+\eps_p}$, which only increase all the time complexity part by a factor of $z$ (the $s=(\log n)^{1/\delta}$ factor is absorbed in $S(n,\delta,\eps)$).
        \item \Cref{lem:cutprojectionmain} is slightly different from the definition \Cref{def:DAGprojection}: the cut $C^D,C$ has the same length while \Cref{def:DAGprojection} has a different length (differ by the length slack $\lambda$ of the projection). This can be easily solved by multiplying all the capacities of the edges in the DAG by $\lambda$, so the cut $C^D$ will increase its length by a factor $\lambda$. Then we can intepretate $C$ as a $\leng(C^D)/\lambda$ cut, the size keep the same and the congestion only increase by a factor of at most $O(1/\delta^2)$. 
    \end{itemize}
\end{proof}

\projectionmain*
\begin{proof}
    Notice that when $h$ is sufficiently small (like $\polylog(n)$), a $h$-length DAG projection is trivial to get according to \Cref{lem:concatenating}.

    \Cref{lem:hlengthprojection} shows how to enlarge the length by a factor of $q$ with oracle $\OPDMCF{\kappa_f,1+\eps_f}$ and
    \begin{itemize}
        \item enlarging the congestion slack by a factor of $\kappa_f\cdot n^{O(\delta)}\cdot (1/\eps)^{O(1/\delta^5)}$,
        \item enlarging the length slack minus $1$ by a factor of $O(1/\delta^2)$ if we always choose $\eps=\eps_f$ and $\eps_f<\eps_p$,
        \item with a fix complexity $qm\cdot S(n,\delta,\eps)$ and fixed size $m\cdot S(n,\delta,\eps)$.
    \end{itemize}
    Thus, we can choose the parameters $q$ to be $n^{\sqrt{\delta}}$, then up to $O(1/\sqrt{\delta})$ iterations, we get sufficient large length bound $h$, and the congestion becomes (redefine $\delta=\sqrt{\delta}$) 
    \[\kappa_f^{O(1/\delta)}\cdot n^{O(\delta)}\cdot (1/\eps_f)^{O(1/\delta^{11})}\]
    the length slack becomes
    \[1+\eps_f\cdot (1/\delta)^{O(1/\delta)}\]

    The complexity also matches.
\end{proof}
\section{Minimum Cost Flow Algorithm on DAGs}\label{sec:dag-mcf}

In this section we give an algorithm that returns an approximate primal--dual solution for minimum-cost $s$--$t$ flow in a DAG. The algorithm is inspired by the weighted push--relabel framework for max flow~\cite{BernsteinBST24}.

\begin{restatable}{theorem}{MCFonDAGs}\label{thm:dag-oracle}
Given a DAG $G=(V,E)$ with source $s$, sink $t$, and a parameter $\eps\in(0,0.1)$, suppose there exists an $s$--$t$ flow $f^*$ on $G$.
Then there is an algorithm that computes a $(1\pm 6\eps)$ primal--dual minimum-cost flow pair $(f,\phi,C)$ with value
$\val(f)\ge \val(f^*)/3$ in time $\tO{n^2/\eps}$.
\end{restatable}
Notice that \Cref{thm:dag-oracle} implies our main theorem about combinatorial minimum cost flow on general graphs.

\main*

\begin{proof}

    We apply \Cref{cor:simpleMCFreductiontoDAGs} which requires a $\OPDMCF{\kappa,1+\eps}$ oracle. We use \Cref{thm:dag-oracle} with $\eps$ set to be $\exp(-\log^cn\log \log n)$, which outputs a flow routing $1/3$ of the demands. We repeat this procedure $O(\log n)$ times on the original graph and the remaining demands until all the demands have been routed, summing over all the flows gives a flow with congestion $O(\log n)$. Moreover, suppose the $i$-th iteration routes $a_i$ amount of demands and the total demand is $A$, we have that the $i$-th flow's cost is at most $(1+O(\eps))\cdot \OPT\cdot \frac{a_i}{A}$ where $\OPT$ is the optimal cost for routing all the demands (this is because we can scale down the optimal flow routing all demands by $\frac{a_i}{A}$ and get a flow with cost $\OPT\cdot \frac{a_i}{A}$ routing $a_i$ demands, and according to the guarantee of \Cref{thm:dag-oracle}). Summing over all the costs gives $(1+O(\eps))\cdot \OPT$, so the required oracle for \Cref{cor:simpleMCFreductiontoDAGs} is satisfied with $\kappa=O(\log n)$. The total running time is $n^{2+o(1)}$.
\end{proof}
\subsection{Algorithm}\label{subsec:dag-alg}

\paragraph{Guessing the scale.}
We first guess (within a factor of $2$) the value and cost of $f^*$.
Assume we are given integers $q,\alpha$ such that
\[
\val(f^*)\in[q,2q]
\qquad\text{and}\qquad
\ell_G(f^*)\in[\alpha n q,\,2\alpha n q].
\]
These guesses can be obtained by trying powers of $2$ for $q$ and $\alpha$.

\paragraph{Topological weights.}
Fix a topological order $(v_1,\dots,v_n)$ of $V$.
For any (directed) edge $e=(v_i,v_j)$ that appears in the residual graph (including reverse residual edges),
define its \emph{topological weight}
\[
w(e)\;:=\;|j-i|.
\]
(Thus reverse residual edges are also assigned weights by the same formula.)

\paragraph{Residual graph and reduced lengths.}
We maintain a feasible flow $f$ (initially the zero flow) and its residual graph $G_f$.
Each residual edge $e=(u,v)\in E(G_f)$ has residual capacity $U_{G_f}(e)>0$ and original length/cost $\ell_G(e)$
(inherited from $G$ for forward residual edges and negated for reverse residual edges, as usual).
We also maintain a potential $\phi:V\to\mathbb{R}$ (initially $\phi\equiv 0$).
For a residual edge $e=(u,v)$ define its \emph{reduced length}
\[
\ell_\phi(e)\;:=\;\ell_G(e)+\phi(u)-\phi(v).
\]

\paragraph{Admissible edges with lazy updates.}
Each residual edge is marked either \emph{admissible} or \emph{inadmissible}.
Initially all residual edges are inadmissible.
We use the following rule to (re)compute the status of an edge, but we do not recompute it on every potential update;
instead we only recompute it at discrete ``checkpoints'' to control the total number of status changes.

\begin{description}
\item[Changing edge state (lazy checking).]
For a residual edge $e=(u,v)\in E(G_f)$, whenever either $\phi(u)$ or $\phi(v)$ hits a multiple of $\eps\alpha\cdot w(e)$,
we \emph{check} $e$ and set $e$ as admissible if
\[\ell_\phi(e)\le -2\eps\alpha\cdot w(e)\]
and otherwise mark $e$ as inadmissible.
\end{description}

The purpose of the checkpointing (multiples of $\eps\alpha w(e)$) is to ensure that each edge is checked only
$O\!\left(\frac{n}{\eps\,w(e)}\right)$ times over the entire execution (formalized in \Cref{subsec:runningtime}).

\paragraph{Active vertices.}
We call a vertex $v$ \emph{active} if $v\neq t$ and
\[
\phi(v)\ge -99\alpha n.
\]
Intuitively, we keep relabeling active vertices until admissible structure emerges; the algorithm terminates once
the source potential becomes sufficiently negative (see below).

\paragraph{Operations.}
The algorithm alternates between the following two types of operations.

\begin{description}
\item[Relabel.]
While there exists an active vertex $u$ with \emph{no} outgoing admissible residual edge, perform
\[
\phi(u)\;\leftarrow\;\phi(u)-\eps\alpha,
\]
and apply the lazy-checking rule above to any incident edges whose endpoint potentials hit a checkpoint.

\item[Push / Augment.]
Once every active vertex has at least one outgoing admissible residual edge, each active vertex picks an arbitrary outgoing admissible edge. Starting from $s$, follow the picked admissible edges. \Cref{lem:stflow} guarantees that we will not end up with a non-active vertex.
Then either we reach $t$ (yielding an admissible $s$--$t$ path), or we repeat a vertex (yielding an admissible directed cycle).
In either case, let $P$ denote the resulting admissible $s$--$t$ path or directed cycle, and augment $f$ along $P$ by
\[
\Delta \;:=\;\min_{e\in P} u_f(e),
\]
updating residual capacities and the flow accordingly.
(Implementation details and how we find $P$ efficiently are deferred to \Cref{subsec:runningtime}.)
\end{description}

The algorithm terminates when $\phi(s)<-50\alpha n$.

\begin{lemma}[Proof in \Cref{subsec:correctness}]\label{lem:stflow}
Repeatedly following an admissible outgoing edge starting from $s$ does not reach a non-active vertex other than $t$.
\end{lemma}

\paragraph{Recovering the primal and the dual.}
The primal solution is the flow $f$ at the end of the algorithm.

To construct the dual, define an auxiliary penalty $C'$ on \emph{backward} residual edges as follows.
For each residual edge $e=(v_i,v_j)\in E(G_f)$ with $i>j$ (i.e., $e$ goes backward in the fixed topological order),
set
\[
C'(e)\;:=\;6\eps\alpha\cdot w(e),
\]
and set $C'(e)=0$ otherwise.
The next lemma (proved in \Cref{subsec:correctness}) shows that this eliminates negative cycles.

\begin{lemma}\label{lem:nonegcycle}
The length function $\ell_\phi(\cdot)-C'(\cdot)$ on $G_f$ has no negative cycle.
Equivalently, the residual graph $G_{f,\phi}-C'$ contains no negative cycle.
\end{lemma}

Given \Cref{lem:nonegcycle}, we can run a single-source shortest path algorithm for graphs with negative edges but no
negative cycles (e.g., nearly-linear-time implementations as in~\cite{abs-2511-08551}) on $G_f$ with edge lengths
$\ell_\phi(\cdot)-C'(\cdot)$ and source $s$.
Let $\phi':V\to\mathbb{R}$ be the resulting distance labels.

Finally, we define the dual slack variables $C:E(G_f)\to\mathbb{R}_{\ge 0}$ by
\[
C(e)\;:=\;
\begin{cases}
\phi(v)-\phi(u)-\ell_G(e), & \text{if $e=(u,v)$ is saturated in $f$ and }\ \ell_G(e)<\phi(v)-\phi(u),\\
0, & \text{otherwise},
\end{cases}
\]
and return $(f,\ \phi+\phi',\ C)$ as the primal--dual output.

\subsection{Correctness}\label{subsec:correctness}

Recall our guessing assumption: there exists an $s$--$t$ flow $f^*$ with
\[
\val(f^*)\in[q,2q]
\qquad\text{and}\qquad
\ell_G(f^*)\in[\alpha n q,\,2\alpha n q].
\]

We first prove an invariant of the algorithm maintained by the lazy admissibility checks.

\begin{lemma}\label{lem:invariance}
At any time during the algorithm, every residual edge $e\in E(G_f)$ satisfies
\[
\ell_\phi(e)\;\ge\;-3\eps\alpha\cdot w(e).
\]
Moreover, if $e$ is currently marked admissible, then
\[
\ell_\phi(e)\;\le\;-\eps\alpha\cdot w(e).
\]
\end{lemma}

\begin{proof}
Fix a residual edge $e=(u,v)$ and let $\tau$ be the most recent time at which $e$ was \emph{checked}
under the ``Changing edge state'' rule.
At time $\tau$, by definition of the marking rule,
\[
e \text{ admissible at time }\tau \implies \ell_\phi^\tau(e)\le -2\eps\alpha\, w(e),
\qquad
e \text{ inadmissible at time }\tau \implies \ell_\phi^\tau(e)> -2\eps\alpha\, w(e),
\]
where $\ell_\phi^\tau$ denotes the reduced length at time $\tau$.

\smallskip
\noindent\emph{Drift between checks.}
Between $\tau$ and the next check of $e$, each endpoint potential decreases by at most $\eps\alpha\,w(e)$
in absolute value, and in particular
\begin{equation}\label{eq:drift}
\bigl|(\phi(u)-\phi(v))-(\phi^\tau(u)-\phi^\tau(v))\bigr|
\;\le\;\eps\alpha\cdot w(e).
\end{equation}
Since $\ell_\phi(e)=\ell_G(e)+\phi(u)-\phi(v)$, the same bound holds for the change in $\ell_\phi(e)$:
\begin{equation}\label{eq:drift-lphi}
\bigl|\ell_\phi(e)-\ell_\phi^\tau(e)\bigr|\;\le\;\eps\alpha\cdot w(e).
\end{equation}

\smallskip
\noindent\emph{Lower bound for all residual edges.}
If $e$ was inadmissible at time $\tau$, then $\ell_\phi^\tau(e)>-2\eps\alpha\,w(e)$, and by~\eqref{eq:drift-lphi},
\[
\ell_\phi(e)\;\ge\;\ell_\phi^\tau(e)-\eps\alpha\,w(e)\;>\;-3\eps\alpha\,w(e).
\]
If $e$ was admissible at time $\tau$, then $\ell_\phi(e)$ cannot decrease since the potential of edges with at least one admissible outgoing edge cannot decrease. Moreover, certainly $\ell_\phi^\tau(e)\le -2\eps\alpha\,w(e)$, and again by~\eqref{eq:drift-lphi},
\[
\ell_\phi(e)\;\le\;\ell_\phi^\tau(e)+\eps\alpha\,w(e)\;\le\;-\eps\alpha\,w(e).
\]
Thus, the upper bound holds.

\smallskip
\noindent\emph{Effect of push operations.}
A push only changes the residual graph by (i) decreasing residual capacities on traversed edges and
(ii) possibly creating reverse residual edges.
Reduced lengths on existing residual edges are unchanged (potentials do not change during a push).
For a forward residual edge $a\!\to\! b$, then the newly created reverse residual edge is $b\!\to\! a$ with
\[
\ell_\phi(b,a)\;=\;-\ell_G(a,b)+\phi(b)-\phi(a)\;=\;-\bigl(\ell_G(a,b)+\phi(a)-\phi(b)\bigr)\;=\;-\ell_\phi(a,b).
\]
Since the push traverses only admissible edges, $\ell_\phi(a,b)\le -\eps\alpha\,w(a,b)$, hence
$\ell_\phi(b,a)\ge \eps\alpha\,w(a,b)\ge -3\eps\alpha\,w(b,a)$.
Thus the invariant is preserved by pushes as well.
\end{proof}

\paragraph{Admissible paths from $s$ cannot enter inactive vertices (proof of \Cref{lem:stflow}).}
We now show that before termination, following admissible edges from $s$ cannot reach an inactive vertex; this implies that
once every active vertex has an outgoing admissible edge, we can always find either an admissible $s$--$t$ path or an admissible cycle.

\begin{proof}[Proof of \Cref{lem:stflow}]
Consider any time during the algorithm at which $\phi(s)\ge -50\alpha n$ (in particular, this holds whenever we execute a Push step).
Suppose for contradiction that there exists a directed path $P$ in $G_f$ from $s$ to some vertex $v$ with $\phi(v)<-99\alpha n$
such that every edge of $P$ is admissible.

Let $e=(b,a)$ be the \emph{last} backward residual edge on $P$ (i.e., $e$ satisfies $b=v_i$, $a=v_j$ with $i>j$ in the fixed topological
order). Such an edge must exist: since all edges of $P$ are admissible, $\ell_\phi(P)\le 0$, and telescoping yields
\[
\ell_{G_f}(P)\;=\;\ell_\phi(P)+\phi(v)-\phi(s)\;\le\;\phi(v)-\phi(s)\;<\;-49\alpha n,
\]
so $P$ must contain a negative-cost residual edge. Because $G$ is a DAG with nonnegative edge costs, negative-cost residual edges can only be
backward residual edges (reverse edges created by the current flow).

Since $(b,a)$ is a backward residual edge, the forward edge $(a,b)$ in $G$ currently carries positive flow in $f$.
Because $G$ is a DAG, the support of $f$ decomposes into $s$--$t$ paths; hence $(a,b)$ lies on some $s$--$t$ flow path in this decomposition.
Let $Q$ denote the prefix of that path from $s$ to $a$. Every edge on $Q$ carries positive flow, so every edge of $Q$ appears in $G_f$
as a \emph{backward} residual edge; reversing $Q$ gives a residual path $Q^{\mathrm{rev}}$ from $a$ to $s$ consisting only of backward residual edges.
Moreover, since $Q$ follows the topological order, its total topological weight satisfies $w(Q)\le n$, hence also $w(Q^{\mathrm{rev}})\le n$.

Applying \Cref{lem:invariance} to each edge of $Q^{\mathrm{rev}}$ and summing, we obtain
\[
\ell_\phi(Q^{\mathrm{rev}})\;\ge\;-3\eps\alpha\cdot w(Q^{\mathrm{rev}})\;\ge\;-3\eps\alpha n.
\]
On the other hand, telescoping along $Q^{\mathrm{rev}}$ yields
\[
\ell_\phi(Q^{\mathrm{rev}})\;=\;\ell_{G_f}(Q^{\mathrm{rev}})+\phi(a)-\phi(s)\;\le\;\phi(a)-\phi(s),
\]
because every edge of $Q^{\mathrm{rev}}$ is a backward residual edge and thus has nonpositive original residual cost.
Combining the two displays gives
\[
\phi(a)-\phi(s)\;\ge\;-3\eps\alpha n,
\qquad\text{i.e.,}\qquad
\phi(a)\;\ge\;\phi(s)-3\eps\alpha n\;\ge\;-50\alpha n-3\eps\alpha n\;>\;-99\alpha n,
\]
using $\eps\le 0.1$. Thus $a$ is active.

Now consider the suffix $P[a\leadsto v]$ of $P$ starting at $a$ and ending at $v$.
By the choice of $(b,a)$ as the last backward residual edge on $P$, every edge of $P[a\leadsto v]$ is a \emph{forward} residual edge,
hence has nonnegative original cost.
For any forward residual admissible edge $x\to y$ we have
\[
\ell_\phi(x,y)=\ell_G(x,y)+\phi(x)-\phi(y)\le -\eps\alpha\,w(x,y)\le 0,
\]
and since $\ell_G(x,y)\ge 0$, it follows that $\phi(y)\ge \phi(x)$.
Therefore, potentials are nondecreasing along $P[a\leadsto v]$, and in particular $\phi(v)\ge \phi(a)>-99\alpha n$.
This contradicts $\phi(v)<-99\alpha n$ and completes the proof.
\end{proof}

\paragraph{The algorithm routes $\Omega(q)$ flow.}
We now show that upon termination we have routed a constant fraction of the guessed value $q$.

\begin{lemma}\label{lem:value-omega-q}
At termination, the algorithm outputs a flow $f$ with $\val(f)\ge 0.9\,q$.
\end{lemma}

\begin{proof}
Let termination mean $\phi(s)<-50\alpha n$. Note also that $\phi(t)$ never changes (we never relabel $t$), hence $\phi(t)=0$.

Assume for contradiction that $\val(f)<0.9\,q$.
Since $\val(f^*)\ge q$, the value gap satisfies
\[
\Delta \;:=\;\val(f^*)-\val(f)\;\ge\;0.1\,q.
\]

Let $g$ be the standard residual representation of the difference $f^*-f$ in $G_f$:
for each original edge $e$ we send $\max\{f^*(e)-f(e),0\}$ units on the forward residual edge $e$,
and $\max\{f(e)-f^*(e),0\}$ units on the backward residual edge $e^{\mathrm{rev}}$.
Then $g$ is a feasible $s$--$t$ flow in $G_f$ of value $\val(g)=\Delta$.
Moreover, since $\ell_G$ is nonnegative on original edges,
\[
\ell_{G_f}(g)\;\le\;\ell_G(f^*)\;\le\;2\alpha n q.
\]
(The inequality $\ell_{G_f}(g)\le \ell_G(f^*)$ holds because each backward residual edge has nonpositive cost, and $\ell_G(f)\ge 0$.)

We now derive a lower bound on $\ell_{G_f}(g)$ that contradicts this.
By flow conservation and telescoping of potentials,
for any $s$--$t$ flow $g$ we have
\begin{equation}\label{eq:telescoping}
\ell_\phi(g)\;=\;\ell_{G_f}(g)+\phi(s)\val(g)-\phi(t)\val(g)
\;=\;\ell_{G_f}(g)+(\phi(s)-\phi(t))\Delta.
\end{equation}
Rearranging and using $\phi(t)=0$ and $\phi(s)<-50\alpha n$ gives
\begin{equation}\label{eq:lower1}
\ell_{G_f}(g)\;=\;\ell_\phi(g)+(\phi(t)-\phi(s))\Delta
\;\ge\;\ell_\phi(g)+50\alpha n\cdot \Delta.
\end{equation}

Next, apply \Cref{lem:invariance} to lower bound $\ell_\phi(g)$:
\[
\ell_\phi(g)\;=\;\sum_{e\in E(G_f)} g(e)\,\ell_\phi(e)
\;\ge\;-3\eps\alpha\sum_{e\in E(G_f)} g(e)\,w(e).
\]
It remains to upper bound the weighted sum $\sum_e g(e)w(e)$.
For any $s$--$t$ flow $h$ on a DAG, define its $w$-weight by $W(h):=\sum_e h(e)w(e)$.
Since $G$ is a DAG, $h$ decomposes into $s$--$t$ paths; for any such path $p$,
\[
\sum_{(v_i,v_j)\in p} w(v_i,v_j)\;=\;\sum_{(v_i,v_j)\in p} (j-i)\;=\;(\text{index}(t)-\text{index}(s))\;\le\;n,
\]
hence $W(h)\le n\cdot \val(h)$.
Applying this to $f^*$ and $f$ and noting that $g$ uses at most the forward part of $f^*$ and the backward part of $f$,
we obtain
\[
\sum_e g(e)w(e)\;\le\;W(f^*)+W(f)\;\le\;n\val(f^*)+n\val(f)\;\le\;n\cdot(2q+0.9q)\;<\;3nq.
\]
Therefore,
\[
\ell_\phi(g)\;\ge\;-3\eps\alpha\cdot 3nq\;=\;-9\eps\alpha n q.
\]
Plugging into~\eqref{eq:lower1} and using $\Delta\ge 0.1q$ yields
\[
\ell_{G_f}(g)
\;\ge\;-9\eps\alpha n q + 50\alpha n\cdot 0.1q
\;=\;\bigl(5-9\eps\bigr)\alpha n q
\;>\;2\alpha n q,
\]
for $\eps\le 0.1$.
This contradicts $\ell_{G_f}(g)\le 2\alpha n q$, completing the proof.
\end{proof}

\paragraph{Primal--dual gap.}
We will use the following exact (and standard) identity, which cleanly formalizes the
“slackness $\Rightarrow$ small gap’’ step.

\begin{lemma}[Exact primal--dual gap identity]\label{lem:pd-gap-identity}
Consider the min-cost flow primal/dual pair:
\[
\text{(P)}\quad \min \sum_{e\in E}\ell(e)f(e)
\ \ \text{s.t. flow conservation and } 0\le f(e)\le U(e),
\]
\[
\text{(D)}\quad \max \sum_{u\in V}(-D(u))\psi(u)-\sum_{e\in E}U(e)C(e)
\ \ \text{s.t. } \psi(v)-\psi(u)\le \ell(e)+C(e),\ C(e)\ge 0.
\]
For any primal-feasible $f$ and any dual-feasible $(\psi,C)$,
\begin{equation}\label{eq:gap-identity}
\sum_{e\in E}\ell(e)f(e)\;-\;
\Bigl(\sum_{u\in V}(-D(u))\psi(u)-\sum_{e\in E}U(e)C(e)\Bigr)
\;=\;
\sum_{e\in E} f(e)\,\gamma(e)
\;+\;
\sum_{e\in E} (U(e)-f(e))\,C(e),
\end{equation}
where $\gamma(e)\coloneqq \ell(e)+C(e)-(\psi(v)-\psi(u))\ge 0$ is the dual slack on $e=(u,v)$.
\end{lemma}
\begin{proof}
Expand the dual objective and use flow conservation:
\[
\sum_{u\in V} D(u)\psi(u)
=
\sum_{u\in V}\psi(u)\Bigl(\sum_{e\in\delta^+(u)}f(e)-\sum_{e\in\delta^-(u)}f(e)\Bigr)
=
\sum_{e=(u,v)\in E} f(e)\bigl(\psi(u)-\psi(v)\bigr).
\]
Hence
\[
\sum_{e}\ell(e)f(e) + \sum_u D(u)\psi(u) + \sum_e U(e)C(e)
=
\sum_e f(e)\bigl(\ell(e)+\psi(u)-\psi(v)\bigr) + \sum_e U(e)C(e).
\]
Add and subtract $\sum_e f(e)C(e)$ and rearrange to obtain~\eqref{eq:gap-identity}.
\end{proof}

\begin{lemma}\label{lem:approx-pd}
Let $\phi^*\coloneqq \phi+\phi'$, where $\phi'$ is computed as in the analysis (with $\phi'(s)=0$),
and let $C$ be the algorithm's output.
Then the primal--dual gap satisfies
\[
\ell_G(f)\;-\;\Bigl(\sum_{u\in V}(-D(u))\phi^*(u)-\sum_{e\in E}U(e)C(e)\Bigr)
\;\le\;
6\eps\alpha \sum_{e\in E} f(e)\,w(e)
\;\le\;
6\eps\alpha n\cdot \val(f).
\]
\end{lemma}
\begin{proof}
We apply Lemma~\ref{lem:pd-gap-identity} with $(\psi,C)=(\phi^*,C)$.
By construction, $C(e)\ge 0$ and $(\phi^*,C)$ is dual-feasible (the algorithm sets $C(e)$ exactly
to fix the upper-capacity constraints on saturated edges).

It remains to upper bound the two nonnegative terms in~\eqref{eq:gap-identity}.
We claim:

\smallskip
\noindent\textbf{(i) Saturated edges do not contribute to the second term.}
If $C(e)>0$, then $e$ is saturated, i.e.\ $f(e)=U(e)$, hence $(U(e)-f(e))C(e)=0$.
Therefore the second sum in~\eqref{eq:gap-identity} vanishes.

\smallskip
\noindent\textbf{(ii) Edges with $f(e)>0$ have slack $\gamma(e)\le 6\eps\alpha\,w(e)$.}
Fix an original edge $e=(u,v)$ with $f(e)>0$.
Then the reverse residual edge $(v,u)$ exists in $G_f$.
By the way $\phi'$ is defined in the analysis (as a distance potential on the modified residual graph),
we have for this residual edge the inequality
\[
\bigl(-\ell(e)\bigr) + \phi^*(v)-\phi^*(u) \;\ge\; -6\eps\alpha\,w(e),
\]
equivalently,
\begin{equation}\label{eq:lower-bound-tight}
\phi^*(v)-\phi^*(u)\;\ge\;\ell(e)-6\eps\alpha\,w(e).
\end{equation}
On the other hand, if $f(e)<U(e)$ then the forward residual edge $(u,v)$ also exists, and the same reasoning gives
$\ell(e)+\phi^*(u)-\phi^*(v)\ge 0$, i.e.\ $\phi^*(v)-\phi^*(u)\le \ell(e)$.
If $f(e)=U(e)$, dual feasibility ensures $\phi^*(v)-\phi^*(u)\le \ell(e)+C(e)$.
In both cases,
\[
\gamma(e)=\ell(e)+C(e)-(\phi^*(v)-\phi^*(u))
\le \ell(e)+C(e)-(\ell(e)-6\eps\alpha w(e))
= C(e)+6\eps\alpha w(e).
\]
Moreover, when $f(e)>0$ and $f(e)<U(e)$ we have $C(e)=0$ (by construction, $C$ is only needed for saturated edges),
so $\gamma(e)\le 6\eps\alpha w(e)$; when $f(e)=U(e)$ we also have $\gamma(e)=0$ by definition of $C(e)$.
Thus in all cases with $f(e)>0$ we have $\gamma(e)\le 6\eps\alpha w(e)$.

\smallskip
Combining (i) and (ii) with~\eqref{eq:gap-identity} gives
\[
\ell_G(f)-\text{Dual}(\phi^*,C)
=
\sum_{e} f(e)\gamma(e)
\le
6\eps\alpha \sum_e f(e)w(e).
\]
Finally, every unit of $s$--$t$ flow can be decomposed into an $s$--$t$ path following the topological order,
and along such a path the $w$-weights telescope to at most $n$; hence
$\sum_e f(e)w(e)\le n\cdot \val(f)$.
This yields the stated bound.
\end{proof}

\subsection{Complexity (Implementation Details)}\label{subsec:runningtime}

We implement the algorithm almost identically to Appendix~A of~\cite{BernsteinBST24} using dynamic trees
(link--cut trees). For completeness we spell out the implementation details here.
The only new feature is that, when a newly chosen admissible edge would create a directed cycle in the maintained structure,
we immediately augment along that cycle (cycle-cancel) before continuing.

\paragraph{Bounding the number of augmentations.}
We first bound how many augmentations (Push operations) occur.

\begin{lemma}\label{lem:sum-inv-w}
For any directed graph on vertices ordered as $(v_1,\dots,v_n)$, letting $w(v_i,v_j)=|j-i|$, we have
\[
\sum_{e\in E}\frac{1}{w(e)}\;\le\;\sum_{d=1}^{n-1}\frac{n-d}{d}\;=\;O(n\log n)\;=\;\tO{n}.
\]
\end{lemma}
\begin{proof}
Group edges by their topological distance $d=w(e)$. For a fixed $d\in\{1,\dots,n-1\}$,
there are at most $n-d$ ordered pairs $(i,i+d)$, hence at most $n-d$ edges with $w(e)=d$.
Therefore
\[
\sum_{e\in E}\frac{1}{w(e)}
\;\le\;
\sum_{d=1}^{n-1}\frac{n-d}{d}
\;=\;
n\sum_{d=1}^{n-1}\frac{1}{d}-\sum_{d=1}^{n-1}1
\;=\;
O(n\log n).
\]
\end{proof}

\begin{lemma}\label{lem:augmentationsmall}
Throughout the algorithm, the total number of augmentations is $\tO{n^2/\eps}$.
Moreover, any fixed original edge $e$ (equivalently, its two residual directions over time) becomes saturated at most
$O\!\left(\frac{n}{\eps\,w(e)}\right)$ times.
\end{lemma}

\begin{proof}
Every augmentation (along an $s$--$t$ path or a directed cycle) sends
\[
\Delta=\min_{a\in P}u_f(a)
\]
units and therefore saturates at least one residual edge on the augmented object $P$.
Thus it suffices to bound the total number of saturation events.

Fix an original edge $e=(v_i,v_j)$ and write $d\coloneqq w(e)=|j-i|$.
Consider two consecutive times at which \emph{the same direction} of $e$ becomes saturated as a residual edge.
At the moment a residual edge $a=(x,y)$ is saturated in an augmentation, it must be admissible, hence by \Cref{lem:invariance}
\[
\ell_\phi(a)\;\le\;-\eps\alpha\,w(a)\;=\;-\eps\alpha d.
\]
When $a$ is saturated, the reverse residual edge $a^{\mathrm{rev}}=(y,x)$ is created (or its capacity increases), and
\[
\ell_\phi(a^{\mathrm{rev}})=-\ell_\phi(a)\;\ge\;\eps\alpha d.
\]
For $a^{\mathrm{rev}}$ to become admissible (and hence be eligible to be saturated in some later augmentation), we must have
$\ell_\phi(a^{\mathrm{rev}})\le -\eps\alpha d$ at that later time.
Therefore, between the creation of $a^{\mathrm{rev}}$ and the time it is next saturated, its reduced length must drop by at least
$2\eps\alpha d$.

Now observe that for a fixed residual edge $(y,x)$,
\[
\ell_\phi(y,x)=\ell_G(y,x)+\phi(y)-\phi(x),
\]
and potentials \emph{never increase} (they only decrease under relabels).
Hence the only way $\ell_\phi(y,x)$ can decrease by $2\eps\alpha d$ is that $\phi(y)$ decreases by at least $2\eps\alpha d$
(since $\phi(x)$ cannot increase to help).
Each vertex potential decreases in steps of size $\eps\alpha$, and a vertex is relabeled only while active, so over the entire run
each vertex $u$ satisfies
\[
0\ge \phi(u)\ge -99\alpha n,
\]
i.e., $\phi(u)$ decreases by at most $99\alpha n$ in total.
Charging $2\eps\alpha d$ units of potential decrease of the tail endpoint each time $e$ is saturated (in the appropriate direction),
we conclude that $e$ can be saturated at most
\[
O\!\left(\frac{99\alpha n}{2\eps\alpha d}\right)
=
O\!\left(\frac{n}{\eps\,w(e)}\right)
\]
times overall.

Summing over all edges and using \Cref{lem:sum-inv-w}, the total number of saturation events (hence augmentations) is
\[
O\!\left(\sum_{e\in E}\frac{n}{\eps\,w(e)}\right)
=
O\!\left(\frac{n}{\eps}\right)\cdot \sum_{e\in E}\frac{1}{w(e)}
=
\tO{n^2/\eps}.
\]
\end{proof}

\paragraph{Dynamic trees.}
We use a standard link--cut tree data structure~\cite{SleatorT83,GoldbergT88}.

\newcommand{\Bnu}{\boldsymbol{\nu}}
\newcommand{\Bc}{\boldsymbol{c}}
\newcommand{\forward}[1]{\overrightarrow{#1}}
\newcommand{\backward}[1]{\overleftarrow{#1}}
\newcommand{\Bsink}{\boldsymbol{\nabla}}
\newcommand{\Bw}{\boldsymbol{w}}
\newcommand{\abs}{\boldsymbol{\mathrm{abs}}}

\begin{lemma}[Dynamic trees {\cite{SleatorT83,GoldbergT88}}]\label{lem:dyn-trees}
There is a data structure that maintains a collection of rooted trees $\cT$ together with values
$\Bnu\in\mathbb{R}^{E(\cT)}$ on the tree edges.
It supports the following operations in amortized $O(\log n)$ time:
\begin{itemize}
  \item $\textsc{Link}(u,v,\nu)$: add an edge making $v$ the parent of $u$ with stored value $\Bnu(u,v)\gets \nu$.
        Requires that $u$ is currently a root and $u,v$ are in different trees.
  \item $\textsc{Cut}(u)$: delete the parent edge of $u$ (if any), making $u$ a root.
  \item $\textsc{FindRoot}(u)$: return the root of the tree containing $u$.
  \item $\textsc{FindMin}(u)$: return an edge of minimum $\Bnu$ on the path from $u$ to its root
        (ties broken toward $u$).
  \item $\textsc{Add}(u,x)$: add $x$ to $\Bnu(e)$ for every edge $e$ on the path from $u$ to its root.
\end{itemize}
\end{lemma}

\paragraph{Maintaining an admissible pseudoforest.}
During a Push phase (i.e., when no more relabels are possible), every active vertex has at least one outgoing admissible residual edge.
For each active vertex $u\neq t$, we maintain \emph{one chosen admissible outgoing edge} $p(u)=(u,v)$, called its \emph{parent edge}.
The chosen parent edges define a directed pseudoforest (outdegree at most $1$).
We maintain this structure inside dynamic trees as an \emph{actual forest} by immediately canceling any cycle that would be created:

\begin{itemize}
\item When $u$ becomes active and currently has no parent edge in $\cT$, choose some admissible outgoing residual edge $e=(u,v)$.
      If $u$ is a root and $u$ and $v$ are in different trees, we insert it via $\textsc{Link}(u,v,u_f(e))$.
\item If inserting $e=(u,v)$ would create a directed cycle (i.e., $u$ is a root but $\textsc{FindRoot}(v)=u$),
      then the cycle consists of the path from $v$ to $u$ in the current tree plus the edge $(u,v)$.
      We augment along this directed cycle by the bottleneck residual capacity (found via $\textsc{FindMin}(v)$ and the residual of $(u,v)$),
      update residual capacities using $\textsc{Add}(v,-\Delta)$ along the tree path, and decrease the residual of $(u,v)$ by $\Delta$.
      At least one edge on the cycle becomes saturated; we delete every saturated tree edge by cutting it from $\cT$.
      After the cycle is broken, we retry linking $u$ to some admissible outgoing edge (possibly the same one if it remains residual).
\end{itemize}

Whenever an edge is marked inadmissible (or its residual capacity drops to $0$), if it is currently used as a parent edge in $\cT$
we remove it using $\textsc{Cut}(\cdot)$ and then pick a replacement admissible outgoing edge for that child (if the child is active).
All such changes take $O(\log n)$ time per link/cut/update.

\paragraph{Augmenting on an $s$--$t$ path.}
When the structure is cycle-free in the component containing $s$, the chosen parent pointers define a (tree) path from $s$ to its root.
By \Cref{lem:stflow} (and the cycle-cancel rule above), this root is $t$.
We then augment along the unique tree path from $s$ to $t$:
we compute the bottleneck $\Delta$ via $\textsc{FindMin}(s)$ and apply $\textsc{Add}(s,-\Delta)$.
After this update, at least one tree edge on the path has residual capacity $0$; we repeatedly call $\textsc{FindMin}(s)$ and delete
any returned edge with $\Bnu=0$ (cutting it from $\cT$ and updating the corresponding residual reverse edge in $G_f$).
Each such deletion corresponds to a saturation event.

\paragraph{Total running time.}
We bound the total number of expensive events.

\begin{itemize}
\item \textbf{Relabels.} Each relabel decreases $\phi(u)$ by $\eps\alpha$, and $\phi(u)$ decreases from $0$ to below $-99\alpha n$
at most. Hence each vertex is relabeled $O(n/\eps)$ times and the total number of relabels is $O(n^2/\eps)$.
\item \textbf{Edge checks.} An edge $e=(u,v)$ is checked only when $\phi(u)$ or $\phi(v)$ hits a multiple of $\eps\alpha w(e)$.
Since each endpoint potential decreases by at most $O(\alpha n)$, the number of checks of $e$ is
$O\!\left(\frac{\alpha n}{\eps\alpha w(e)}\right)=O\!\left(\frac{n}{\eps w(e)}\right)$.
Summing over edges and using \Cref{lem:sum-inv-w} gives a total of $\tO{n^2/\eps}$ edge checks.
Each check causes $O(1)$ admissibility work plus $O(\log n)$ dynamic-tree updates in the worst case.
\item \textbf{Augmentations and saturated-edge deletions.}
By \Cref{lem:augmentationsmall}, the total number of augmentation steps is $\tO{n^2/\eps}$, and the total number of saturated-edge
deletions from $\cT$ is also $\tO{n^2/\eps}$. Each such event costs $O(\log n)$ amortized time.
\end{itemize}

Putting everything together, the total running time is $\tO{n^2/\eps}$.

\appendix
\section{Missing Proofs in \Cref{sec:LCED}}\label{sec:LCEDmissing}
\subsection{Union of Sparse Moving Cuts}\label{subsec:unionofsparsecuts}
In this section, we will prove \Cref{lem:existential_cut_sequence_different_length}.

We first need the following lemma from \cite{HaeuplerLSW25}.

\begin{lemma}[Lemma 3.17 of \cite{HaeuplerLSW25}]\label{lem:existential_cut_sequence}
Let $C_1, \ldots, C_k$ be an sequence of $(h,s)$-length $\phi$-sparse for some node-weighting $A$ in the graph $G$ where $h > 1$, $\phi < 1$, $1 \le \alpha \le \log n$ and
$s > \frac{4\log_2 n}{\alpha}$, then $\sum_{i} |C_i| \leq (2^{8\alpha+2}\phi\ln n)\cdot |A|$.
\end{lemma}

We can generalize \Cref{lem:existential_cut_sequence} to allow the cut sequence to be a subsequence as follows.

\begin{lemma}\label{lem:existential_cut_sequence_general}
Let $A$ be a node-weighting on a graph $G$ with positive integral edge lengths and capacities.
Let $C_1,\ldots,C_k$ be a sequence of moving cuts, and let
\[
C_{a_1},\ldots,C_{a_z}
\]
be a subsequence of $C_1,\ldots,C_k$.
Let $h>1$, $\phi\in(0,1)$, $1\le \alpha\le \log n$, and
\[
s>\frac{4\log_2 n}{\alpha}.
\]
Suppose that for every $r\in[z]$, the cut $C_{a_r}$ is $(h,s)$-length $\phi$-sparse with respect to $A$ in the graph
\[
G-\sum_{j<a_r} C_j.
\]
Then
\[
\sum_{r=1}^z |C_{a_r}|
\le
(2^{8\alpha+2}\phi\ln n)\cdot |A|.
\]
\end{lemma}

\begin{proof}
This is proved by the same potential argument as in the proof of Lemma~3.17 of \cite{HaeuplerLSW25}, with only one change: instead of summing the potential drops over all cuts, we sum only over the subsequence \(C_{a_1},\ldots,C_{a_z}\).

More precisely, define
\[
G_1:=G,
\qquad
G_i:=G-\sum_{j<i} C_j
\qquad\text{for every }i\in[k+1].
\]
For each \(i\), let \(w_i\) denote the exponential distance weight \(w_h^\alpha\) with respect to the graph \(G_i\), and let \(D_i\) denote the corresponding exponential demand \(D^i_{h,A}\). Also define the potential
\[
P_i(u):=A(u)\ln(w_i(u))
\qquad\forall u\in V.
\]

Fix any \(r\in[z]\). Since \(C_{a_r}\) is \((h,s)\)-length \(\phi\)-sparse with respect to \(A\) in \(G_{a_r}\), the same argument as in Lemma~3.17 of \cite{HaeuplerLSW25} gives
\[
\sep^{a_r}_{h\cdot s/2}(C_{a_r},D_{a_r})
\ge 2^{-8\alpha-1}\cdot \frac{|C_{a_r}|}{\phi},
\]
and therefore
\[
\sum_{u\in V}\bigl(P_{a_r}(u)-P_{a_r+1}(u)\bigr)
\ge
2^{-8\alpha-2}\cdot \frac{|C_{a_r}|}{\phi}.
\]
Equivalently,
\[
|C_{a_r}|
\le
2^{8\alpha+2}\phi\sum_{u\in V}\bigl(P_{a_r}(u)-P_{a_r+1}(u)\bigr).
\]

Summing over all \(r\in[z]\), we obtain
\[
\sum_{r=1}^z |C_{a_r}|
\le
2^{8\alpha+2}\phi
\sum_{r=1}^z\sum_{u\in V}\bigl(P_{a_r}(u)-P_{a_r+1}(u)\bigr).
\]
Now the key observation is that the potential is monotone under arbitrary length increases: for every \(i\) and every \(u\in V\),
\[
P_i(u)\ge P_{i+1}(u),
\]
since passing from \(G_i\) to \(G_{i+1}\) only increases edge lengths. Hence all terms in the above sum are nonnegative, and so
\[
\sum_{r=1}^z\sum_{u\in V}\bigl(P_{a_r}(u)-P_{a_r+1}(u)\bigr)
\le
\sum_{i=1}^k\sum_{u\in V}\bigl(P_i(u)-P_{i+1}(u)\bigr).
\]
The latter telescopes to
\[
\sum_{i=1}^k\sum_{u\in V}\bigl(P_i(u)-P_{i+1}(u)\bigr)
\le
\sum_{u\in V} P_1(u).
\]
Therefore,
\[
\sum_{r=1}^z |C_{a_r}|
\le
2^{8\alpha+2}\phi\sum_{u\in V}P_1(u)
=
2^{8\alpha+2}\phi\sum_{u\in V}A(u)\ln(w_1(u)).
\]
Finally, as in Lemma~3.17 of \cite{HaeuplerLSW25}, we have \(\ln(w_1(u))\le \ln n\) for all \(u\), and so
\[
\sum_{r=1}^z |C_{a_r}|
\le
(2^{8\alpha+2}\phi\ln n)\cdot |A|.
\]
\end{proof}

Now we are ready to prove \Cref{lem:existential_cut_sequence_different_length}.
\cutsequence*
\begin{proof}
Partition the index set \([k]\) into dyadic buckets according to the values \(h_i\).
For each integer \(\ell\), let
\[
I_\ell:=\{\, i\in[k] : 2^\ell < h_i \le 2^{\ell+1}\,\}.
\]
Only \(O(\log H)\) of these buckets are nonempty.

Fix one nonempty bucket \(I_\ell\), and let
\[
\widehat h_\ell:=2^{\ell+1}.
\]
For each \(i\in I_\ell\), define a rescaled cut
\[
\widehat C_i:=\frac{2h_i}{\widehat h_\ell}\, C_i = \frac{h_i}{2^\ell}\, C_i.
\]
Since \(2^\ell < h_i \le 2^{\ell+1}\), we have
\[
1<\frac{2h_i}{\widehat h_\ell}\le 2,
\]
and therefore
\[
|C_i|\le |\widehat C_i| \le 2|C_i|.
\]

We now view \(\widehat C_i\) as a \(\widehat h_\ell\cdot (s/2)\)-length moving cut. Indeed,
\[
\widehat h_\ell\cdot \frac{s}{2}\cdot \widehat C_i(e)
=
2^{\ell+1}\cdot \frac{s}{2}\cdot \frac{h_i}{2^\ell} C_i(e)
=
h_i s\, C_i(e),
\]
so \(\widehat C_i\) induces exactly the same edge-length increases as \(C_i\). Hence
\[
G-\widehat C_i = G-C_i.
\]

Next, let \(D_i^\star\) be a witnessing \(A\)-respecting \(h_i\)-length symmetric demand showing that \(C_i\) is \((h_i,s)\)-length \(\phi\)-sparse. Since \(h_i\le \widehat h_\ell\), the same demand \(D_i^\star\) is also \(A\)-respecting \(\widehat h_\ell\)-length symmetric.
Moreover,
\[
\widehat h_\ell\cdot \frac{s}{2}
=
2^\ell s
\le
h_i s.
\]
Therefore every demand pair that is \(h_i s\)-separated by \(C_i\) is also \(\widehat h_\ell(s/2)\)-separated by \(\widehat C_i\). Hence
\[
\sep_{\widehat h_\ell (s/2)}(\widehat C_i,D_i^\star)
\ge
\sep_{h_i s}(C_i,D_i^\star).
\]
Combining this with \(|\widehat C_i|\le 2|C_i|\), we obtain
\[
|\widehat C_i|
\le
2|C_i|
\le
2\phi\cdot \sep_{h_i s}(C_i,D_i^\star)
\le
2\phi\cdot \sep_{\widehat h_\ell (s/2)}(\widehat C_i,D_i^\star).
\]
Thus \(\widehat C_i\) is \((\widehat h_\ell, s/2)\)-length \(2\phi\)-sparse with respect to \(A\) in the graph
\[
G-\sum_{j<i} C_j
=
G-\sum_{j<i} \widehat C_j,
\]
where equality again holds because \(C_j\) and \(\widehat C_j\) induce the same metric.

Now apply \Cref{lem:existential_cut_sequence_general} to the subsequence
\[
\{\widehat C_i : i\in I_\ell\}
\]
with parameters \(\widehat h_\ell\), \(2\phi\), and \(s/2\). This is valid because
\[
\frac{s}{2}>\frac{4\log_2 n}{\alpha}
\]
by assumption. Therefore,
\[
\sum_{i\in I_\ell} |\widehat C_i|
\le
(2^{8\alpha+2}\cdot 2\phi\ln n)\cdot |A|
=
O(2^{8\alpha+2}\phi\ln n)\cdot |A|.
\]
Since \(|C_i|\le |\widehat C_i|\), we also get
\[
\sum_{i\in I_\ell} |C_i|
\le
O(2^{8\alpha+2}\phi\ln n)\cdot |A|.
\]

Finally, summing over all nonempty buckets gives
\[
\sum_{i=1}^k |C_i|
\le
O(\log H)\cdot O(2^{8\alpha+2}\phi\ln n)\cdot |A|
=
O\!\bigl(2^{8\alpha+2}\phi\ln n\cdot \log H\bigr)\cdot |A|.
\]
If \(H\le n^{O(1)}\), then \(\log H=O(\log n)\), and hence
\[
\sum_{i=1}^k |C_i|
\le
O\!\bigl(2^{8\alpha+2}\phi\log^2 n\bigr)\cdot |A|.
\]
\end{proof}

\subsection{FatLDD on Directed Graphs}\label{sec:FatLDD}
In this section, we will fill in the missing proofs in \Cref{subsec:FatLDD}.

We first show that it is possible to construct an undirected graph approximating the round-trip distance of a directed graph. We need the following lemma about randomized low-diameter decomposition.

\begin{lemma}[\cite{AshvinkumarBCGH24}]\label{lem:LDD}
    There is a randomized algorithm that, given a directed graph $G=(V,E)$ with edge length $\ell$ and a parameter $d$, makes $\tO{1}$ calls to an exact SSSP oracle and computes a distribution of a partition of the vertex set $V=V_1\uplus V_2...\uplus V_k$ such that
    \begin{itemize}
        \item for each edge $e=(u,v)$, the probability that $e$ is \emph{reversed}, i.e., $u\in V_i,v\in V_j$ with $j<i$, is at most $O(\log^2n)\cdot w(e)/d$.
        \item for every $i\in[k]$ and every $u,v\in V_i$, we have $\dist_G(u,v)\le d$. 
    \end{itemize}
\end{lemma}

\begin{restatable}{lemma}{roundtripapprox}\label{lem:apxroundtrip} There is a randomized algorithm that, given a directed graph $G=(V,E)$ with edge lengths $\ell$, constructs an undirected graph $G'=(V,E')$ using $\tO(1)$ calls to an exact directed SSSP oracle such that, with high probability, for every $u,v\in V$, \[ \Rdist_G(u,v)\le \dist_{G'}(u,v)\le O(\log^2 n)\cdot \Rdist_G(u,v). \] \end{restatable} \begin{proof} For each $d=2^i$ from $i=1$ to $i=C\cdot \log n$ for sufficiently large $C$, we construct $O(\log n)$ independent LDDs with parameter $d$ using \Cref{lem:LDD}. For each computed LDD $V_1,...,V_k$ with parameter $d$, and each $j\in[k]$, we add edges from every node in $V_j$ to a single node $w\in V_j$ with length $d$, which results in the final graph $G'=(V,E)$. Notice that constructing LDD makes $\tO{1}$ calls to exact directed SSSP oracle, and connecting the edges has complexity subsumed by the LDD computation. It remains to prove the correctness. \paragraph{Correctness.} Since for every two nodes $u,v\in V_j$, it is from the LDD property that $\dist_G(u,v)\le d$, we immediately get one side of the inequality $\Rdist_G(u,v)\le \dist_{G'}(u,v)$. For every $u,v\in V$, let us focus on the LDD with parameter $d=C\log^2n\cdot \Rdist(u,v)$ for sufficiently large constant $C$, denoted by $V_1,...,V_k$. By the choice of $d$, such $d$ must exist. Consider the shortest path $p_{u,v}$ from $u$ to $v$ and the shortest path $p_{v,u}$ from $v$ to $u$. We have that $\ell(p_{u,v})+\ell(p_{v,u})=\Rdist(u,v)$. According to the property of LDD \Cref{lem:LDD}, with probability at least $0.1$, no edges in either $p_{u,v}$ or $p_{v,u}$ will be a reversed edge for $V_1,...,V_k$, which implies that $u,v\in V_j$ for some $j\in[k]$. Thus, in $G'$, the distance between $u,v$ is $2\cdot d=O(\log^2n)\cdot \Rdist_G(u,v)$. Since we repeat the LDD for $O(\log n)$ independent times, one of them will satisfy this inequality w.h.p. Applying union bound over all pairs of vertices gives the final result. \end{proof}
\FatLDD*

We need the following lemma for undirected graphs.

\begin{lemma}[Theorem 5.32 of \cite{HaeuplerHT24}]\label{lem:undirfatLDD}
For any $h_{\cov}>0$, $s>1$, and $\eps\in(0,1)$, and any undirected graph $G=(V,E)$ with edge lengths $\ell$,
there is an algorithm that computes a neighborhood cover with respect to $\dist_G$ with
covering radius $h_{\cov}$, separation factor $s$,
diameter
\[
h_{\diam}=\frac{1}{\eps}\cdot O(s)^{O(1/\eps)}\cdot h_{\cov},
\]
and width
\[
\omega=n^{O(\eps)}\cdot \log n.
\]
The algorithm runs in $m\cdot h_{\diam}\cdot n^{O(\eps)}$ time.
\end{lemma}

\begin{proof}[Proof of \Cref{lem:directedlddwithseparationFatLDD}]
Let $\alpha=O(\log^2 n)$ denote the distortion in \Cref{lem:apxroundtrip}, so that with high probability
\[
\Rdist_G(u,v)\le \dist_{G'}(u,v)\le \alpha\cdot \Rdist_G(u,v)
\qquad\text{for all }u,v\in V.
\]
Condition on this event.

\paragraph{Algorithm.}
We construct the desired neighborhood cover in three steps.

\begin{enumerate}
    \item Apply \Cref{lem:apxroundtrip} to $G$ and obtain an undirected graph $G'$.

    \item Apply the hopset construction from \cite{Cohen00} to $G'$. This yields a set $H$ of
    $n^{1+o(1)}$ edges such that:
    \begin{itemize}
        \item every hopset edge $(x,y)\in H$ has length $\dist_{G'}(x,y)$, and hence
        \[
        \dist_{G'\cup H}(u,v)=\dist_{G'}(u,v)
        \qquad\text{for all }u,v\in V;
        \]
        \item for every pair $u,v\in V$, there exists a path in $G'\cup H$ from $u$ to $v$
        using at most $h=n^{o(1)}$ edges and having length at most
        $1.1\cdot \dist_{G'}(u,v)$.
    \end{itemize}

    \item Define
    \[
    \lambda:=\frac{\alpha h_{\cov}}{h}.
    \]
    Let $\widehat{G}$ be the graph on vertex set $V$ and edge set $E(G')\cup H$ in which
    every edge $e$ of original length $w(e)$ is assigned the new length
    \[
    \widehat{\ell}(e):=\left\lceil\frac{w(e)}{\lambda}\right\rceil.
    \]
    Since all edge lengths in $\widehat{G}$ are positive integers, \Cref{lem:undirfatLDD}
    applies to $\widehat{G}$.

    We invoke \Cref{lem:undirfatLDD} on $\widehat{G}$ with covering radius $10h$ and
    separation factor $40\alpha s$. Let
    \[
    \widehat{\cN}=\{\widehat{\cS}_1,\dots,\widehat{\cS}_\omega\}
    \]
    denote the resulting neighborhood cover. For each $i\in[\omega]$, let
    $\widehat h_i$ be the diameter parameter of the clustering $\widehat{\cS}_i$.
    Thus
    \[
    10h\le \widehat h_i\le \widehat h_{\diam}
    \qquad\text{for all }i\in[\omega],
    \]
    where
    \[
    \widehat h_{\diam}
    =\frac{1}{\eps}\cdot O(\alpha s)^{O(1/\eps)}\cdot h.
    \]

    For each clustering $\widehat{\cS}_i$, we output the same partition of vertices,
    now viewed as a clustering in the metric $\Rdist_G$, and assign it diameter
    \[
    h_i:=\lambda\widehat h_i.
    \]
    For each cluster $S\in\widehat{\cS}_i$, define its witness set by
    \[
    S':=\ball_{\widehat{G}}(S,10\alpha s\,\widehat h_i).
    \]
\end{enumerate}

We claim that the resulting family is the desired neighborhood cover in $G$.

\paragraph{Correctness analysis.}
We verify the covering property, the diameter bound, the separation guarantee,
and the existence of separation witnesses.

\smallskip
\noindent
\emph{Covering radius.}
Fix any vertex $u\in V$ and let $v\in \ball_{\Rdist_G}(u,h_{\cov})$.
Then
\[
\dist_{G'}(u,v)\le \alpha h_{\cov}.
\]
By the hopset guarantee, there exists a path in $G'\cup H$ from $u$ to $v$ with at most $h$
edges and total length at most $1.1\alpha h_{\cov}$.
After scaling and rounding, the length of this path in $\widehat{G}$ is at most
\[
\frac{1.1\alpha h_{\cov}}{\lambda}+h
=1.1h+h
<10h.
\]
Hence
\[
\ball_{\Rdist_G}(u,h_{\cov})\subseteq \ball_{\widehat{G}}(u,10h).
\]
Since $\widehat{\cN}$ is a neighborhood cover in $\widehat{G}$ with covering radius $10h$,
there exists a cluster in some $\widehat{\cS}_i$ containing
$\ball_{\widehat{G}}(u,10h)$, and therefore also containing
$\ball_{\Rdist_G}(u,h_{\cov})$.
Thus the output has covering radius $h_{\cov}$ with respect to $\Rdist_G$.

\smallskip
\noindent
\emph{Diameter bound.}
Fix any clustering $\widehat{\cS}_i$ and any two vertices $x,y$ in the same cluster of
$\widehat{\cS}_i$.
Since that clustering has diameter $\widehat h_i$ in $\widehat{G}$,
\[
\dist_{\widehat{G}}(x,y)\le \widehat h_i.
\]
By construction of $\widehat{G}$, every edge $e$ satisfies
\[
w(e)\le \lambda\cdot \widehat{\ell}(e).
\]
Therefore every path of length $L$ in $\widehat{G}$ corresponds to a path in $G'\cup H$
of length at most $\lambda L$, and hence
\[
\dist_{G'}(x,y)=\dist_{G'\cup H}(x,y)\le \lambda\cdot \dist_{\widehat{G}}(x,y)
\le \lambda\widehat h_i=h_i.
\]
Using $\Rdist_G(x,y)\le \dist_{G'}(x,y)$, we get
\[
\Rdist_G(x,y)\le h_i.
\]
Thus the $i$th output clustering has diameter $h_i$.

Moreover,
\[
h_i\le \lambda\widehat h_{\diam}
=\frac{\alpha h_{\cov}}{h}\cdot
\frac{1}{\eps}\cdot O(\alpha s)^{O(1/\eps)}\cdot h
=
\frac{1}{\eps}\cdot O(\alpha s)^{O(1/\eps)}\cdot \alpha h_{\cov}.
\]
Since $\alpha=O(\log^2 n)$, this is
\[
h_i\le h_{\diam}
:=\frac{1}{\eps}\cdot O(s\log^2 n)^{O(1/\eps)}\cdot h_{\cov}.
\]
Hence the output has diameter bound $h_{\diam}$.

\smallskip
\noindent
\emph{Separation factor.}
Fix a clustering $\widehat{\cS}_i$, and let $S,T$ be two distinct clusters in it.
We show that
\[
\Rdist_G(S,T)\ge s\,h_i.
\]

Suppose for contradiction that there exist $x\in S$ and $y\in T$ such that
\[
\Rdist_G(x,y)< s\,h_i=s\lambda\widehat h_i.
\]
Then
\[
\dist_{G'}(x,y)\le \alpha\Rdist_G(x,y)< \alpha s\lambda\widehat h_i.
\]
By the hopset guarantee, there exists a path in $G'\cup H$ from $x$ to $y$ using at most $h$
edges and having total length at most
\[
1.1\alpha s\lambda\widehat h_i.
\]
After scaling and rounding, this yields
\[
\dist_{\widehat{G}}(x,y)\le 1.1\alpha s\,\widehat h_i+h.
\]
Since $\widehat h_i\ge 10h$, we have
\[
1.1\alpha s\,\widehat h_i+h
\le \left(1.1+\frac{1}{10}\right)\alpha s\,\widehat h_i
<40\alpha s\,\widehat h_i,
\]
contradicting the fact that distinct clusters in $\widehat{\cS}_i$ are at distance at least
\[
40\alpha s\,\widehat h_i
\]
in $\widehat{G}$.
Therefore
\[
\Rdist_G(S,T)\ge s\,h_i,
\]
and the output clustering has separation factor $s$.

\smallskip
\noindent
\emph{Separation witnesses.}
Fix a clustering $\widehat{\cS}_i$, and recall that for each cluster $S\in\widehat{\cS}_i$
we defined
\[
S':=\ball_{\widehat{G}}(S,10\alpha s\,\widehat h_i).
\]
We first show that these witness sets are pairwise disjoint.
Indeed, if $S\neq T$ and some vertex belonged to both $S'$ and $T'$, then by the triangle
inequality in $\widehat{G}$ we would have
\[
\dist_{\widehat{G}}(S,T)\le 20\alpha s\,\widehat h_i,
\]
contradicting the separation guarantee
\[
\dist_{\widehat{G}}(S,T)\ge 40\alpha s\,\widehat h_i.
\]

Now let $x$ be any vertex with
\[
\Rdist_G(x,S)\le \frac{s\,h_i}{2}
= \frac{s\lambda\widehat h_i}{2}.
\]
Then for some $y\in S$,
\[
\Rdist_G(x,y)\le \frac{s\lambda\widehat h_i}{2},
\qquad\text{hence}\qquad
\dist_{G'}(x,y)\le \frac{\alpha s\lambda\widehat h_i}{2}.
\]
Applying the hopset guarantee once more, there exists a path from $x$ to $y$ in $G'\cup H$
using at most $h$ edges and having total length at most
\[
1.1\cdot \frac{\alpha s\lambda\widehat h_i}{2}.
\]
After scaling and rounding, this gives
\[
\dist_{\widehat{G}}(x,S)\le 0.55\alpha s\,\widehat h_i+h.
\]
Since $\widehat h_i\ge 10h$, we obtain
\[
0.55\alpha s\,\widehat h_i+h<10\alpha s\,\widehat h_i,
\]
and therefore $x\in S'$.
Thus
\[
\ball_{\Rdist_G}\!\left(S,\frac{s\,h_i}{2}\right)\subseteq S'.
\]
So the sets $\{S':S\in\widehat{\cS}_i\}$ form a valid separation witness for the
$i$th clustering.

This proves that the output is a neighborhood cover with respect to $\Rdist_G$
with the claimed parameters, together with separation witnesses.

\paragraph{Running time analysis.}
The graph $\widehat{G}$ has
\[
|E(\widehat{G})|=|E(G')|+|H|=m+n^{1+o(1)}=\hO{m}
\]
edges.

The work of the call to \Cref{lem:undirfatLDD} on $\widehat{G}$ is
\[
|E(\widehat{G})|\cdot \widehat h_{\diam}\cdot n^{O(\eps)}.
\]
Since
\[
\widehat h_{\diam}
=
\frac{h}{\alpha h_{\cov}}\cdot h_{\diam},
\]
the assumption that both $h_{\diam}/h_{\cov}$ and $\omega=n^{O(\eps)}\log n$
are of order $n^{o(1)}$ implies that
\[
\widehat h_{\diam}=n^{o(1)}
\qquad\text{and}\qquad
n^{O(\eps)}=n^{o(1)}.
\]
Therefore the work of the FatLDD step is
\[
\hO{m}\cdot n^{o(1)}\cdot n^{o(1)}=\hO{m}.
\]
The preprocessing steps \Cref{lem:apxroundtrip} and the hopset construction also take
$\hO{m}$ time. Hence the total running time is $\hO{m}$.
\end{proof}

\subsection{The Cutmatch Algorithm}\label{sec:cutmatch}

In this section we will prove \Cref{lem:cutmatch}. 

We first prove two lemmas about using DAG projection in the way of length-bounded flow.

At first sight, it is not obvious how to use the dual part of the projection algorithm together with an oracle for $\OPDMCF{\kappa,\lambda}$ on DAGs. The following lemma is a convenient bridge.

\begin{lemma}\label{lem:MCFtoLCflow}
    There is an algorithm that takes as input a directed graph $G$ with positive polynomially bounded integer edge capacities and lengths, and an $s$-$t$ feasible demand $D$ on $G$, and in time proportional to one call to $\OPDMCF{\kappa,\lambda}$ on $G$, outputs
    \begin{itemize}
        \item a flow $f$ routing $D$ with congestion at most $\kappa$ and average cost $h=\cost(f)/\val(f)$, and
        \item a $h/\lambda$-length moving cut $C$ with size at most $\val(f)$ such that $\dist_{G-C}(s,t)\ge h/\lambda$.
    \end{itemize}
\end{lemma}

\begin{proof}
Run $\OPDMCF{\kappa,\lambda}$ on $(G,D)$ and let $(f,\phi,\tl)$ be the output. By definition, $f$ routes $D$ with congestion at most $\kappa$; set $h=\cost(f)/\val(f)$.

Let $S=\sum_{e\in E}U(e)\tl(e)$, and let $\widetilde G$ be the graph obtained from $G$ by replacing each edge length $\ell(e)$ by $\ell(e)+\tl(e)$. Since $(f,\phi,\tl)$ is a $(\kappa,\lambda)$-\PDMCF{}, we have
\[
\cost(f)\le \lambda\bigl(\val(f)\cdot(\phi(t)-\phi(s))-S\bigr),
\]
so
\[
\frac{h}{\lambda}\le (\phi(t)-\phi(s))-\frac{S}{\val(f)}.
\]
By dual feasibility, $\phi(v)-\phi(u)\le \ell(e)+\tl(e)$ for every edge $e=(u,v)$, hence $\phi(t)-\phi(s)\le \dist_{\widetilde G}(s,t)$. Therefore
\[
\frac{h}{\lambda}\le \dist_{\widetilde G}(s,t)-\frac{S}{\val(f)}.
\]

Let $L:=\dist_{\widetilde G}(s,t)$ and define $C(e):=\tl(e)/L$ for every edge $e$. Viewed as an $L$-length moving cut, $G-C$ is exactly $\widetilde G$, so $\dist_{G-C}(s,t)=L$. Also,
\[
|C|=\sum_{e\in E}U(e)C(e)=\frac{S}{L}\le \val(f),
\]
where the last inequality follows from $h/\lambda \le L-S/\val(f)$.

It remains to show that the same cut values also form a $(h/\lambda)$-length moving cut with $\dist_{G-C}(s,t)\ge h/\lambda$. Since all edge lengths are positive and $h/\lambda\le L$, for every $s$-$t$ path $P$ we have
\[
\ell_G(P)+\frac{h/\lambda}{L}\tl(P)
\ge
\frac{h/\lambda}{L}\bigl(\ell_G(P)+\tl(P)\bigr).
\]
The right-hand side is at least $(h/\lambda)\cdot \dist_{\widetilde G}(s,t)/L = h/\lambda$. Hence every $s$-$t$ path in the graph with length function $\ell_G+(h/\lambda)\cdot C$ has length at least $h/\lambda$, i.e. $\dist_{G-C}(s,t)\ge h/\lambda$.
\end{proof}

Now we can combine this lemma with the projection algorithm.

\begin{corollary}\label{cor:PDMCFtoDAG}
    There is an algorithm that takes as input a directed graph $G$ with positive polynomially bounded integer edge capacities and lengths, an $s$-$t$ feasible demand $D$ on $G$, and a graph $G'$ associated with a $h$-length $(\kappa',\lambda')$-projection algorithm, and in time proportional to one call to $\OPDMCF{\kappa,\lambda}$ on a graph of size $O(|G'|)$ and one call to the projection algorithm, outputs
    \begin{itemize}
        \item a flow $f$ routing $D$ with congestion at most $\kappa\kappa'$ and average cost $a=\cost(f)/\val(f)$, and
        \item a $\min(h,a/(\lambda\lambda'))$-length moving cut $C$ with size at most $\val(f)$ such that
        \[
        \dist_{G-C}(s,t)\ge \min\!\left\{h,\frac{a}{\lambda\lambda'}\right\}.
        \]
    \end{itemize}
\end{corollary}

\begin{proof}
Let $D':=\flpi(D)$. Since $D$ is an $s$-$t$ demand, $D'$ is a $\first{s}$-$\last{t}$ demand on $G'$ with the same value.

Apply \Cref{lem:MCFtoLCflow} to $(G',D')$ using the oracle $\OPDMCF{\kappa,\lambda}$. We obtain a flow $\hat f$ routing $D'$ with congestion at most $\kappa$ and average cost $a=\cost(\hat f)/\val(\hat f)$, together with an $(a/\lambda)$-length moving cut $C'$ on $G'$ such that $|C'|\le \val(\hat f)$ and
\[
\dist_{G'-C'}(\first{s},\last{t})\ge \frac{a}{\lambda}.
\]

We first project the flow. Since $\Dem(\hat f)=D'=\flpi(D)$, the primal part of the projection algorithm returns a flow $f$ on $G$ routing $D$, with the same cost as $\hat f$ and congestion at most $\kappa'\kappa$. As $\val(f)=\val(\hat f)$, its average cost is also $a$.

For the cut, let $L':=\min\{a/\lambda,\;h\lambda'\}$. Since all edge lengths in $G'$ are positive and $L'\le a/\lambda$, the same cut values $C'$ may also be viewed as an $L'$-length moving cut, and the same scaling argument as in \Cref{lem:MCFtoLCflow} gives
\[
\dist_{G'-C'}(\first{s},\last{t})\ge L'.
\]

Now apply the dual part of the projection algorithm to this $L'$-length moving cut. We obtain a moving cut $C$ on $G$ with length $L'/\lambda'=\min\{h,a/(\lambda\lambda')\}$ and size $|C|=|C'|\le \val(\hat f)=\val(f)$.

It remains to prove the distance lower bound. Suppose for contradiction that
\[
\dist_{G-C}(s,t)< \frac{L'}{\lambda'}.
\]
Since $L'/\lambda' \le h$, the guarantee of the $h$-length projection algorithm applies to the pair $(s,t)$, giving
\[
\dist_{G'-C'}(\first{s},\last{t})
\le \lambda' \cdot \dist_{G-C}(s,t)
< L',
\]
contradicting $\dist_{G'-C'}(\first{s},\last{t})\ge L'$. Therefore
\[
\dist_{G-C}(s,t)\ge \frac{L'}{\lambda'}
=
\min\!\left\{h,\frac{a}{\lambda\lambda'}\right\}.
\]
This proves the corollary.
\end{proof}

We also need the following helpful lemma.

\begin{lemma}\label{lem:cutmatchoneiteration}
    Given a directed graph $G$ with positive integral edge capacities and lengths, an $h$-length single-commodity demand $D$, and parameters $\phi\in(0,0.1)$ and $s>1$, there is an algorithm making one call to $\OProj{hs,\kappa',\lambda'}$ and one call to $\OPDMCF{\kappa,\lambda}$, each on a graph of size $O(|G|)$, that either
    \begin{itemize}
        \item outputs an implicit flow-path decomposition of an $hs$-length flow routing some demand $D'\preceq D$ with congestion at most $\kappa\kappa'$ such that
        \[
        |D'|\ge \frac{\phi}{6}\cdot |D|,
        \]
        or
        \item outputs a $\frac{hs}{4\lambda\lambda'}$-length moving cut $C$ that is $(h,\frac{s}{4\lambda\lambda'})$-length $\phi$-sparse with respect to $A_D$ and has size at least
        \[
        \frac{\phi}{3}\cdot |D|.
        \]
    \end{itemize}
\end{lemma}

\begin{proof}
Let \(\Delta:=|D|\). Construct an auxiliary graph \(H\) from \(G\) by adding a super-source \(s\) and a super-sink \(t\), where:
\begin{itemize}
    \item for every vertex \(v\) with \(D(v)>0\), we add an edge \((s,v)\) of capacity \(D(v)\) and length \(0\);
    \item for every vertex \(v\) with \(D(v)<0\), we add an edge \((v,t)\) of capacity \(-D(v)\) and length \(0\).
\end{itemize}
Clearly, \(H\) has size \(O(|G|)\).

Now consider the \(s\)-\(t\) demand of value \(\phi\Delta/3\) in \(H\), and apply \Cref{cor:PDMCFtoDAG}. This returns
\begin{itemize}
    \item a flow \(f\) routing value \(\phi\Delta/3\) with congestion at most \(\kappa\kappa'\) and average cost
    \[
    a:=\frac{\cost(f)}{\val(f)},
    \]
    and
    \item a moving cut \(\widehat C\) of length
    \[
    h':=\min\!\left\{hs,\frac{a}{\lambda\lambda'}\right\}
    \]
    and size at most
    \[
    |\widehat C|\le \val(f)=\frac{\phi\Delta}{3},
    \]
    such that
    \[
    \dist_{H-\widehat C}(s,t)\ge h'.
    \]
\end{itemize}

We distinguish two cases.

\medskip
\noindent
\textbf{Flow case: \(a\le hs/2\).}
Apply \Cref{thm:flowpathdecomposition} to \(f\), obtaining a decomposition
\[
\{(L_i,\lambda_i,s_i,t_i)\}_{i\in[k]}.
\]
Keep only those triples with \(L_i\le hs\), and output the resulting collection
\[
\{(\lambda_i,s_i,t_i)\}_{i\in[k']}.
\]
This is an implicit flow-path decomposition of an \(hs\)-length flow.

It remains to lower bound its value. If the retained flow had value less than \(\val(f)/2\), then more than \(\val(f)/2\) units of flow would be routed on paths of length strictly larger than \(hs\). This would imply
\[
\cost(f)> \frac{\val(f)}{2}\cdot hs,
\]
contradicting the assumption \(a=\cost(f)/\val(f)\le hs/2\). Therefore, the retained flow has value at least
\[
\frac{1}{2}\val(f)=\frac{\phi\Delta}{6}.
\]
Hence it routes some demand \(D'\preceq D\) of value at least \(\phi|D|/6\), as desired.

\medskip
\noindent
\textbf{Cut case: \(a> hs/2\).}
In this case,
\[
h'=\min\!\left\{hs,\frac{a}{\lambda\lambda'}\right\}\ge \frac{hs}{2\lambda\lambda'}.
\]

Let \(C^\star\) be the restriction of \(\widehat C\) to the edges of \(G\), multiplied by \(2\). We interpret \(C^\star\) as a \(\frac{hs}{4\lambda\lambda'}\)-length moving cut on \(G\). Thus, for every \(e\in E(G)\),
\[
\ell_{C^\star}(e)=\frac{hs}{4\lambda\lambda'}\cdot C^\star(e)
=\frac{hs}{2\lambda\lambda'}\cdot \widehat C(e).
\]

Call a vertex \(u\in V(G)\) \emph{heavy} if the auxiliary edge incident to \(u\) in \(H\) receives cut value more than \(1/4\) under \(\widehat C\), and \emph{light} otherwise. Since
\[
|\widehat C|\le \frac{\phi\Delta}{3},
\]
the total demand on heavy vertices is at most
\[
4|\widehat C|\le \frac{4\phi\Delta}{3}<\frac{\Delta}{6},
\]
where the last inequality uses \(\phi<0.1\). Hence the total demand on light vertices is at least
\[
\frac{5\Delta}{6}.
\]

Now let \(u,v\in V(G)\) be light vertices with \(D(u)>0\) and \(D(v)<0\). Since the auxiliary edges \((s,u)\) and \((v,t)\) each contribute less than \(h'/4\) in \(H-\widehat C\), we must have
\[
\dist_{H-\widehat C}(u,v)\ge \dist_{H-\widehat C}(s,t)-\frac{h'}{2}\ge \frac{h'}{2}.
\]

We claim that
\[
\dist_{G-C^\star}(u,v)\ge \frac{hs}{4\lambda\lambda'}.
\]
Indeed, for any \(u\)-to-\(v\) path \(P\) in \(G\),
\[
\ell_{G-C^\star}(P)
=
\ell_G(P)+\frac{hs}{2\lambda\lambda'}\widehat C(P).
\]
Since \(\frac{hs}{2\lambda\lambda'}\le h'\) and \(\ell_G(P)\ge 0\), we have
\[
\ell_G(P)+\frac{hs}{2\lambda\lambda'}\widehat C(P)
\ge
\frac{hs}{2\lambda\lambda' h'}\bigl(\ell_G(P)+h'\widehat C(P)\bigr).
\]
Taking the minimum over all such paths \(P\), we obtain
\[
\dist_{G-C^\star}(u,v)
\ge
\frac{hs}{2\lambda\lambda' h'}\cdot \dist_{H-\widehat C}(u,v)
\ge
\frac{hs}{2\lambda\lambda' h'}\cdot \frac{h'}{2}
=
\frac{hs}{4\lambda\lambda'}.
\]

Therefore every light positive vertex and every light negative vertex are \(\frac{hs}{4\lambda\lambda'}\)-separated by \(C^\star\). Since the total demand on light vertices is at least \(5\Delta/6\), it follows that \(C^\star\) is \((h,\frac{s}{4\lambda\lambda'})\)-length \(\phi\)-sparse with respect to \(A_D\) after padding if necessary.

Finally, note that
\[
|C^\star|=2|\widehat C|\le \frac{2\phi\Delta}{3}.
\]
If \(|C^\star|<\phi\Delta/3\), we pad \(C^\star\) until its size becomes exactly \(\phi\Delta/3\). Since padding only increases edge lengths, the separation property is preserved. Moreover, after padding the cut still has size at most \(2\phi\Delta/3\), while it separates at least \(5\Delta/6\) units of demand, so its sparsity remains at most
\[
\frac{(2\phi\Delta/3)}{(5\Delta/6)}=\frac{4}{5}\phi<\phi.
\]
Thus the resulting cut is still \((h,\frac{s}{4\lambda\lambda'})\)-length \(\phi\)-sparse with respect to \(A_D\), and has size at least \(\phi|D|/3\).
\end{proof}

Now we are ready to prove \Cref{lem:cutmatch}.

\cutmatch*
\begin{proof}
Let \(D^{(0)}:=D\), and let \(\Delta^{(0)}:=|D|\). We repeatedly apply \Cref{lem:cutmatchoneiteration} to the current residual demand.

At iteration \(r\), suppose the current residual demand is \(D^{(r)}\). If
\[
|D^{(r)}|<\eps |D|,
\]
then we stop and output the sum of all previously returned flows. Otherwise, apply \Cref{lem:cutmatchoneiteration} to \(D^{(r)}\).

If the lemma returns a cut, then this cut has size at least
\[
\frac{\phi}{3}\,|D^{(r)}|
\ge
\frac{\phi\eps}{3}\,|D|,
\]
and we output it.

Otherwise, the lemma returns an implicit flow-path decomposition of an \(hs\)-length flow routing some demand \(E^{(r)}\preceq D^{(r)}\) with
\[
|E^{(r)}|\ge \frac{\phi}{6}\,|D^{(r)}|.
\]
We then update the residual demand by
\[
D^{(r+1)}:=D^{(r)}-E^{(r)}.
\]
Hence
\[
|D^{(r+1)}|
\le
\left(1-\frac{\phi}{6}\right)|D^{(r)}|.
\]
By induction,
\[
|D^{(r)}|
\le
\left(1-\frac{\phi}{6}\right)^r |D|.
\]
Therefore, after
\[
T=O\!\left(\frac{\log(1/\eps)}{\phi}\right)
\]
iterations, we must have
\[
|D^{(T)}|<\eps |D|.
\]

If the algorithm terminates in this way, then the sum of all returned flows routes a demand \(D'\preceq D\) with
\[
|D'|\ge (1-\eps)|D|.
\]
Each iteration contributes congestion at most \(\kappa\kappa'\), so the total congestion is at most
\[
T\cdot \kappa\kappa'
=
O\!\left(\frac{\kappa\kappa'\log(1/\eps)}{\phi}\right).
\]
Since each constituent flow is \(hs\)-length bounded, their sum is also \(hs\)-length bounded, and concatenating all implicit flow-path decompositions gives the required output.

This proves the lemma.
\end{proof}

\newcommand{\PC}{\textsc{PathCover}}
\newcommand{\widthPC}{\mathsf{width}}
\newcommand{\muw}{\mu}
\newcommand{\mass}[1]{\muw(#1)}
\newcommand{\rhoPC}{\rho}
\newcommand{\gammaPC}{\gamma}
\newcommand{\tauPC}{\tau}
\newcommand{\RingOut}{R^{\mathrm{out}}}
\newcommand{\RingIn}{R^{\mathrm{in}}}
\newcommand{\Active}{\mathcal{F}}

\newcommand{\Ball}{\mathrm{Ball}}
\newcommand{\Bo}{\mathrm{Ball}^{\mathrm{out}}}
\newcommand{\Bi}{\mathrm{Ball}^{\mathrm{in}}}
\newcommand{\Start}{\mathrm{Start}}
\newcommand{\End}{\mathrm{End}}

\newcommand{\To}{T^{\mathrm{out}}}
\newcommand{\tTo}{\tilde{T}^{\mathrm{out}}}
\newcommand{\Ti}{T^{\mathrm{in}}}
\newcommand{\tTi}{\tilde{T}^{\mathrm{in}}}
\newcommand{\Ii}{i^{\mathrm{in}}}
\newcommand{\Io}{i^{\mathrm{out}}}
\newcommand{\Ho}{H^{\mathrm{out}}}
\newcommand{\Hi}{H^{\mathrm{in}}}
\newcommand{\Hm}{H^{\mathrm{mid}}}
\newcommand{\bHo}{\bar{H}^{\mathrm{out}}}
\newcommand{\tHo}{\tilde{H}^{\mathrm{out}}}
\newcommand{\bHi}{\bar{H}^{\mathrm{in}}}
\newcommand{\iBo}{B^{\mathrm{out}}}
\newcommand{\bBo}{\bar{B}^{\mathrm{out}}}
\newcommand{\iBi}{B^{\mathrm{in}}}
\newcommand{\bBi}{\bar{B}^{\mathrm{in}}}
\newcommand{\GC}{G'_{\mathrm{SCC}}}
\section{Path Cover with Bounded Width}\label{sec:pathcover}

In this section, we will prove \Cref{lem:pathcover}.

\newcommand{\Chain}{\mathsf{Chain}}
\newcommand{\degtot}{\deg}

\pathcover*
Throughout the section, we write $\rhoPC=2^{1/\eps}$ for convenience.

\subsection{The Algorithm}\label{subsec:pathcoveralg-width}

The algorithm is recursive.  A recursive call carries a positive integer weight on every
vertex.

\begin{definition}[Loaded recursive instance]
A \emph{loaded instance} is a pair $(A,\muw)$ where $A\subseteq V$ and
\[
\muw : A \to \{1,\rhoPC,\rhoPC^2,\ldots\}.
\]
For every $S\subseteq A$, define
\[
\mass{S}:=\sum_{v\in S}\muw(v).
\]
\end{definition}

For $S\subseteq A$, define the reweighting operator
\[
(\rhoPC\uparrow_S \muw)(v):=
\begin{cases}
\rhoPC\cdot \muw(v), & v\in S,\\
\muw(v), & v\notin S.
\end{cases}
\]

The recursive procedure is denoted by
\[
\PC(A,\muw),
\]
and the initial call is
\[
\PC(V,\mathbf{1}),
\]
where $\mathbf{1}(v)=1$ for every $v\in V$.

\paragraph{Base case.}
If $|A|=1$, return $G[A]$ with the trivial projection map and representative.

\paragraph{Recursive step.}
Fix a loaded instance $(A,\muw)$ with $|A|>1$, and choose an arbitrary root vertex
$u\in A$.

We run a forward Dijkstra search in $G[A]$ from $u$
and a backward Dijkstra search in $G[A]$ toward $u$.  
For $i\ge 0$, define
\[
B_i^{\mathrm{out}} := \Bo_{G[A]}(u,i\cdot d),
\qquad
B_i^{\mathrm{in}} := \Bi_{G[A]}(u,i\cdot d).
\]

For either direction, let $i$ be the smallest positive integer such that
\[
\mass{B_i} \le (1+\tauPC)\mass{B_{i-1}},
\]
where
\[
\tauPC := \frac{1}{1000\rhoPC^3}.
\]

Let $\Io$ and $\Ii$ be the forward and backward stopping indices.  Define
\begin{align*}
\iBo &= \Bo_{G[A]}(u,\Io\cdot d), &
\bBo &= A-\Bo_{G[A]}(u,(\Io-1)\cdot d),\\
\iBi &= \Bi_{G[A]}(u,\Ii\cdot d), &
\bBi &= A-\Bi_{G[A]}(u,(\Ii-1)\cdot d).
\end{align*}
Also define the forward and backward rings
\begin{align*}
\RingOut &:= \iBo\cap \bBo
= \Bo_{G[A]}(u,\Io\cdot d)\setminus \Bo_{G[A]}(u,(\Io-1)\cdot d),\\
\RingIn &:= \iBi\cap \bBi
= \Bi_{G[A]}(u,\Ii\cdot d)\setminus \Bi_{G[A]}(u,(\Ii-1)\cdot d).
\end{align*}

\emph{Simultaneous forward and backward ball growing:} For running time analysis, we do the ball growing forward and backward simultaneously using Dijkstra's algorithm. Specifically, we interleave Dijkstra's algorithm for the forward and backward ball growing in the following way: let $A$ contain all the nodes $v$ that have been extracted from the heap (i.e., have been relaxed and contribute $\deg(v)$ to the running time) in the forward Dijkstra's algorithm, and $B$ be the backward version, if $\mu(A)\ge \mu(B)$ then we continue the forward Dijkstra's algorithm (extract the next node and do relaxation, update $A$), otherwise we continue the backward Dijkstra's algorithm. If one of them finds a thin layer, both procedures stop.

Without loss of generality, assume the forward search stops first.
Fix
\[
\gammaPC := \frac{1}{100\rhoPC}.
\]

We distinguish two cases.

\begin{enumerate}
    \item \textbf{Case 1: $\mass{\iBo}\le (1-\gammaPC)\mass{A}$.}

    Recurse on the two overlapping pieces, and multiply the weights on the forward ring
    in both recursive children:
    \[
    \Ho \leftarrow \PC\!\left(\iBo,(\rhoPC\uparrow_{\RingOut}\muw)\big|_{\iBo}\right),
    \qquad
    \bHo \leftarrow \PC\!\left(\bBo,(\rhoPC\uparrow_{\RingOut}\muw)\big|_{\bBo}\right).
    \]
    Let $G'$ be the layered projection induced from $(\bHo,\Ho)$.
    Return $G'$.

    \item \textbf{Case 2: $\mass{\iBo}>(1-\gammaPC)\mass{A}$.}

    Continue the backward search until it reaches its own stopping index $\Ii$.

    Let $\To$ be the shortest-path tree in $G[\iBo]$ rooted at $u$, and let $\Ti$ be the
    reversed shortest-path tree in $G[\iBi]$ rooted at $u$.

    Let
    \[
    M := \iBo\cap \iBi.
    \]
    Let $\widetilde{\To}$ be the subtree of $\To$ containing all vertices that can reach $M$
    in $\To$, and let $\widetilde{\Ti}$ be the subtree of $\Ti$ containing all vertices that
    can be reached from $M$ in $\Ti$.
    Define
    \[
    \Hm := G[V(\widetilde{\To})\cup V(\widetilde{\Ti})].
    \]
    As before, $\Hm$ is a subgraph of $G$, so the projection map and representatives are
    defined trivially.

    Now recurse on the two recursive pieces, multiplying \emph{all} weights in each child by
    $\rhoPC$:
    \[
    \tHo \leftarrow \PC\!\left(\bBo\cap \iBi,(\rhoPC\muw)\big|_{\bBo\cap \iBi}\right),
    \qquad
    \bHi \leftarrow \PC\!\left(\bBi,(\rhoPC\muw)\big|_{\bBi}\right),
    \]
    where
    \[
    (\rhoPC\muw)(v):=\rhoPC\muw(v)\qquad\text{for all }v\in A.
    \]

    Let $G'$ be the layered projection induced from $(\tHo,\Hm,\bHi)$.
    Return $G'$.
\end{enumerate}

\subsection{The Global Potential}

We show that a global $\ell_2^2$ potential over the currently active recursive
calls never increases.

\begin{definition}[Active frontier and potential]
At any moment during the recursion, let $\Active$ be the set of loaded instances that have
been created but not yet expanded (i.e., all the leaves of the current recursion tree).  Define
\[
\Phi(\Active):=\sum_{(A,\muw)\in \Active} \mass{A}^2.
\]
\end{definition}

Initially $\Active=\{(V,\mathbf{1})\}$, so
\[
\Phi(\Active)=n^2.
\]

Finally, $\Active$ becomes all the leaves of the final recursion tree. Our goal is to show that in the end, we still have $\Phi(\Active)\le n^2$.

We now show that every recursive expansion never increases $\Phi$.

\begin{lemma}[Thin rings]\label{lem:thinrings}
We have
\[
\mass{\RingOut}\le \tauPC\mass{\iBo},
\qquad
\mass{\RingIn}\le \tauPC\mass{\iBi}.
\]
\end{lemma}
\begin{proof}
By definition of the stopping index,
\[
\mass{\iBo}\le (1+\tauPC)\mass{\iBo\setminus \RingOut}.
\]
Hence
\[
\mass{\RingOut}
=
\mass{\iBo}-\mass{\iBo\setminus \RingOut}
\le
\tauPC\mass{\iBo\setminus \RingOut}
\le
\tauPC\mass{\iBo}.
\]
The proof for $\RingIn$ is identical.
\end{proof}

\begin{lemma}[Case 1: the ring is tiny compared to the smaller side]\label{lem:ring-smaller-side}
Assume Case~1. Then
\[
\mass{\RingOut}\le \frac{1}{10\rhoPC^2}\cdot
\min\{\mass{\iBo},\mass{\bBo}\}.
\]
\end{lemma}
\begin{proof}
If $\mass{\iBo}\le \mass{\bBo}$, then by \Cref{lem:thinrings},
\[
\mass{\RingOut}\le \tauPC\mass{\iBo}
\le \frac{1}{10\rhoPC^2}\mass{\iBo},
\]
since $\tauPC\le 1/(10\rhoPC^2)$.

Now assume $\mass{\bBo}<\mass{\iBo}$.
Let
\[
Y:=A\setminus \iBo.
\]
Then $\bBo=Y\cup \RingOut$, and since we are in Case~1,
\[
\mass{Y}=\mass{A}-\mass{\iBo}\ge \gammaPC\mass{A}.
\]
Also, by \Cref{lem:thinrings},
\[
\mass{\RingOut}\le \tauPC\mass{\iBo}\le \tauPC\mass{A}
\le \frac{\tauPC}{\gammaPC}\mass{Y}
\le \frac{1}{10\rhoPC^2}\mass{Y}
\le \frac{1}{10\rhoPC^2}\mass{\bBo},
\]
because $\tauPC/\gammaPC = 1/(10\rhoPC^2)$.
\end{proof}

\begin{lemma}[Case 2: both complements are small]\label{lem:largemid}
Assume Case~2. Then
\[
\mass{\bBo}\le 2\gammaPC\mass{A},
\qquad
\mass{\bBi}\le 2\gammaPC\mass{A}.
\]
\end{lemma}
\begin{proof}
Since we are in Case~2,
\[
\mass{\iBo}>(1-\gammaPC)\mass{A}.
\]

Because the searches are interleaved one virtual access at a time and the forward search
stops first, by the time the forward search reaches $\Io$, the backward search has already
performed at least as many virtual accesses.  Therefore the backward search has already
settled weighted mass at least $\mass{\iBo}$.  Continuing it until its own stopping index can
only increase that mass.  Hence
\[
\mass{\iBi}\ge \mass{\iBo}>(1-\gammaPC)\mass{A}.
\]

Now
\[
\bBo=(A\setminus \iBo)\cup \RingOut,
\]
so by Case~2 and \Cref{lem:thinrings},
\[
\mass{\bBo}
\le
\gammaPC\mass{A}+\tauPC\mass{\iBo}
\le
(\gammaPC+\tauPC)\mass{A}
\le 2\gammaPC\mass{A}.
\]
The proof for $\bBi$ is identical.
\end{proof}

\begin{lemma}[The potential never increases]\label{lem:l2potential}
Every recursive expansion never increases $\Phi(\Active)$.
Consequently,
\[
\Phi(\Active)\le n^2
\qquad\text{throughout the entire recursion.}
\]
\end{lemma}
\begin{proof}
Fix an active instance $(A,\muw)$ and write
\[
s:=\mass{A}.
\]

\paragraph{Case 1.}
Write $A$ as a disjoint union
\[
A = X \,\dot\cup\, \RingOut \,\dot\cup\, Y,
\]
where
\[
X:=\Bo_{G[A]}(u,(\Io-1)\cdot d),
\qquad
Y:=A\setminus \iBo.
\]
Let
\[
x:=\mass{X},
\qquad
r:=\mass{\RingOut},
\qquad
y:=\mass{Y}.
\]
The parent mass is
\[
s=x+r+y.
\]
The two child masses are
\[
s_1=x+\rhoPC r,
\qquad
s_2=y+\rhoPC r.
\]

Assume without loss of generality that $x\le y$, i.e., $\iBo$ is the smaller side.
By \Cref{lem:ring-smaller-side},
\[
r\le \frac{x+r}{10\rhoPC^2},
\]
which implies
\[
x\ge (10\rhoPC^2-1)r.
\]

Now
\[
s_1^2+s_2^2-s^2
=
-2xy + 2(\rhoPC-1)r(x+y) + (2\rhoPC^2-1)r^2.
\]
Since $y\ge x$, this is at most
\[
-2x^2 + 4(\rhoPC-1)rx + (2\rhoPC^2-1)r^2.
\]
Also $x\ge (10\rhoPC^2-1)r\ge 4\rhoPC r$, so
\[
4(\rhoPC-1)rx \le x^2.
\]
Therefore
\[
s_1^2+s_2^2-s^2
\le
-x^2 + (2\rhoPC^2-1)r^2 < 0.
\]

\paragraph{Case 2.}
By \Cref{lem:largemid},
\[
\mass{\bBo\cap \iBi}\le \mass{\bBo}\le 2\gammaPC s,
\qquad
\mass{\bBi}\le 2\gammaPC s.
\]
After multiplying all weights in these recursive children by $\rhoPC$, their masses satisfy
\[
s_1 \le \rhoPC\cdot 2\gammaPC s,
\qquad
s_2 \le \rhoPC\cdot 2\gammaPC s.
\]
Since $\gammaPC=1/(100\rhoPC)$, we obtain
\[
s_1,s_2\le \frac{s}{50}.
\]
Hence
\[
s_1^2+s_2^2 \le 2\left(\frac{s}{50}\right)^2 < s^2.
\]

In both cases the total potential of the active frontier does not increase.
Initially the potential is $n^2$, so the claim follows.
\end{proof}

\begin{corollary}\label{cor:active-mass-at-most-n}
For every active recursive instance $(A,\muw)$,
\[
\mass{A}\le n.
\]
\end{corollary}
\begin{proof}
$\mass{A}^2$ is one term of $\Phi(\Active)\le n^2$.
\end{proof}

\subsection{Correctness: Path Covering}

The path-covering property is proved the same as in the original proof in \cite{abs-2511-08551}, because the recursive
domains are still
\[
(\bBo,\iBo)
\qquad\text{or}\qquad
(\bBo\cap \iBi,\Hm,\bBi),
\]
and the final graph is still obtained by the same layered-projection mechanism.  In this section, we only mark several changes to the proof that is affected by our definition of thin layer, which is defined by the $\mu$ mass and several parameter changes.

\begin{lemma}[Number of layers]\label{lem:numberlayers}
Both $\Io$ and $\Ii$ are at most $O(\rhoPC^3\log n)$.
\end{lemma}
\begin{proof}
For every $i<\Io$, the forward stopping rule fails, so
\[
\mass{\Bo_{G[A]}(u,i\cdot d)}
>
(1+\tauPC)\mass{\Bo_{G[A]}(u,(i-1)\cdot d)}.
\]
Multiplying over all $i<\Io$ gives
\[
\mass{\iBo} > (1+\tauPC)^{\Io-1}\mass{\{u\}}.
\]
Since $\mass{\{u\}}\ge 1$ and \Cref{cor:active-mass-at-most-n} gives $\mass{\iBo}\le n$, we get
\[
(1+\tauPC)^{\Io-1}\le n,
\]
and hence
\[
\Io = O\!\left(\frac{\log n}{\tauPC}\right)=O(\rhoPC^3\log n).
\]
The proof for $\Ii$ is identical.
\end{proof}

\begin{lemma}[Diameter bound]\label{lem:diameterbound}
Every SCC of the graph returned by $\PC(A,\muw)$ has diameter
\[
O(\rhoPC^3\log n)\cdot d.
\]
\end{lemma}
\begin{proof}
The proof is identical to the original proof in \cite{abs-2511-08551} of the diameter bound, replacing the old bound
on $\Io,\Ii$ by \Cref{lem:numberlayers}.
\end{proof}

\begin{lemma}[Path Covering]\label{lem:anchored-cover}
For every loaded instance $(A,\muw)$, the graph returned by $\PC(A,\muw)$ is a $d$-path
cover of $G[A]$.  Moreover, every path $p$ in $G[A]$ of length at most
$d$ has a lift $p'$.
\end{lemma}
\begin{proof}
This is identical to the original path-covering proof in \cite{abs-2511-08551}.
\end{proof}

\subsection{Width, Size and Running time}

We now use the potential to bound the width.

\begin{lemma}[Width bound]\label{lem:width-bound}
Let $H$ be the final graph returned by $\PC(V,\mathbf{1})$.
Then the width is at most $n^{O(\eps)}$.
\end{lemma}
\begin{proof}
Fix a vertex $v\in V$ and consider the rooted tree $T_v$ of all copies of $v$ created by the
recursion.

In Case~1, a copy of $v$ can create at most two recursive children.
In Case~2, it can create at most two recursive children, plus possibly one additional terminal
copy inside $\Hm$.

Let $L(t)$ be the maximum number of final copies obtainable from a current copy that can still
be reweighted at most $t$ more times.
Then
\[
L(0)=1,
\qquad
L(t)\le 2L(t-1)+1.
\]
Hence $L(t)\le 2^{t+1}$.

Every time a reweight happens on $v$, the weight of $v$'s copy is multiplied by
$\rhoPC$.
Thus, after $t$ such events, that active copy has weight at least $\rhoPC^t$.
On the other hand, by \Cref{cor:active-mass-at-most-n}, every active recursive instance has
total mass at most $n$, so no active copy can ever have weight larger than $n$.
Hence $\rhoPC^t \le n$ and $t=O(\log _{\rhoPC}n)$.

Given that $\rhoPC=\exp(1/\eps)$, the number of final copies of $v$ is at most $n^{O(\eps)}$.
\end{proof}

\begin{lemma}[Size bound]\label{lem:size-bound}
Let $H$ be the final graph returned by $\PC(V,\mathbf{1})$. Then $|H|=m^{1+O(\eps)}$.
\end{lemma}
\begin{proof}
Since the path cover is induced, which means $(u',v')\in E(H)$ implies $(\pi(u),\pi(v))\in E(G)$, according to \Cref{lem:width-bound}, each node can have at most $n^{O(\eps)}$ copies, so the number of copies for each edge is at most $n^{O(\eps)}$ and the lemma follows.
\end{proof}

\begin{lemma}
    The algorithm takes time $m^{1+O(\eps)}$.
\end{lemma}
\begin{proof}
    The running time of Dijkstra and the modification of the graph can be charged into $\sum_{v'\in V(H)}\deg_G(\pi(v'))$, according to \Cref{lem:width-bound}, the running time is at most $m^{1+O(\eps)}$.
\end{proof}
\subsection{Amplifying the covered length by concatenation}

We now show how to cover paths of length $qd$ for an arbitrary integer $q\ge 1$ while keeping
the SCC diameter at $O(d\cdot \rhoPC^3\log n)$.

\begin{definition}[Consecutive $q$-chain of a path cover]\label{def:qchain}
Let $H$ be a projection onto $G$.
For an integer $q\ge 1$, define $\Chain_q(H\to G)$ as follows.

Take $q$ disjoint copies
\[
H_1,H_2,\ldots,H_q
\]
of $H$, preserving the same projection map on every layer.
For every $j\in\{1,\ldots,q-1\}$ and every two vertices
\[
x\in V(H_j),\qquad y\in V(H_{j+1}),
\]
if
\[
(\pi(x),\pi(y))\in E(G)
\]
then add the edge $(x,y)$
with weight $w(\pi(x),\pi(y))$.
\end{definition}

\begin{lemma}[Amplifying the covered length]\label{lem:q-amplification}
Suppose $H$ is a $d$-path cover of $G$ with representatives, and every SCC of $H$ has diameter
at most $\lambda d$.
Then for every integer $q\ge 1$, the graph
\[
G^{(q)}:=\Chain_q(H\to G)
\]
satisfies:
\begin{enumerate}
    \item $G^{(q)}$ covers every path in $G$ of length at most $qd$,
    \item every SCC of $G^{(q)}$ has diameter at most $\lambda d$,
    \item
    \[
    \widthPC(G^{(q)}) \le q\cdot \widthPC(H),
    \]
    \item
    \[
    |E(G^{(q)})| \le q\cdot |E(H)| + (q-1)\sum_{x\in V(H)} \degtot_G(\pi(x))
    = O\!\left(q\cdot |E(H)|\right).
    \]
\end{enumerate}
\end{lemma}
\begin{proof}
Fix any path $p$ of length at most $qd$ in $G$.
Because all edge weights are non-negative, we can greedily partition $p$ into at most $q$
consecutive subpaths
\[
p=p_1\oplus p_2\oplus \cdots \oplus p_t,
\qquad t\le q,
\]
such that each $p_i$ has length at most $d$.

For every $i\in [t]$, since $H$ is a $d$-path cover, $p_i$ has a lift $p_i'$ inside layer $H_i$.
Now the last vertex of $p_i'$ projects to the last vertex of $p_i$, and the first vertex of
$p_{i+1}'$ is the representative in layer $H_{i+1}$ of the first vertex of $p_{i+1}$.
Since the last edge of $p_i$ and the first edge of $p_{i+1}$ are consecutive on the original
path $p$, the definition of $\Chain_q(H\to G)$ adds the necessary connector edge from
$\End(p_i')$ to $\Start(p_{i+1}')$.
Thus
\[
p_1'\oplus p_2'\oplus \cdots \oplus p_t'
\]
is a lift of $p$ in $G^{(q)}$.

All inter-layer edges go strictly forward from layer $j$ to layer $j+1$, so no SCC can use
two different layers.  Hence every SCC of $G^{(q)}$ is contained in a single copy of $H$ and
still has diameter at most $\lambda d$.

The width bound is immediate because every original vertex has at most $\widthPC(H)$ copies
per layer and there are $q$ layers.

For the size bound, each layer contributes $|E(H)|$ internal edges.  Between two consecutive
layers, every vertex copy $x$ adds at most $\degtot_G(\pi(x))$ outgoing connector edges,
because the head is always the unique representative in the next layer.  Summing over the
$q-1$ consecutive pairs of layers gives the stated bound.
\end{proof}

\bibliographystyle{alpha}
\bibliography{refs}
\paragraph{AI Disclosure.}
The authors used ChatGPT to assist in drafting parts of some proofs from author-provided descriptions of the main proof ideas. The proof was subsequently refined, revised, and verified by the authors. The tool materially affected \Cref{sec:boosting,sec:LCED,sec:expandertoembedding,sec:LCEDmissing,sec:pathcover}. The authors verified the correctness and originality of all content including references.
\end{document}